\documentclass[11pt]{amsart}
\usepackage{amssymb,amsmath,amsthm,amsfonts,mathrsfs}
\usepackage[hmargin=3cm,vmargin=3.5cm]{geometry}
\usepackage[dvipsnames,table,xcdraw]{xcolor}
\usepackage[colorlinks = true,
            linkcolor = Fuchsia,
            urlcolor  = ForestGreen,
            citecolor = WildStrawberry,
            anchorcolor = blue]{hyperref}
\usepackage{graphicx,psfrag}
\usepackage{amscd}
\usepackage[shellescape]{gmp}
\usepackage{stmaryrd}  
\usepackage[all,2cell]{xy} \UseAllTwocells \SilentMatrices
\usepackage{tikz-cd,pgfplots}
\usepackage[dvips]{epsfig}
\usepackage{MnSymbol}
\usepackage{comment}
\usepackage{enumitem}
\usepackage{kbordermatrix}
\usepackage[colorinlistoftodos]{todonotes}
\usepackage{cancel}

\usepackage{stackengine,scalerel}
\newcommand\dddag{%
  \sbox0{\ddag}\scalerel*{%
  \stackengine{-.6\ht0}{\ddag}{\ddag}{O}{c}{F}{F}{S}}{\ddag}%
}

\pgfplotsset{
  compat=1.16,
  myplot/.style={smooth, tension=0.5, mark=none, thick}
}

\renewcommand{\kbldelim}{(}
\renewcommand{\kbrdelim}{)}

\usepackage[modulo]{lineno}

\usetikzlibrary{arrows,automata}
\usetikzlibrary{decorations.markings}
\usetikzlibrary{decorations.pathmorphing}

\makeatletter
\def\@adminfootnotes{%
  \let\@makefnmark\relax  \let\@thefnmark\relax
  \ifx\@empty\@date\else \@footnotetext{\@setdate}\fi
  \ifx\@empty\@subjclass\else \@footnotetext{\@setsubjclass}\fi
  \ifx\@empty\@keywords\else \@footnotetext{\@setkeywords}\fi
  \ifx\@empty\thankses\else \@footnotetext{%
    \def\par{\let\par\@par}\@setthanks}%
  \fi
}
\makeatother

\newtheorem{theorem}{Theorem}[section]
\newtheorem{proposition}[theorem]{Proposition}

\newtheorem{lemma}[theorem]{Lemma}

\theoremstyle{definition}

\theoremstyle{remark}
\newtheorem{remark}[theorem]{Remark}

\numberwithin{equation}{section}

\let\oldtocsection=\tocsection
\let\oldtocsubsection=\tocsubsection
\renewcommand{\tocsection}[2]{\hspace{0em}\oldtocsection{#1}{#2}}
\renewcommand{\tocsubsection}[2]{\hspace{1em}\oldtocsubsection{#1}{#2}}

\usepackage{chngcntr}
\counterwithin{figure}{subsection}

\makeatletter
\@namedef{subjclassname@2020}{%
  \textup{2020} Mathematics Subject Classification}

\makeatother
 
\title[Diagrammatics of information]{Diagrammatics of information}

\author{Mee Seong Im}
\address{\parbox{\linewidth}{Department of Mathematics, Johns Hopkins University, Baltimore, MD 21218, USA (current)\\ 
Department of Mathematics, United States Naval Academy, Annapolis, MD 21402, USA}}
\email{\href{mailto:meeseong@jhu.edu}{meeseong@jhu.edu}}

\author{Clement Kam} 
\address{U.S. Naval Research Laboratory, Washington, DC 20375, USA}
\email{\href{mailto:khovanov@jhu.edu}{clement.k.kam.civ@us.navy.mil}}
 
\author{Caden Pici}
\address{U.S. Naval Research Laboratory, Washington, DC 20375, USA}
\email{\href{mailto:khovanov@math.columbia.edu}{caden.j.pici.civ@us.navy.mil}}

\date{\today}

\providecommand{\keywords}[1]{\textbf{\textit{Key words and phrases.}} #1}


\begin{document}

\def\Aff{\mathsf{Aff}}
\def\AND{\mathsf{AND}}
\def\concatenate{\mathsf{concatenate}}
\def\Br{\mathsf{Br}}
\def\Gal{\mathsf{Gal}}
\def\gen{\mathsf{generators}}
\def\GL{\mathsf{GL}}
\def\SL{\mathsf{SL}}
\def\init{\mathsf{in}}
\def\t{\mathsf{t}}
\def\out{\mathsf{out}}
\def\inner{\mathsf{inner}}
\def\I{\mathsf I}
\def\region{\mathsf{region}}
\def\plane{\mathsf{plane}}
\def\R{\mathbb R}
\def\Q{\mathbb Q}
\def\Z{\mathbb Z}
\def\mc{\mathcal{c}}
\def\finite{\mathsf{finite}}
\def\infinite{\mathsf{infinite}}
\def\N{\mathbb N} 
\def\C{\mathbb C}
\def\sep{\mathsf{sep}}
\def\S{\mathbb S}
\def\SS{\mathbb S} 
\def\CP{\mathbb P}
\def\Ob{\mathsf{Ob}}
\def\op{\mathsf{op}}
\def\new{\mathsf{new}}
\def\old{\mathsf{old}}
\def\OR{\mathsf{OR}}
\def\AND{\mathsf{AND}}
\def\rat{\mathsf{rat}}
\def\rec{\mathsf{rec}}
\def\tail{\mathsf{tail}}
\def\coev{\mathsf{coev}}
\def\eps{\varepsilon}
\def\ev{\mathsf{ev}}
\def\id{\mathsf{id}}
\def\s{\mathsf{s}}
\def\S{\mathsf{S}}
\def\t{\mathsf{t}}
\def\start{\textsf{starting}}
\def\Notation{\textsf{Notation}}
\def\circleft{\raisebox{-.18ex}{\scalebox{1}[2.25]{\rotatebox[origin=c]{180}{$\curvearrowright$}}}}
\renewcommand\SS{\ensuremath{\mathbb{S}}}
\newcommand{\kllS}{\kk\llangle  S \rrangle} 
\newcommand{\kllSS}[1]{\kk\llangle  #1 \rrangle}
\newcommand{\klS}{\kk\langle S\rangle}  
\newcommand{\aver}{\mathsf{av}}  
\newcommand{\ophana}{\overline{\phantom{a}}}
\newcommand{\Bool}{\mathbb{B}}
\newcommand{\dmod}{\mathsf{-mod}}
\newcommand{\lang}{\mathsf{lang}}
\newcommand{\pfmod}{\mathsf{-pfmod}}
\newcommand{\primitive}{\mathsf{irr}}
\newcommand{\Bmod}{\Bool\mathsf{-mod}}  
\newcommand{\Bmodo}[1]{\Bool_{#1}\mathsf{-mod}}  
\newcommand{\Bfmod}{\Bool\mathsf{-fmod}} 
\newcommand{\Bfpmod}{\Bool\mathsf{-fpmod}} 
\newcommand{\Bfsmod}{\Bool\mathsf{-}\underline{\mathsf{fmod}}}  
\newcommand{\undvar}{\underline{\varepsilon}} 
\newcommand{\RLang}{\mathsf{RLang}}
\newcommand{\undotimes}{\underline{\otimes}}
\newcommand{\sigmaacirc}{\Sigma^{\ast}_{\circ}} 
\newcommand{\cl}{\mathsf{cl}}
\newcommand{\PP}{\mathcal{P}} 
\newcommand{\wedgezero}{\{ \vee ,0\} } 
\newcommand{\whA}{\widehat{A}}
\newcommand{\whC}{\widehat{C}}
\newcommand{\whM}{\widehat{M}}
\newcommand{\Sigmalr}{\Sigma^{\Z}}
\newcommand{\Sigmal}{\Sigma^{-}}
\newcommand{\Sigmar}{\Sigma^{+}}
\newcommand{\Sigmaa}{\Sigma^{\ast}}
\newcommand{\SigmaZ}{\Sigma^{\Z}}  
\newcommand{\Sigmac}{\Sigma^{\circ}}

\newcommand{\alphai}{\alpha_I}  
\newcommand{\alphac}{\alpha_{\circ}}  
\newcommand{\alphap}{(\alphai,\alphac)} 
\newcommand{\alphalr}{\alpha_{\leftrightarrow}}
\newcommand{\alphaZ}{\alpha_{\Z}}
\newcommand{\mcCinfalpha}{\mcC^{\infty}_{\alpha}}
\newcommand{\mathT}{\mathsf{T}}
\newcommand{\mathF}{\mathsf{F}}
\newcommand{\mcS}{\mathcal{S}}
\newcommand{\mcN}{\mathcal{N}}
\newcommand{\wmcN}{\widetilde{\mcN}}
\newcommand{\Net}{\mathsf{Net}}

\let\oldemptyset\emptyset
\let\emptyset\varnothing

\newcommand{\undempty}{\underline{\emptyset}}
\newcommand{\undsigma}{\underline{\sigma}}
\newcommand{\undtau}{\underline{\tau}}
\def\basis{\mathsf{basis}}
\def\irr{\mathsf{irr}} 
\def\spanning{\mathsf{spanning}}
\def\elmt{\mathsf{elmt}}

\def\H{\mathsf{H}}
\def\I{\mathsf{I}}
\def\II{\mathsf{II}}
\def\l{\lbrace}
\def\r{\rbrace}
\def\o{\otimes}
\def\lra{\longrightarrow}
\def\Ext{\mathsf{Ext}}
\def\mf{\mathfrak} 
\def\mcC{\mathcal{C}}
\def\mcO{\mathcal{O}}
\def\Fr{\mathsf{Fr}}

\def\ovb{\overline{b}}
\def\tr{{\sf tr}} 
\def\str{{\sf str}} 
\def\det{{\sf det }} 
\def\tral{\tr_{\alpha}}
\def\one{\mathbf{1}}   

\def\lra{\longrightarrow}
\def\twoheadlra{\longrightarrow\hspace{-4.6mm}\longrightarrow}
\def\hooklra{\raisebox{.2ex}{$\subset$}\!\!\!\raisebox{-0.21ex}{$\longrightarrow$}}
\def\kk{\mathbf{k}}  
\def\gdim{\mathsf{gdim}}  
\def\rk{\mathsf{rk}}
\def\undep{\underline{\varepsilon}}
\def\mathM{\mathbf{M}}  

\def\CCC{\mathcal{C}} 
\def\wCCC{\widehat{\CCC}}  

\def\complement{\mathsf{comp}}
\def\Rec{\mathsf{Rec}} 

\def\Cob{\mathsf{Cob}} 
\def\Kar{\mathsf{Kar}}   

\def\dmod{\mathsf{-mod}}   
\def\pmod{\mathsf{-pmod}}    

\newcommand{\brak}[1]{\ensuremath{\left\langle #1\right\rangle}}
\newcommand{\brakspace}[1]{\ensuremath{\left\langle \:\: #1\right\rangle}}

\newcommand{\oplusop}[1]{{\mathop{\oplus}\limits_{#1}}}
\newcommand{\ang}[1]{\langle #1 \rangle } 
\newcommand{\ppartial}[1]{\frac{\partial}{\partial #1}} 
\newcommand{\checkr}{{\bf \color{red} CHECK IT}}
\newcommand{\checkb}{{\bf \color{blue} CHECK IT}}
\newcommand{\checkk}[1]{{\bf \color{red} #1}}

\newcommand{\mcA}{{\mathcal A}}
\newcommand{\cZ}{{\mathcal Z}}
\newcommand{\sq}{$\square$}
\newcommand{\bi}{\bar \imath}
\newcommand{\bj}{\bar \jmath}

\newcommand{\undn}{\underline{n}}
\newcommand{\undm}{\underline{m}}
\newcommand{\undzero}{\underline{0}}
\newcommand{\undone}{\underline{1}}
\newcommand{\undtwo}{\underline{2}}

\newcommand{\cob}{\mathsf{cob}} 
\newcommand{\comp}{\mathsf{comp}} 

\newcommand{\Aut}{\mathsf{Aut}}
\newcommand{\Hom}{\mathsf{Hom}}
\newcommand{\Idem}{\mathsf{Idem}}
\newcommand{\Ind}{\mbox{Ind}}
\newcommand{\Id}{\textsf{Id}}
\newcommand{\End}{\mathsf{End}}
\newcommand{\iHom}{\underline{\mathsf{Hom}}}
\newcommand{\Bools}{\Bool^{\mathfrak{s}}}
\newcommand{\mfs}{\mathfrak{s}}
\newcommand{\blueline}{line width = 0.45mm, blue}

\newcommand{\drawing}[1]{
\begin{center}{\psfig{figure=fig/#1}}\end{center}}

\def\endomCempt{\End_{\mcC}(\emptyset_{n-1})}

\def\MS#1{{\color{blue}[MS: #1]}}
\def\MK#1{{\color{red}[MK: #1]}}

\allowdisplaybreaks

\begin{abstract}
We introduce a diagrammatic perspective for Shannon entropy created by the first author and Mikhail Khovanov and connect it to information theory and mutual information. We also give two complete proofs that the $5$-term  dilogarithm deforms to the $4$-term infinitesimal dilogarithm.
\end{abstract}

\maketitle
\tableofcontents

%
%

\section{Introduction}
\label{section:intro}

Mathematics and theoretical physics have emerged in the rapid development of modern technology in recent decades. In particular, information theory~\cite{Khinchin57,Khinchin56,Nielsen20,Brill56,LR97,DH99} plays a pivotal role in communication, quantification, storage, and transfer of knowledge. One of the key measures in information theory is entropy, which quantifies the uncertainty in a random variable or the outcome of a random process. In geometric information theory~\cite{MKW23,Nielsen_geometric_theory_info,Niel_geom_23_1,Niel_geom_23_2,NB_geom_info}, spaces and manifolds of probability distributions are studied from a geometric and analysis point-of-view. 

Applications and interpretations of one-dimensional cobordisms have recently emerged, predominantly explored by M.S. Im and her collaborators, e.g., see~\cite{GIKK24, GIKKL23, GIK22, IK21, IK-top-automata, IK_journey23, IK24_dilogarithms_entropy, IK24_SAF, IK_22_linear, IKV23, IZ21}.
Im and M. Khovanov introduce in \cite{IK24_dilogarithms_entropy} a diagrammatical perspective of entropy and cocycles. In particular, they view entropy as certain cobordisms (see, e.g., \cite{IK_22_linear,Kh20_univ_const_two}) where information about the network only depends on the boundary. This coincides with the fact that in physics, the rich information about black holes also only depends on the geometric and topological structure of their boundaries.

In this paper, we recall diagrammatic interpretation of Shannon entropy in \cite{IK24_dilogarithms_entropy}, and reinterpret them in terms of information theory. We also provide two complete proofs that certain $5$-term dilogarithm deforms into the $4$-term infinitesimal dilogarithm.

\subsection*{Acknowledgments}
The first author would like to thank Mikhail Khovanov and the Research Scientists at the Naval Research Laboratory for extensive discussions. The first author would like to thank the Office of Naval Research (ONR) and United States Naval Research Laboratory (NRL) in Washington, DC for their support. 

\section{Background}
\label{section:background}

Let $\kk$ be a field of characteristic $0$.

\vspace{.25cm}

\subsection{Cathelineau's \texorpdfstring{$\kk$}{k}-vector space}
\label{subsection:Cathelineau_k_vs}

J.-L.~Cathelineau constructed a $\kk$-vector space $J(\kk)$ in \cite{Cath11}, with its spanning set $\langle a,b\rangle$, where $a,b\in\kk$, satisfying the relations: 
\begin{enumerate}
\item\label{item:symm} (symmetry) 
$\langle a,b\rangle  =  \langle b,a \rangle$, 
\item\label{item:scale} (scaling) 
$\langle ca,cb\rangle  =  c \langle a,b\rangle$, $c\in \kk$, 
\item 
\label{item:Cath_cocycle} (2-cocycle relation)
$\langle a,b+c\rangle + \langle b,c \rangle   = \langle a+b,c \rangle + \langle a,b \rangle$.
\end{enumerate} 
For relation~\eqref{item:Cath_cocycle},  the field $(\kk,+)$ under addition is viewed as a group. Furthermore, the symbols $\langle a,b\rangle$ are reminiscent of values of $2$-cocycles. 

First, we state a few preliminary lemmas:
\begin{lemma}
In the vector space $J(\mathbf{k})$, we have 
 $\langle a,0\rangle =\langle 0,a\rangle =0$, and  $\langle a,-a\rangle = 0$.
\end{lemma}

\begin{proof}
We will first prove $\langle a,0\rangle =0$.
Using \eqref{item:scale}, 
$\langle 0, 0 \rangle = \langle 0a, 0b \rangle = 0 \langle a,b\rangle = 0$. 
So we obtain $\langle a , 0+0\rangle + \langle 0,0 \rangle =\langle a+0,0\rangle + \langle a,0\rangle$ by setting $b=c=0$ in \eqref{item:Cath_cocycle}. This tells us that $\langle a,0\rangle =0$. Using symmetry in \eqref{item:symm}, $\langle 0,a\rangle =\langle a,0 \rangle =0$. 
Now, to prove the last equality, 
$\langle a, -a\rangle = a \langle 1,-1\rangle$. 
So $\langle 1,-1 \rangle = \langle -(-1),-1 \rangle = -\langle -1,1 \rangle$. 
We move the symbols to one side and use symmetry to obtain $0 = \langle 1,-1 \rangle + \langle -1,1\rangle = \langle 1,-1\rangle + \langle 1,-1\rangle = 2 \langle 1,-1 \rangle$. Since $\mathbf{k}$ is a field of characteristic $0$, $\langle 1,-1\rangle=0$. Thus $\langle a,-a\rangle=0$.
\end{proof}

\begin{lemma}
\label{lemma_minus_1_a}
We have $\langle 1-a, -1\rangle = \langle a,1-a \rangle$.
\end{lemma}

\begin{proof}
Applying \eqref{item:Cath_cocycle} and letting $a=a$, $b=1-a$, and $c=-1$,  we have 
$\langle a,1-a\rangle + \langle a+1-a,-1\rangle = \langle a,1-a-1\rangle + \langle 1-a,-1\rangle$. So 
$\langle a,1-a\rangle + \cancel{\langle 1,-1\rangle} = \cancel{\langle a,-a\rangle} + \langle 1-a,-1\rangle$ since $\langle a,-a \rangle=0$. We thus obtain the equality. 
\end{proof}
 
If $\mathbf{k}$ is a field of characteristic $p\not= 0$, then 
\begin{equation}
\label{eqn_char_p_Cathelineau}
\sum_{n=1}^p \langle 1, n1\rangle =0. 
\end{equation}

The vector space $J(\kk)$ is known as the space of infinitesimal dilogarithms, which is often  infinite-dimensional over $\kk$.

In the case when $\kk=\mathbb{R}$, $J(\kk)$ is isomorphic to the entropy of finite random variables.

\subsection{Vector space isomorphic to Shannon's entropy}
\label{subsection:vs_Shannon_entropy}

The $\kk$-vector space $\beta(\kk)$
has a spanning set $[a]$, $a\in \kk^{*}$ is invertible, of vectors and relations: 
\begin{enumerate}
\item $[1]=0$, 
\item
\label{item_beta_k_4_term}
$[a]-[b]+ a\biggl[ \dfrac{b}{a}\biggr] + (1-a)\biggl[ \dfrac{1-b}{1-a} \biggr]=0, \ \ a\in \kk\setminus\{0,1\}, \ b\in \kk^{*}$.
\end{enumerate}
This vector space also satisfies $[a]=[1-a]$.

\begin{lemma}
The vector space $\beta(\kk)$ satisfies
$\left[\dfrac{1}{a}\right] = -\dfrac{1}{a}[a]$.
\end{lemma} 

\begin{proof}
Let $b=1$ in \eqref{item_beta_k_4_term} to obtain 
\[
[a]- \cancel{[1]} + a\left[ \frac{1}{a}\right] + \cancel{(1-a) \left[ \frac{1-1}{1-a}\right]} = 0. 
\]
This completes the proof.
\end{proof}
If $\mathbf{k}$ is a field of characteristic $p\not=0,2$, then 
\[
\sum_{n=2}^{p-1} [n1] = 0. 
\]
 
The space $\beta(\kk)$ generalizes entropy, which we explain in  Section~\ref{section:entropy_Cath_vs}. This 4-term equation \eqref{item_beta_k_4_term} is a limit of the 5-term equation for the dilogarithm (see [\href{https://people.mpim-bonn.mpg.de/zagier/files/scanned/DilogarithmInGeometryAndNumberTh/fulltext.pdf}{Zagier}]), hence the name infinitesimal dilogarithm. For completeness, we provide a proof of this derivation in Section~\ref{section_deformation}.

\subsection{Isomorphism of \texorpdfstring{$J(\mathbf{k})$}{Jk} and \texorpdfstring{$\beta(\mathbf{k})$}{betak}}
Cathelineau proves in ~\cite{Cath88} that the two vector spaces $J(\kk)$ and $\beta(\kk)$ are isomorphic by sending $[a]\mapsto \langle a, 1-a\rangle$. 
Conversely, $\langle a,-a \rangle \mapsto [0]=0$, and $\langle a,b\rangle\mapsto (a+b) \left[ \dfrac{a}{a+b} \right]$ whenever $a+b\not= 0$.
Furthermore, $\langle a,b\rangle = \langle b,a\rangle$ if and only if $[a]=[1-a]$.

Via the isomorphism $J(\kk) \cong \beta(\kk)$, the $2$-cocycle condition in \eqref{item:Cath_cocycle} for $J(\kk)$ is equivalent to 
\[
(a+b+c) \left[ \dfrac{a}{a+b+c} \right] + 
(b+c) \left[ \dfrac{b}{b+c} \right] = 
(a+b+c) \left[ \frac{a+b}{a+b+c}  \right] + (a+b) \left[ \frac{a}{a+b} \right]. 
\]

Cathelineau shows that there is an isomorphism between the second homology and his construction  $H_2(\Aff_1(\kk),\kk_r)\simeq \beta(\kk)\simeq J(\kk)$, where 
\[ 
\Aff_1(\kk) = 
\left\{ 
\begin{pmatrix}
c & a \\ 0 & 1
\end{pmatrix} : c\in \kk^{*}, a\in \kk 
\right\}
\]
is the group of affine symmetries of a $\kk$-line and $\kk_r$ is a suitable right $\Aff_1(\kk)$-module. 
We refer to \cite{IK24_dilogarithms_entropy} for further details.

\section{Shannon entropy and and Cathelineau's vector space}
\label{section:entropy_Cath_vs}

\subsection{Shannon entropy}
\label{subsection:Shannon_entropy_inner_product}

Let $X=\{x_1,\ldots, x_n\}$, a finite set. 
Shannon entropy of a finite probability distribution $p_X$ on $X$ that associates probabilities $p_1,\ldots, p_n$, where $\displaystyle{\sum_{i=1}^n} p_i=1$, $0<p_i<1$ to its points $x_i$, respectively, is given by 
\begin{equation}
\label{eq_shannon} 
H(p_X)  =  -\sum_{i=1}^n p_i 
\log p_i. 
\end{equation}
We can also think of $H(p_X)=-E[\log p_i]=E\left[\log \frac{1}{p_i}\right]$, where $E[\log p_i]$ is the expected value of $\log p_i$.
When $n=2$, the function $H(p_X)$ becomes a function of a single probability $p := p_1$. Then the entropy is  
\begin{equation}\label{eq_four_term_relation_entropy_2}
H(p) \ = \ - p\log p - (1-p)\log (1-p). 
\end{equation}
It is natural to extend $H$ from the open interval $(0,1)$ to all real numbers by  defining $H(0) = H(1) = 0$ and 
\begin{equation}
\label{eqn:two_term_entropy}
H(p) \ := \ - p\log|p| - (1-p)\log |1-p|. 
\end{equation}
Then $H$ is defined and continuous on all of $\mathbb{R}$.
The graph of $H$ is given in Figure~\ref{fig8_002}.
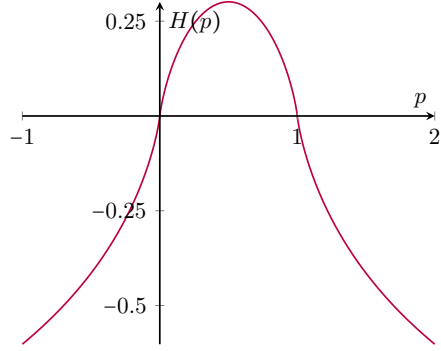
\begin{figure}
    \centering
\begin{tikzpicture}[scale=0.8]
\begin{axis}
[thick,
axis on top,
axis lines = middle,
xtick={-3,-2,-1,0,1,2,3},
ytick={-1,-0.75,-0.5,-0.25,0,0.25,0.5,0.75,1},
xlabel=\(p\),
ylabel = {\( H(p)\)},
clip=false,
every axis plot/.append style={myplot},
]
\addplot[thick,
domain = -1:2,
samples = 500,
color=purple,
]
{-x*log10(abs(x))-(1-x)*log10(abs(1-x))};
\end{axis}

\end{tikzpicture}
    \caption{Entropy function $H(p)=-p\log |p|-(1-p)\log |1-p|$ for $-1\leq p\leq 2$. }
    \label{fig8_002}
\end{figure}

\begin{lemma}
\label{lemma_entropy_basic}
We have 
$H(a) = H(1-a)$ and $H\left(\dfrac{1}{a}\right) = -\dfrac{1}{a}H(a)$.
\end{lemma}

\begin{proof}
For the first equality, we have 
\[
H(a) 
= -a \log |a| -(1-a) \log |1-a| 
= -(1-a) \log |1-a| - (1-(1-a)) \log |1-(1-a)|
= H(1-a).
\]
For the second equality, we have 
\begin{align*}
H\left(\frac{1}{a} \right) 
&= -\frac{1}{a} \log \left| \frac{1}{a}\right|
- \left(1-\frac{1}{a}\right)\log \left|1-\frac{1}{a}\right| \\ 
&= \frac{1}{a} \log |a| - \left( 1-\frac{1}{a}\right) \log \left|  \frac{a-1}{a} \right| \\ 
&= \frac{1}{a} \log |a| - \left( 1-\frac{1}{a}\right) \left(  \log |1-a| - \log |a| \right) \\ 
&= \log |a| - \left( 1-\frac{1}{a}\right)   \log |1-a|  \\ 
&= -\frac{1}{a} \left( -a \log |a| - (1-a) \log |1-a| \right) \\ 
&= -\frac{1}{a} H(a).
\end{align*}
This concludes the proof. 
\end{proof}

The entropy function satisfies a four-term functional equation, see~\cite{Kont02}, 
\begin{equation}
\label{eq_four_term_relation_entropy}
H(p) - H(q) +p\: H\left(\dfrac{q}{p}\right) + (1-p)H\left(\dfrac{1-q}{1-p}\right) = 0.
\end{equation}
Equation~\eqref{eq_four_term_relation_entropy} together with $H(1-p)=H(p)$ and continuity property uniquely determine $H$ as a nonzero function $\R\lra \R$. See~\cite{Faddeev56,Kont02}. Also see \cite[Equation (38)]{Taneja77}.

\begin{proposition}
\label{prop:4_term_relation_entropy}
Let $H:\mathbb{R}\rightarrow \mathbb{R}$, where $H$ is \eqref{eqn:two_term_entropy} satisfying conditions $H(0)=H(1)=0$. Then $H$ satisfies the $4$-term relation~\eqref{eq_four_term_relation_entropy}.
\end{proposition}

Since the proof of Proposition~\ref{prop:4_term_relation_entropy} is not explicitly stated anywhere in the literature, we will provide its proof.

\begin{proof}
We have the following set of equalities: 
\begin{equation}
\label{eqn:proving_4_term}
\begin{split}
 -H(p)&=p\log |p| + (1-p)\log |1-p|, \\ 
 -H(q)&=q\log |q| + (1-q)\log |1-q|, \\ 
 -H\left(\frac{q}{p}\right) &= \frac{q}{p}\log\left|\frac{q}{p}\right| + \left(1-\frac{q}{p}\right) \log \left| 1-\frac{q}{p}\right|  \\ 
    &= \frac{q}{p}\left(\log|q| -\log|p|\right) 
    + \frac{p-q}{p}\left(\log |p-q| -\log|p| \right), \\ 
 -H\left(\frac{1-q}{1-p}\right) 
&= \frac{1-q}{1-p} \log\left| \frac{1-q}{1-p}\right| 
+ \frac{q-p}{1-p}\log \left|\frac{q-p}{1-p}\right|  \\
&= \frac{1-q}{1-p}\left( \log |1-q|-\log|1-p| \right) +\frac{q-p}{1-p}\left( \log |q-p| - \log|1-p| \right).\\
\end{split}
\end{equation}
So 
\begin{align*}
-p H&\left(\dfrac{q}{p}\right) - (1-p)H\left(\dfrac{1-q}{1-p}\right)
=  q\left(\log|q| -\log|p|\right) 
    + (p-q)\left(\log |p-q| -\log|p| \right)\\
&\hspace{4mm}+ 
 (1-q) \left( \log |1-q|-\log|1-p| \right) + (q-p)\left( \log |q-p| - \log|1-p| \right) \\ 
 &= q\log |q| - q\log|p| + \cancel{(p-q)\log|p-q|} - (p-q)\log|p| \\
&\hspace{4mm}+ (1-q)\log|1-q| - (1-q)\log|1-p| + \cancel{(q-p)\log|q-p|} - (q-p)\log|1-p| \\ 
 &= q\log |q| + (1-q)\log|1-q| - \cancel{q\log|p|} + \cancel{q \log |p|} - p\log|p|\\
&\hspace{4mm} - \log|1-p| + \cancel{q \log|1-p|}  - \cancel{q\log|1-p|} + p \log |1-p|\\ 
&= q\log |q| + (1-q)\log|1-q| 
  - (p\log|p| + (1-p) \log|1-p|) \\ 
&= -H(q) + H(p).
\end{align*}
Therefore, 
$H(p) - H(q) +p H\left(\dfrac{q}{p}\right) + (1-p)H\left(\dfrac{1-q}{1-p}\right) = 0$ holds. 
\end{proof}

We end this section by mentioning an important relationship relating Cathelineau's vector space~\cite{Cath88,Cath11} and Kontsevich's observation~\cite{Kont02}, which is
\begin{equation}
\label{eqn_diagrammatic_entropy}
\langle p_1, p_2 \rangle 
= 
\begin{cases}
(p_1 + p_2) H\left( \dfrac{p_1}{p_1 + p_2} \right) \quad  
 &\mbox{ if } p_1 + p_2 \not= 0, \\
\qquad\qquad 0 &\mbox{ if } p_1 + p_2 = 0, 
\end{cases}
\end{equation}
relating Cathelineau's vector space and entropy.

\begin{lemma}
\label{lemma_Cath_properties}
Equation~\eqref{eqn_diagrammatic_entropy} satisfies symmetry, scaling, and the $2$-cocycle relation. 
\end{lemma}

\begin{proof}
We will first prove symmetry: 
\begin{align*}
\langle p_1,p_2 \rangle 
&= (p_1 + p_2) H\left( \frac{p_1}{p_1 + p_2} \right) \\
&= (p_1 + p_2) \left( - \frac{p_1}{p_1 + p_2} \log \left|\frac{p_1}{p_1 + p_2}\right| - \left( 1- \frac{p_1}{p_1 + p_2}\right)\log \left| 1 - \frac{p_1}{p_1 + p_2} \right| \right) \\
&= (p_1 + p_2) \left( - \frac{p_1}{p_1 + p_2} \log \left| \frac{p_1}{p_1 + p_2}\right| - \frac{p_2}{p_1 + p_2} \log \left| \frac{p_2}{p_1 + p_2} \right| \right) \\
&= (p_1 + p_2) \left(- \frac{p_2}{p_1 + p_2} \log \left| \frac{p_2}{p_1 + p_2} \right| - \frac{p_1}{p_1 + p_2} \log \left| \frac{p_1}{p_1 + p_2}\right| \right) \\
&= (p_1 + p_2) \left(- \frac{p_2}{p_1 + p_2} \log \left| \frac{p_2}{p_1 + p_2} \right| - \left( 1- \frac{p_2}{p_1 + p_2}\right)\log \left| 1 - \frac{p_2}{p_1 + p_2} \right| \right) \\
&= (p_1 + p_2) H\left( \frac{p_2}{p_1+p_2} \right) \\ 
&= \langle p_2, p_1\rangle.
\end{align*}
Next, we will prove scaling:
\begin{align*}
\langle c \: p_1, c \: p_2\rangle 
&= (c\: p_1 + c \: p_2) H\left( \frac{c\: p_1}{c\: p_1 + c\: p_2} \right) \\ 
&= (c\: p_1 + c \: p_2)
\left(-\frac{c\: p_1}{c\: p_1 + c\: p_2} \log \left| \frac{c\: p_1}{c\: p_1 + c\: p_2}\right| - \left( 1- \frac{c\: p_1}{c\: p_1 + c\: p_2}\right)\log \left| 1- \frac{c\: p_1}{c\: p_1 + c\: p_2} \right| \right) \\ 
&= c(p_1 + p_2)
\left(-\frac{\cancel{c}\: p_1}{\cancel{c}\: p_1 + \cancel{c}\: p_2} 
\log \left| \frac{ \cancel{c}\: p_1}{\cancel{c}\: p_1 + \cancel{c}\: p_2}\right| 
-  \frac{\cancel{c}\: p_2}{\cancel{c}\: p_1 + \cancel{c}\: p_2} 
\log \left| \frac{\cancel{c}\: p_2}{\cancel{c}\: p_1 + \cancel{c}\: p_2} \right| \right) \\
&= c(p_1 + p_2)
\left(-\frac{p_1}{p_1 + p_2} 
\log \left| \frac{p_1}{p_1 + p_2}\right| 
-  \frac{p_2}{p_1 +p_2} 
\log \left| \frac{p_2}{p_1 + p_2} \right| \right) \\
&= c(p_1 + p_2) H\left( \frac{p_1}{p_1 + p_2} \right) \\ 
&= c \langle p_1, p_2\rangle. 
\end{align*}
Finally, we will prove the $2$-cocycle condition: 
\begin{align*}
\langle p_1, p_2 \rangle &= (p_1 + p_2) H\left( \frac{p_1}{p_1 + p_2} \right)  \\ 
&= \cancel{(p_1 + p_2)} \left( 
-\frac{p_1}{\cancel{p_1 + p_2}} \log 
\left| \frac{p_1}{p_1 + p_2} \right| 
- \frac{p_2}{\cancel{p_1 + p_2}} \log 
\left| \frac{p_2}{p_1 + p_2} \right| \right) \\ 
&= - p_1 \log \left| \frac{p_1}{p_1 + p_2}\right| 
- p_2 \log \left| \frac{p_2}{p_1 + p_2}\right| \\ 
&=  - p_1 \log |p_1| - p_2 \log | p_2| 
+ (p_1 + p_2) \log | p_1 + p_2|
\end{align*}
while 
\begin{align*}
\langle p_1 + p_2, p_3 \rangle 
&= (p_1 + p_2 + p_3) H\left( \frac{p_1 + p_2}{p_1 + p_2 + p_3}\right) \\ 
&= H(p_1 + p_2) \qquad \qquad \qquad \qquad \mbox{ since } p_1 + p_2 + p_3=1 \\ 
&= -(p_1 + p_2) \log |p_1 + p_2| -(1 - (p_1 + p_2))\log |1-(p_1 + p_2)| \\
&= -(p_1 + p_2) \log |p_1 + p_2| - p_3 \log |p_3|. 
\end{align*}
So the sum is: 
\begin{align*}
\langle p_1, p_2 \rangle + \langle p_1 + p_2, p_3 \rangle &= - p_1 \log |p_1| - p_2 \log | p_2| 
+ \cancel{(p_1 + p_2) \log | p_1 + p_2|} \\
&\hspace{4mm} - \cancel{(p_1 + p_2) \log |p_1 + p_2|} - p_3 \log |p_3| \\
&=- p_1 \log |p_1| - p_2 \log | p_2|- p_3 \log |p_3| \\
&=H(p_X).
\end{align*}
On the other hand, 
\begin{align*}
\langle p_2 , p_3 \rangle 
&= (p_2 + p_3) H\left( \frac{p_2}{p_2 + p_3} \right)  \\ 
&= \cancel{(p_2 + p_3)} \left( 
-\frac{p_2}{\cancel{p_2 + p_3}} \log 
\left| \frac{p_2}{p_2 + p_3} \right| 
- \frac{p_3}{\cancel{p_2 + p_3}} \log 
\left| \frac{p_3}{p_2 + p_3} \right| \right) \\ 
&= - p_2 \log \left| \frac{p_2}{p_2 + p_3}\right| 
- p_3 \log \left| \frac{p_3}{p_2 + p_3}\right| \\ 
&=  - p_2 \log |p_2| - p_3 \log | p_3| 
+ (p_2 + p_3) \log | p_2 + p_3| 
\end{align*}
while 
\begin{align*}
\langle p_1, p_2+ p_3 \rangle 
&= (p_1 + p_2 + p_3) H\left( \frac{p_1}{p_1 + p_2 + p_3}\right) \\ 
&= H(p_1)  \qquad \qquad \qquad \qquad \qquad \mbox{ since } p_1 + p_2 + p_3=1\\ 
&= -p_1 \log |p_1| - (1 - p_1)\log |1 - p_1| \\
&= -p_1 \log |p_1| - (p_2 + p_3) \log |p_2 + p_3|.
\end{align*}
So the sum is: 
\begin{align*}
\langle p_2 , p_3 \rangle + \langle p_1, p_2+ p_3 \rangle &= - p_2 \log |p_2| - p_3 \log | p_3| + \cancel{(p_2 + p_3) \log | p_2 + p_3|} \\ 
&\hspace{4mm} -p_1 \log |p_1| - \cancel{(p_2 + p_3) \log |p_2 + p_3|} \\
&=  -p_1 \log |p_1| - p_2 \log |p_2| - p_3 \log | p_3| \\ 
&= H(p_X). 
\end{align*}
This concludes the proof.
\end{proof}

\begin{lemma}
If $p_1 + p_2 = 1$, we have 
\begin{equation}
H\left( p_1 \right)
= H\left( p_2 \right).
\end{equation}
\end{lemma}

\begin{proof}
This follows from Lemma~\ref{lemma_Cath_properties}. 
\end{proof}

Let $\rho: H(\kk) \rightarrow J(\mathbf{k})$, where $\rho(H(a)) = \langle a,1-a \rangle$. 

\begin{theorem}
\label{thm_isom_Cath_entropy}
The map $\rho$ is an isomorphism of vector spaces.
\end{theorem}

\begin{proof}
We have
\[
\rho(H(a)) = \langle a,1-a \rangle 
= \langle 1-a, a \rangle 
= \langle 1-a, 1-(1-a) \rangle 
= \rho(H(1-a))
\]
and 
\begin{align*}
\rho\left(H\left(\frac{1}{a}\right)\right) 
&= \left\langle \frac{1}{a},1-\frac{1}{a} \right\rangle 
= -\frac{1}{a} \langle -1, -a + 1\rangle
= -\frac{1}{a} \langle -1, 1 - a\rangle \\
&= -\frac{1}{a} \langle 1-a, -1 \rangle 
\stackrel{\dagger}{=} -\frac{1}{a} \langle a, 1-a \rangle  
=  -\frac{1}{a}\rho(H(a)),
\end{align*}
where $\dagger$ holds by Lemma~\ref{lemma_minus_1_a}.

Furthermore, equation \eqref{eq_four_term_relation_entropy} under $\rho$ becomes 
\begin{align*}
\rho(H(0)) 
&= \rho(H(a)) -\rho(H(b)) + a \rho\left(H\left(\frac{b}{a}\right)\right) + (1-a)\rho\left(H\left(\frac{1-b}{1-a}\right) \right) \\ 
&= \langle a, 1-a\rangle  - \langle b, 1-b \rangle + a \left\langle \frac{b}{a}, 1-\frac{b}{a} \right\rangle + (1-a) \left\langle \frac{1-b}{1-a}, 1-\frac{1-b}{1-a} \right\rangle \\ 
&= \langle a,1-a \rangle - \langle b,1-b\rangle + \langle b,a-b \rangle 
+ \langle 1-b, 1-a - (1-b)\rangle \\ 
&\stackrel{\ddagger}{=} \langle a,1-a \rangle - \langle b,1-b\rangle - \langle a,b-a \rangle 
+ \langle 1-b, b-a \rangle, 
\end{align*}
where $\ddagger$ holds due to
\[
\langle b,a-b \rangle 
= b \left\langle 1, \frac{a}{b} - 1 \right\rangle
= -b \left\langle 1-\frac{a}{b},-1 \right\rangle 
= -b \left\langle \frac{a}{b},1-\frac{a}{b} \right\rangle 
= -\langle a, b-a \rangle
\] via Lemma~\ref{lemma_minus_1_a}. Now, replacing $a$ with $a$, $b$ with $b-a$, and $c=1-b$ in the $2$-cocycle relation~\eqref{item:Cath_cocycle}, we see that 
\begin{align*}
\langle a, b-a+1-b\rangle 
&+ \langle b-a, 1-b \rangle 
- \langle a+b -a, 1-b \rangle 
- \langle a, b-a \rangle \\
&= 
\langle a, 1-a \rangle + \langle 1-b, b-a \rangle - \langle a, b-a \rangle  
- \langle b,1-b \rangle =0.
\end{align*}
So 
\begin{align*}
\rho(0) 
&= \rho(H(0)) = \rho\left(H(a)-H(b) + a H\left(\frac{b}{a}\right) + (1-a)H\left(\frac{1-b}{1-a}\right)\right) \\
&=  \langle a,1-a \rangle - \langle b,1-b\rangle - \langle a,b-a \rangle 
+ \langle 1-b, b-a \rangle 
= 0. 
\end{align*}
Conversely, let $\lambda: J(\mathbf{k})\rightarrow H(\mathbf{k})$, where 
$\lambda(\langle a,b \rangle) 
= (a+b) H \left( \frac{a}{a+b} \right)$ if $a+b\not=0$, and $\lambda(\langle a,b\rangle)=0$ if $a+b=0$.
So we have 
\begin{align*}
\lambda(\langle a,b \rangle) 
&= (a+b)H\left( \frac{a}{a+b} \right) \\
&= (a+b)H\left( 1- \frac{a}{a+b} \right) \mbox{ by the first equality in Lemma~\ref{lemma_entropy_basic}} \\
&= (a+b)H\left( \frac{\cancel{a} + b - \cancel{a}}{a+b} \right) \\ 
&= (b+a)H\left( \frac{b}{a+b} \right) \\
&= \lambda(\langle b,a \rangle).
\end{align*}
Secondly, $\lambda(\langle ca, cb \rangle)= c\lambda(\langle a,b \rangle)$ holds by scaling in Lemma~\ref{lemma_Cath_properties}. 

Lastly, we have 
\begin{align*}
0 &= \lambda(0) \\ 
&= \lambda(\langle a,b+c\rangle +\langle b,c \rangle - \langle a+b,c \rangle - \langle a,b\rangle) \\ 
&= \lambda(\langle a,b+c\rangle) +\lambda (\langle b,c \rangle) - \lambda (\langle a+b,c \rangle) - \lambda (\langle a,b\rangle) \\ 
&= (a+b+c) H\left( \frac{a}{a+b+c} \right) 
+ (b+c) H\left( \frac{b}{b+c} \right)
- (a+b+c) H\left( \frac{a+b}{a+b+c} \right) 
- (a+b) H\left( \frac{a}{a+b} \right) \\
&= (a+b+c) H\left( \frac{a}{a+b+c} \right) 
+ (b+c) H\left( \frac{b}{b+c} \right)
- (a+b+c) H\left( \frac{a+b}{a+b+c} \right) 
- (a+b) H\left( \frac{b}{a+b} \right)
\end{align*}
since $H\left( \frac{a}{a+b} \right)= H\left( \frac{b}{a+b} \right)$ by the first equality in Lemma~\ref{lemma_entropy_basic}.

Let 
\[
u= \frac{a}{a+b+c} 
\quad 
\mbox{ and }
\quad 
v = \frac{c}{a+b+c}. 
\]
Then 
\begin{align*}
(a+b+c) & H\left( \frac{a}{a+b+c} \right) 
+ (b+c) H\left( \frac{b}{b+c} \right)
- (a+b+c) H\left( \frac{a+b}{a+b+c} \right) 
- (a+b) H\left( \frac{b}{a+b} \right) \\ 
&= (a+b+c) 
\left( H(u) + (1-u)H\left( \frac{1-u-v}{1-u}\right) - H(1-v) - (1-v)H\left( \frac{1-u-v}{1-v}  \right) 
\right) \\ 
&= (a+b+c) 
\left( H(u) + (1-u)H\left( \frac{v}{1-u}\right) - H(1-v) - (1-v)H\left( \frac{u}{1-v}  \right) 
\right)
\end{align*}
by the first equality in Lemma~\ref{lemma_entropy_basic}. 
Now let $p=u$ and $q=1-v$. Then 
\begin{align*}
H(u) &+ (1-u)H\left( \frac{v}{1-u}\right) - H(1-v) - (1-v)H\left( \frac{u}{1-v}  \right) \\
&= H(p) + (1-p)H\left( \frac{1-q}{1-p} \right) - H(q) - q H\left( \frac{p}{q} \right)  \\
&= H(p) + (1-p)H\left( \frac{1-q}{1-p} \right) - H(q) + \cancel{q} \frac{p}{\cancel{q}}  H\left( \frac{q}{p} \right)  \\
&= H(p) + (1-p)H\left( \frac{1-q}{1-p} \right) - H(q) +  p H\left( \frac{q}{p} \right) \\ 
&= 0  
\end{align*}
by the second equality in Lemma~\ref{lemma_entropy_basic} and \eqref{eq_four_term_relation_entropy}.
\end{proof}

In characteristic $p>0$, we have 
\begin{align*}
\rho\left(\sum_{n=1}^p H(n)\right) 
&= \sum_{n=1}^p \rho( H(n)) 
= \sum_{n=1}^p \langle n, 1-n \rangle 
= \sum_{n=1}^p \langle 1-n, -1 \rangle \\ 
&= - \sum_{n=1}^p \langle n-1, 1 \rangle 
= - \sum_{n=1}^p \langle 1, n-1 \rangle 
= - \sum_{n=1}^p \langle 1, n \rangle 
= 0,
\end{align*}
where the last equality holds by \eqref{eqn_char_p_Cathelineau}.

Entropy is related to the pair of symbols $\langle \cdot , \cdot \rangle$ via the following: 
\begin{proposition}
\label{thm_entropy_inner_prod}
Let $p_1,\ldots, p_n$ be a finite probability distribution on a finite set $X=\{x_1,\ldots, x_n \}$.
Then for $n\geq 3$, 
\begin{equation}
\label{eqn_relating_entropy_inner}
H(p_X) = H(p_1,\ldots, p_n) = \sum_{j=2}^n 
\left\langle \sum_{i=1}^{j-1} p_i, p_j \right\rangle. 
\end{equation}
\end{proposition}
Expanding out~\eqref{eqn_relating_entropy_inner}, we have 
\[
H(p_1,\ldots, p_n) = \langle p_1, p_2 \rangle + \langle p_1 + p_2, p_3 \rangle 
+ \langle p_1 + p_2 + p_3 , p_4\rangle 
+ \ldots
+ \langle p_1 + p_2 + \ldots + p_{n-1} , p_n \rangle. 
\]

\begin{proof}
Since  
\begin{align*}
\langle p_1, p_2 \rangle 
&+ \langle p_1 + p_2, p_3 \rangle 
+ \langle p_1 + p_2 + p_3 , p_4\rangle 
+ \ldots 
+ \langle p_1 + p_2 + \ldots + p_{n-1} , p_n \rangle \\
&= (p_1+ p_2) H\left( \frac{p_1}{p_1 + p_2} \right) + 
(p_1+ p_2+p_3) H\left( \frac{p_1+p_2}{p_1 + p_2 + p_3} \right) + \ldots + \\  
&\hspace{4mm}+ 
(p_1+ p_2+p_3+ \ldots + p_n) H\left( \frac{p_1+p_2+ \ldots + p_{n-1}}{p_1 + p_2 + p_3 +\ldots + p_n} \right) \\
&= (p_1 + p_2) \left(-\frac{p_1}{p_1 + p_2} \log \left| \frac{p_1}{p_1 + p_2} \right|
- \left(1 - \frac{p_1}{p_1 + p_2}\right) 
\log \left| 1 - \frac{p_1}{p_1 + p_2}\right| \right) \\
&\hspace{4mm}+ (p_1 + p_2 + 
    p_3) \left(-\frac{p_1 + p_2}{p_1 + p_2 + p_3} \log \left| \frac{p_1 + p_2}{
      p_1 + p_2 + p_3}\right| \right. \\
&\hspace{4mm}\left. - \left(1 - \frac{p_1 + p_2}{p_1 + p_2 + p_3}\right) \log \left|
      1 - \frac{p_1 + p_2}{p_1 + p_2 + p_3}\right| 
      \right)  + \ldots + \\
&\hspace{4mm}+ (p_1 + p_2 + p_3 + \ldots + 
    p_n) \left(-\frac{p_1 + p_2 + \ldots + p_{n-1}}{p_1 + p_2 + p_3 + \ldots + p_n} \log \left| \frac{p_1 + p_2 + \ldots + p_{n-1}}{p_1 + p_2 + p_3 + \ldots +  p_n} \right| \right. \\
&\hspace{4mm}\left. - \left(1 - \frac{p_1 + p_2 + \ldots + p_{n-1}}{p_1 + p_2 + p_3 + \ldots + p_n}\right) \log \left|
      1 - \frac{p_1 + p_2 + \ldots + p_{n-1}}{p_1 + p_2 + p_3 + \ldots + p_n}\right|
      \right) \\ 
&= \cancel{(p_1 + p_2)} \left(-\frac{p_1}{\cancel{p_1 + p_2}} \log \left| \frac{p_1}{p_1 + p_2} \right|
-  \frac{p_2}{\cancel{p_1 + p_2}} 
\log \left| \frac{p_2}{p_1 + p_2}\right| \right) \\
&\hspace{4mm}+ \cancel{(p_1 + p_2 + 
    p_3)} \left(-\frac{p_1 + p_2}{\cancel{p_1 + p_2 + p_3}} \log \left| \frac{p_1 + p_2}{
      p_1 + p_2 + p_3}\right| - \frac{p_3}{\cancel{p_1 + p_2 + p_3}} \log \left|
       \frac{p_3}{p_1 + p_2 + p_3}\right|
      \right)  \\
&\hspace{4mm}+ \ldots + \cancel{(p_1 + p_2 + \ldots + p_n)} \left(-\frac{p_1 + p_2 + \ldots + p_{n-1}}{\cancel{p_1 + p_2 + p_3 +\ldots + p_n}} \log \left|\frac{p_1 + p_2 + \ldots + p_{n-1}}{p_1 + p_2 + p_3 + \ldots + p_n} \right| \right. \\
&\hspace{4mm}\left. - \frac{p_n}{\cancel{p_1 + p_2 + p_3 + \ldots + p_n}} \log \left|
      \frac{p_n}{p_1 + p_2 + p_3 + \ldots + p_n}\right|
      \right) \\ 
&=  -p_1\log \left| \frac{p_1}{p_1 + p_2} \right| - p_2  
\log \left| \frac{p_2}{p_1 + p_2}\right| 
- (p_1 + p_2) 
\log \left| \frac{p_1 + p_2}{
 p_1 + p_2 + p_3}\right|  
- p_3 \log \left|
 \frac{p_3}{p_1 + p_2 + p_3}\right| \\
&\hspace{4mm}   - \ldots 
- (p_1 + p_2 + \ldots + p_{n-1}) \log \left|
      \frac{p_1 + p_2 + \ldots + p_{n-1}}{p_1 + p_2 + p_3 + \ldots + p_n}\right| \\ 
&\hspace{4mm} - p_n \log \left|
      \frac{p_n}{p_1 + p_2 + p_3 + \ldots + p_n}\right| \\ 
&= -p_1 \log |p_1| - p_2 \log | p_2| 
+ \cancel{(p_1 + p_2) \log |p_1 + p_2|}
- \cancel{(p_1 + p_2) \log |p_1 + p_2|} 
- p_3 \log |p_3| \\
&\hspace{4mm}
+ \cancel{(p_1 + p_2 + p_3) \log |p_1 + p_2 + p_3|}
- \cancel{(p_1 + p_2 + p_3) \log |p_1 + p_2 + p_3|} + \ldots + \\ 
&\hspace{4mm} + 
(p_1 +p_2 + \ldots + p_{n-1}) \cancel{\log |p_1 + p_2 + \ldots + p_{n-1}|} \\ 
&\hspace{4mm} - 
(p_1 +p_2 + \ldots + p_{n-1}) \cancel{\log |p_1 + p_2 + \ldots + p_{n-1}|} - p_n \log |p_n| \\ 
&\hspace{4mm} + (\underbrace{p_1 + p_2 + p_3 + \ldots + p_n}_{1}) \log |\underbrace{p_1 + p_2 + p_3 + \ldots + p_n}_{1}| \\
&= -p_1 \log |p_1| - p_2 \log |p_2| - p_3 \log |p_3| -\ldots - p_n \log |p_n| = H(p_X). 
\end{align*}
This concludes the proof.
\end{proof}

\section{Joint entropy}

\subsection{Joint entropy of \texorpdfstring{$2$}{2} random variables}
\label{subsection:joint_entropy}

Let $X$ be a discrete random variable associated with probabilities $p_1,\ldots, p_n$. Recall Shannon entropy in~\eqref{eq_shannon}.
Let $Y$ be a discrete random variable associated with probabilities $q_1,\ldots, q_k$. Then the joint entropy of $X$ and $Y$ is defined to be 
\begin{equation}
\label{eqn:joint_entropy_prob}
H(p_X,p_Y) =  -\sum_{i=1}^{n} \sum_{j=1}^{m} p_{ij} \log p_{ij},
\end{equation}
where $p_{ij}= P(X = p_i, Y = q_j)$, the probability that the random variables $X=p_i$ and $Y=q_j$.
Then similar to ~\eqref{eqn_diagrammatic_entropy}, we have 
\begin{equation}
\label{eqn:joint_entropy_inner_product}
\langle p_{ij}, p_{\ell k} \rangle = 
\begin{cases} 
(p_{ij}+p_{\ell k}) H\left( \dfrac{p_{ij}}{p_{ij}+p_{\ell k}} \right) &\mbox{ if } p_{ij}+p_{\ell k} \not=0, \\
\hspace{1.75cm} 0 &\mbox{ if } p_{ij}+p_{\ell k} = 0.
\end{cases}
\end{equation}

\begin{proposition}
\label{proposition:joint_entropy_inner_product}
Equation~\eqref{eqn:joint_entropy_inner_product} satisfies symmetry, scaling, and the $2$-cocycle relation for the vector space $J(\mathbf{k})$. 
\end{proposition}

The proof of Proposition~\ref{proposition:joint_entropy_inner_product} is similar to the proof of Lemma~\ref{lemma_Cath_properties}.
 
Similar in spirit to Proposition~\ref{thm_entropy_inner_prod} and in order to provide a motivation for Theorem~\ref{thm_relating_entropy_inner_2}, we have the following: 
\begin{proposition}
\label{prop_joint_entropy_two}
Let $p_1, p_2$ and $q_1, q_2$ be finite probability distributions on sets $X=\{ x_1, x_2\}$ and $Y=\{ y_1, y_2\}$, respectively. Let $p_{ij}$ be the probability of obtaining $p_i$ and $q_j$. Then a relation between entropy and Cathelineau's vector space is
\begin{equation}
H(p_X,p_Y) = \langle p_{11},p_{21}\rangle + \langle p_{12},p_{22} \rangle +\langle p_{11}+p_{21},p_{12}+p_{22}\rangle. 
\end{equation}
\end{proposition}

\begin{proof}
Using \eqref{eqn:joint_entropy_inner_product}, we have: 
\begin{align*}
\langle p_{11}, p_{21} \rangle &= (p_{11}+p_{21}) H\left(\frac{p_{11}}{p_{11}+p_{21}} \right), \\
\langle p_{12}, p_{22} \rangle &= (p_{12} + p_{22}) H\left(\frac{p_{12}}{p_{12}+p_{22}} \right), \\
\langle p_{11}+p_{21}, p_{12}+p_{22} \rangle &= (p_{11}+p_{21}+p_{12}+p_{22}) H\left(\frac{p_{11}+p_{21}}{p_{11}+p_{21}+p_{12}+p_{22}} \right). \\
\end{align*}
So 
\begin{align*}
\langle p_{11}&,p_{21}\rangle + \langle p_{12},p_{22} \rangle +\langle p_{11}+p_{21},p_{12}+p_{22}\rangle
= (p_{11}+p_{21}) H\left(\frac{p_{11}}{p_{11}+p_{21}} \right) \\
&\hspace{4mm}+ (p_{12} + p_{22}) H\left(\frac{p_{12}}{p_{12}+p_{22}} \right) + (p_{11}+p_{21}+p_{12}+p_{22}) H\left(\frac{p_{11}+p_{21}}{p_{11}+p_{21}+p_{12}+p_{22}} \right) \\
&= (p_{11} + p_{21}) \left(-\frac{p_{11}}{p_{11} + p_{21}} \log \left( \frac{p_{11}}{p_{11} + p_{21}} \right) 
      - \left(1 - \frac{p_{11}}{p_{11} + p_{21}}\right) \log \left( 1 - \frac{p_{11}}{p_{11} + p_{21}}\right) \right) \\
&\hspace{4mm}+
 (p_{12} + p_{22}) \left(-\frac{p_{12}}{p_{12} + p_{22}} \log \left( \frac{p_{12}}{p_{12} + p_{22}} \right) 
      - \left(1 - \frac{p_{12}}{p_{12} + p_{22}}\right) 
      \log \left( 1 - \frac{p_{12}}{p_{12} + p_{22}}\right) \right) \\
&\hspace{4mm}+
 (p_{11} + p_{21} + p_{12} + p_{22}) 
 \left( -\frac{p_{11} + p_{21}}{p_{11} + p_{21} + p_{12} + p_{22}} 
 \log\left( \frac{p_{11} + p_{21}}{p_{11} + p_{21} + p_{12} + p_{22}} \right) \right. \\
&\hspace{4mm} \left. - \left(1 - 
\frac{p_{11} + p_{21}}{p_{11} + p_{21} + p_{12} + p_{22}}\right) 
       \log \left( 
      1 - \frac{p_{11} + p_{21}}{p_{11} + p_{21} + p_{12} + p_{22}}\right) \right) \\ 
&= \cancel{(p_{11} + p_{21})} 
\left(-\frac{p_{11}}{\cancel{p_{11} + p_{21}}} \log \left( \frac{p_{11}}{p_{11} + p_{21}} \right) 
      - \frac{p_{21}}{\cancel{p_{11} + p_{21}}} \log \left(\frac{p_{21}}{p_{11} + p_{21}}\right) \right) \\
&\hspace{4mm}+
 \cancel{(p_{12} + p_{22})} \left(-\frac{p_{12}}{\cancel{p_{12} + p_{22}}} \log \left( \frac{p_{12}}{p_{12} + p_{22}} \right) 
      - \frac{p_{22}}{\cancel{p_{12} + p_{22}}} 
      \log \left(\frac{p_{22}}{p_{12} + p_{22}}\right) \right) \\
&\hspace{4mm}+
 \cancel{(p_{11} + p_{21} + p_{12} + p_{22})} 
 \left( -\frac{p_{11} + p_{21}}{
 \cancel{p_{11} + p_{21} + p_{12} + p_{22}}} 
 \log\left( \frac{p_{11} + p_{21}}{p_{11} + p_{21} + p_{12} + p_{22}} \right) \right. \\
&\hspace{4mm} \left. - 
\frac{p_{12} + p_{22}}{
\cancel{p_{11} + p_{21} + p_{12} + p_{22}}}  
       \log \left( 
     \frac{p_{12} + p_{22}}{p_{11} + p_{21} + p_{12} + p_{22}}\right) \right) \\ 
&=  
- p_{11} \log \left( \frac{p_{11}}{p_{11} + p_{21}} \right) 
      - p_{21} \log \left(\frac{p_{21}}{p_{11} + p_{21}}\right) - p_{12} \log \left( \frac{p_{12}}{p_{12} + p_{22}} \right) - p_{22} 
      \log \left(\frac{p_{22}}{p_{12} + p_{22}}\right) \\
&\hspace{4mm} 
  - (p_{11} + p_{21}) 
 \log\left( \frac{p_{11} + p_{21}}{p_{11} + p_{21} + p_{12} + p_{22}} \right) - (p_{12} + p_{22}) 
\log \left( 
     \frac{p_{12} + p_{22}}{p_{11} + p_{21} + p_{12} + p_{22}}\right) \\ 
&= - p_{11} \log ( p_{11} ) 
      - p_{21} \log (p_{21}) 
      + \cancel{(p_{11} + p_{21}) \log (p_{11} + p_{21})} \\ 
&\hspace{4mm}- p_{12} \log ( p_{12} ) - p_{22} 
      \log (p_{22}) 
      + \cancel{(p_{12} + p_{22}) \log (p_{12} + p_{22})} \\
&\hspace{4mm} 
- \cancel{(p_{11} + p_{21})\log\left( p_{11} + p_{21} \right)} - 
\cancel{(p_{12} + p_{22}) \log (p_{12} + p_{22})} \\
&\hspace{4mm}
+ (\underbrace{p_{11} + p_{21} + p_{12} + p_{22}}_{1}) \log (\underbrace{p_{11} + p_{21} + p_{12} + p_{22}}_{1}) \\
&=- p_{11} \log ( p_{11} ) 
      - p_{21} \log (p_{21}) - p_{12} \log ( p_{12} ) - p_{22} 
      \log (p_{22}) = H(p_X,p_Y).
\end{align*}

\end{proof}

More generally, we have the following: 
\begin{theorem}
\label{thm_relating_entropy_inner_2}
Let $p_1,\ldots, p_n$ and $q_1,\ldots, q_m$ be finite probability distributions on finite sets $X=\{x_1,\ldots, x_n \}$ and $Y=\{y_1, \ldots, y_m \}$, respectively. 
Let $p_{ij}=P(p_i,q_j)$, the probability of obtaining $p_i$ and $q_j$. 
Then 
\begin{equation}
H(p_X, p_Y) = \sum_{i=1}^n \sum_{k=1}^{m-1} \langle  \sum_{j=1}^k p_{ij}, p_{i,j+1} \rangle + \sum_{k=1}^{n-1} \langle \sum_{i=1}^{k} \sum_{j=1}^{m} p_{ij}, \sum_{\ell = 1}^m p_{i\ell} \rangle. 
\end{equation}
\end{theorem}

The proof of Theorem~\ref{thm_relating_entropy_inner_2} is a direct calculation, similar to the proof of Proposition~\ref{thm_entropy_inner_prod} and Proposition~\ref{prop_joint_entropy_two}. We give a diagrammatical calculus perspective of entropy in Section~\ref{section:Shannon_entropy}, which originated in~\cite{IK24_dilogarithms_entropy}. Using this, one can easily prove Theorem~\ref{thm_relating_entropy_inner_2}.

\subsection{Entropy for \texorpdfstring{$u$}{u} random variables}

One may extend \eqref{eqn:joint_entropy_prob} as follows. Consider the finite set $X_i = \{x_{i_1},\ldots, x_{i_{k_i}} \}$, where $1\leq i\leq u$. Let $p_{X_i}$ be a probability distribution on $X_i$ that associates probabilities $p_{i_1},\ldots, p_{i_{k_i}}$ to the points $x_{i_1},\ldots, x_{i_{k_i}}$, respectively, where the sum $\displaystyle{\sum_{j=1}^{k_i}} p_{i_j}=1$, $0< p_{i_j}<1$, $1\leq i\leq u$.

Write 
\[ I_{J} := 1_{j_1}2_{j_2}\cdots u_{j_u}, 
\hspace{4mm} 
J := (j_1,j_2,\ldots, j_u), 
\hspace{4mm} 
\mbox{ and }
\hspace{4mm} 
K := (k_1,\ldots, k_u).
\] 
Let the joint probability for random variables $\mathfrak{X}_1, \ldots, \mathfrak{X}_u$ be 
\[ 
p_{I_{J}} 
= p_{1_{j_1}2_{j_2}\cdots u_{j_u}} 
= 
P(\mathfrak{X}_1 = p_{1_{j_1}}, \mathfrak{X}_2 = p_{2_{j_2}}, \ldots, \mathfrak{X}_u = p_{u_{j_u}}).
\]
We obtain the joint entropy
\begin{equation}
\label{eqn:joint_entropy_generalized}
H(p_{X_1},\ldots, p_{X_u}) 
= 
-\sum_{j_1=1}^{k_1}\sum_{j_2=1}^{k_2} \cdots \sum_{j_u=1}^{k_u} 
p_{1_{j_1}2_{j_2}\cdots u_{j_u}} 
\log (p_{1_{j_1}2_{j_2}\cdots u_{j_u}}),
\end{equation}
or more compactly, we write 
\begin{equation}
\label{eqn_joint_entropy_u_cmpt}
H(p_{X_1},\ldots, p_{X_u}) = 
- \sum_{J = \mathbf{1}}^{K} 
p_{I_{J}} \log(p_{I_{J}}).
\end{equation}

\begin{proposition}
\label{proposition:joint_entropy_general_inner_product}
Under the conditions as above, we have 
\begin{equation}
\label{eqn:joint_entropy_u_rv}
\langle p_{I_{J}}, p_{I_{J'}} \rangle = 
\left( p_{I_{J}} + p_{I_{J'}} \right) 
H\left( \frac{p_{I_{J}}}{p_{I_{J}} + p_{I_{J'}}} \right), 
\end{equation}
which satisfies symmetry, scaling, and the $2$-cocycle relation for $J(\kk)$ in Section~\ref{subsection:Cathelineau_k_vs}.
\end{proposition}

The proof of Proposition~\ref{proposition:joint_entropy_general_inner_product} is similar to the proof of Lemma~\ref{lemma_Cath_properties}.

\subsection{Rescaling a finite probability distribution by a scalar}
\label{section:scaling_scalar}
Rescaling the pair of symbols $\langle c p_1, cp_2\rangle = c \langle p_1, p_2 \rangle$ corresponds to rescaling entropy. 
In our diagrammatic story, it will correspond to a red wavy line, which is to the left of our graphical network.

\section{Diagrammatics of Shannon entropy}
\label{section:Shannon_entropy}
Let $\kk=\mathbb{R}$ and let $\kk^*$ be the set of units in $\kk$. In this section, we discuss a new perspective of entropy using diagrammatics, as introduced in~\cite{IK24_dilogarithms_entropy}.

We decorate the boundary of a network of lines from values in $\kk$. 
Figure~\ref{fig_1001} gives us the properties of black and red (network) lines. Black lines are additive and take values in $\kk$ while red wavy lines are multiplicative and take values in $\kk^*$. Whenever two additive (black) lines merge (at the intersections of additive lines), we evaluate at the additive vertex using the two symbols $\langle a,b \rangle$, see Figure~\ref{fig_1001} top left. Whenever two additive lines split, we evaluate at the additive vertex the pair of symbols $\langle a,b\rangle$, but with a negative sign:
$-\langle a,b\rangle$, see Figure~\ref{fig_1001} top middle. We can also have a virtual crossing, where two lines cross but their intersection is said to be virtual, see Figure~\ref{fig_1001} top right. In the second row of Figure~\ref{fig_1001}, we see that red wavy (multiplicative) lines have a conormal direction. In particular, whenever an additive line passes through a multiplicative line of weight $c\in \mathbf{k}^*$, we rescale the additive line from $a\mapsto ca$, see Figure~\ref{fig_1001} bottom left. Similar to the additive lines, whenever multiplicative lines merge, we multiply the two weights, Figure~\ref{fig_1001} bottom middle. Finally, we introduce a red dot to represent the switching of normal co-orientations, see Figure~\ref{fig_1001} bottom right.

\begin{figure}
\begin{center}
\begin{tikzpicture}[scale=0.6,decoration={
    markings,
    mark=at position 0.70 with {\arrow{>}}}]
\begin{scope}[shift={(0,0)}]

\node at (5.75,2.4) {additive};
\node at (5.75,1.65) {lines};
\node at (5.75,0.5) {$\langle a,b\rangle$};

\node at (2,3.5) {$a+b$};
\draw[thick,dashed] (0,3) -- (4,3);
\draw[thick,dashed] (0,0) -- (4,0);

\node at (0.5,-0.5) {$a$};
\node at (3.5,-0.5) {$b$};

\draw[thick,postaction={decorate}] (0.5,0) .. controls (0.75,0.75) and (1.3,1.25) .. (2,1.25);

\draw[thick,postaction={decorate}] (3.5,0) .. controls (3.25,0.75) and (2.7,1.25) .. (2,1.25);

\draw[thick,<-] (2,2) -- (2,1.25);
\draw[thick] (2,3) -- (2,2);
\end{scope}

\begin{scope}[shift={(9,0)}]
\draw[thick,dashed] (0,3) -- (4,3);
\draw[thick,dashed] (0,0) -- (4,0);

\node at (5.75,1.5) {$-\langle a,b\rangle$};

\node at (0.5,3.5) {$a$};
\node at (3.5,3.5) {$b$};

\draw[thick,postaction={decorate}] (2,1.75) .. controls (1.6,1.75) and (0.75,2.25) .. (0.5,3);

\draw[thick,postaction={decorate}] (2,1.75) .. controls (2.4,1.75) and (3.25,2.25) .. (3.5,3);

\draw[thick] (2,1) -- (2,1.75);
\draw[thick,->] (2,0) -- (2,1);

\node at (2,-0.5) {$a+b$};
\end{scope}

\begin{scope}[shift={(18,0)}]

\draw[thick,dashed] (0,3) -- (4,3);
\draw[thick,dashed] (0,0) -- (4,0);

\node at (5.75,1.9) {virtual};
\node at (5.75,1.1) {crossing};

\draw[thick,postaction={decorate}] (0.75,0) -- (3.25,3);

\draw[thick,postaction={decorate}] (3.25,0) -- (0.75,3);

\node at (0.6,3.5) {$b$};
\node at (3.4,3.5) {$a$};

\node at (0.6,-0.5) {$a$};
\node at (3.4,-0.5) {$b$};
\end{scope}

\begin{scope}[shift={(0,-5.5)}]

\node at (0.85,3.4) {$ca$};
\draw[thick,dashed] (0,3) -- (4,3);
\draw[thick,dashed] (0,0) -- (4,0);
\node at (3.15,-0.4) {$a$};

\draw[thick,->] (3,0) -- (1.67,2);
\draw[thick] (1.67,2) -- (1,3);

\draw[line width=0.50mm,red] (1,0) decorate [decoration={snake,amplitude=0.15mm}] {-- (3,3)};

\draw[line width=0.50mm,red] (0.9,0.75) -- (1.5,0.75);

\draw[line width=0.50mm,red] (2.1,2.5) -- (2.65,2.5);

\node at (0.6,0.40) {$c$};

\node at (3,2.3) {$c$};

\end{scope}

\begin{scope}[shift={(9,-5.5)}]

\node at (6,2.5) {merging of};
\node at (6,1.75) {multiplicative};
\node at (6,1.00) {lines};

\draw[thick,dashed] (0,3) -- (4,3);
\draw[thick,dashed] (0,0) -- (4,0);

\node at (2.75,2.2) {$c_1c_2$};

\draw[line width=0.50mm,red] (0.75,0) decorate [decoration={snake,amplitude=0.15mm}] {.. controls (0.85,0.75) and (1,0.85) .. (2,1.25)};

\draw[line width=0.50mm,red] (3.25,0) decorate [decoration={snake,amplitude=0.15mm}] {.. controls (3.15,0.75) and (3,0.85) .. (2,1.25)};
 
\draw[line width=0.50mm,red] (2,1.25) decorate [decoration={snake,amplitude=0.15mm}] {-- (2,3)};
 
\draw[line width=0.50mm,red] (0.30,0.45) -- (0.80,0.45);

\draw[line width=0.50mm,red] (2.6,0.45) -- (3.2,0.45);

\draw[line width=0.50mm,red] (1.5,2) -- (2,2);

\node at (0.5,0.85) {$c_1$};

\node at (3.5,0.85) {$c_2$};
\end{scope}

\begin{scope}[shift={(18,-5.5)}]

\draw[thick,dashed] (0,3) -- (4,3);
\draw[thick,dashed] (0,0) -- (4,0);

\draw[line width=0.50mm,red] (2,0) decorate [decoration={snake,amplitude=0.15mm}] {-- (2,3)};

\draw[thick,fill,red] (2.15,1.5) arc (0:360:1.5mm);

\draw[line width=0.50mm,red] (2,2.45) -- (2.5,2.45);

\draw[line width=0.50mm,red] (1.5,0.55) -- (2,0.55);

\node at (1.1,0.5) {$c$};

\node at (3.15,2.40) {$c^{-1}$};

\node at (6,2.00) {swapping the};
\node at (6,1.25) {co-orientation};
\end{scope}

\end{tikzpicture}
\end{center}
    \caption{Upper left: Whenever two black additive lines merge, we evaluate $\langle a,b\rangle$ at the additive vertex. Upper middle: whenever two additive lines split, we evaluate $-\langle a,b\rangle$. Upper right: we are allowed to have virtual crossing whenever two additive lines cross but the intersection of these two lines is virtual, so there is no corresponding evaluation for these two additive lines. Bottom left: whenever a red line is to the left of a black line, then we rescale the value $a$ of the additive line by $c\in \mathbf{k}^*$ of the multiplicative line.  Bottom middle: multiplicative lines can merge, resulting in multiplication of their weights. Bottom right: we may have a red vertex on the multiplicative network, resulting in the swapping of co-orientations of the multiplicative line.}
    \label{fig_1001}
\end{figure}
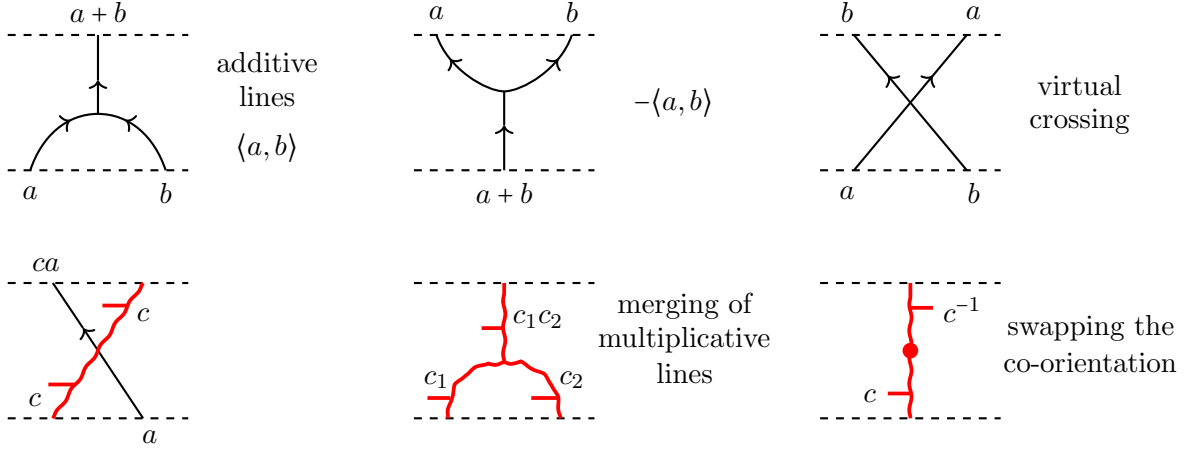

\begin{figure}
\begin{center}
\begin{tikzpicture}[scale=0.6,decoration={
    markings,
    mark=at position 0.60 with {\arrow{>}}}]


\begin{scope}[shift={(0,0)}]

\draw[thick,dashed] (0,4) -- (4,4);

\draw[thick,dashed] (0,0) -- (4,0);

\node at (0.5,4.5) {$a$};

\node at (3.5,4.5) {$b$};

\node at (2,-0.5) {$a+b$};

\draw[thick,postaction={decorate}] (0.5,4) .. controls (0.7,2.75) and (1.3,2.5) .. (2,2);

\draw[thick,postaction={decorate}] (3.5,4) .. controls (3.3,2.75) and (2.7,2.5) .. (2,2);

\draw[thick,postaction={decorate}] (2,2) -- (2,0);
 
\node at (6,2) {$=$};

\end{scope}


\begin{scope}[shift={(8,0)}]

\draw[thick,dashed] (0,4) -- (4,4);
\draw[thick,dashed] (0,0) -- (4,0);

\node at (0.5,4.5) {$a$};

\node at (3.5,4.5) {$b$};

\node at (2,-0.5) {$a+b$};

\draw[thick,postaction={decorate}] (2,2) .. controls (2.25,3) and (3.25,3) .. (3.5,2);

\draw[thick,postaction={decorate}] (0.75,2) .. controls (1,1.25) and (1.75,1.25) .. (2,2);

\draw[thick,postaction={decorate}] (2,0.75) .. controls (2.75,1) and (2.75,1.5) .. (2,2);

\draw[thick,postaction={decorate}] (3.5,4) .. controls (4,3.5) and (0.5,3) .. (0.75,2);

\draw[thick,postaction={decorate}] (3.5,2) .. controls (3.75,0.75) and (2.25,0.25) .. (2,0);

\draw[thick,postaction={decorate}] (0.5,4) .. controls (-0.25,3) and (0.25,0.25) .. (2,0.75);

\node at (6,2) {$=$};
\end{scope}


\begin{scope}[shift={(16,0)}]

\draw[thick,dashed] (0,4) -- (4,4);
\draw[thick,dashed] (0,0) -- (4,0);

\node at (0.5,4.5) {$a$};
\node at (3.5,4.5) {$b$};

\node at (2,-0.5) {$a+b$};

\draw[thick,postaction={decorate}] (2,2) .. controls (1.75,3) and (0.75,3) .. (0.5,2);

\draw[thick,postaction={decorate}] (3.25,2) .. controls (3,1.25) and (2.25,1.25) .. (2,2);

\draw[thick,postaction={decorate}] (2,0.75) .. controls (1.25,1) and (1.25,1.5) .. (2,2);

\draw[thick,postaction={decorate}] (0.5,4) .. controls (0,3.5) and (3.5,3) .. (3.25,2);

\draw[thick,postaction={decorate}] (0.5,2) .. controls (0.25,0.75) and (1.75,0.25) .. (2,0);

\draw[thick,postaction={decorate}] (3.5,4) .. controls (4.25,3) and (3.75,0.25) .. (2,0.75);

\node at (6,2) {$\langle a,b\rangle$};
\end{scope}

\begin{scope}[shift={(0,-6.5)}]

\node at (2,4.5) {$a+b$};
\draw[thick,dashed] (0,4) -- (4,4);
\draw[thick,dashed] (0,0) -- (4,0);
\node at (0.5,-0.5) {$a$};
\node at (3.5,-0.5) {$b$};

\draw[thick,postaction={decorate}] (2,2) .. controls (1.3,1.5) and (0.7,1.25) .. (0.5,0);

\draw[thick,postaction={decorate}] (2,2) .. controls (2.7,1.5) and (3.3,1.25) .. (3.5,0);

\draw[thick,postaction={decorate}] (2,4) -- (2,2);
\node at (6,2) {$=$};

\end{scope}

\begin{scope}[shift={(8,-6.5)}]

\node at (2,4.5) {$a+b$};

\draw[thick,dashed] (0,4) -- (4,4);

\draw[thick,dashed] (0,0) -- (4,0);

\node at (0.5,-0.5) {$a$};

\node at (3.5,-0.5) {$b$};

\draw[thick,postaction={decorate}] (0.5,2) .. controls (0.75,1) and (1.75,1) .. (2,2);

\draw[thick,postaction={decorate}] (2,2) .. controls (2.25,2.75) and (3,2.75) .. (3.25,2);

\draw[thick,postaction={decorate}] (2,2) .. controls (1.25,2.5) and (1.25,3) .. (2,3.25);

\draw[thick,postaction={decorate}] (3.25,2) .. controls (3.5,1) and (0,0.5) .. (0.5,0);

\draw[thick,postaction={decorate}] (2,4) .. controls (1.75,3.75) and (0.25,3.25) .. (0.5,2);

\draw[thick,postaction={decorate}] (2,3.25) .. controls (3.75,3.75) and (4.25,1) .. (3.5,0);

\node at (6,2) {$=$};
\end{scope}


\begin{scope}[shift={(16,-6.5)}]

\node at (2,4.5) {$a+b$};

\draw[thick,dashed] (0,4) -- (4,4);

\draw[thick,dashed] (0,0) -- (4,0);

\node at (0.5,-0.5) {$a$};

\node at (3.5,-0.5) {$b$};

\draw[thick,postaction={decorate}] (3.5,2) .. controls (3.25,1) and (2.25,1) .. (2,2);

\draw[thick,postaction={decorate}] (2,2) .. controls (1.75,2.75) and (1,2.75) .. (0.75,2);

\draw[thick,postaction={decorate}] (2,2) .. controls (2.75,2.5) and (2.75,3) .. (2,3.25);

\draw[thick,postaction={decorate}] (0.75,2) .. controls (0.5,1) and (4,0.5) .. (3.5,0);

\draw[thick,postaction={decorate}] (2,4) .. controls (2.25,3.75) and (3.75,3.25) .. (3.5,2);

\draw[thick,postaction={decorate}] (2,3.25) .. controls (0.25,3.75) and (-0.25,1) .. (0.5,0);

\node at (6,2) {$-\langle a,b\rangle$};
\end{scope}

\end{tikzpicture}
\end{center}
    \caption{When additive lines are pointing downwards, we rotate the additive vertices in the network whilst fixing the boundary points so that the orientations at the additive vertices are upwards. Top row: the additive vertex in each of the figures gives the contribution of $\langle a,b\rangle$. Bottom row: the additive vertices give $-\langle a,b\rangle$. }
    \label{fig_1007}
\end{figure}
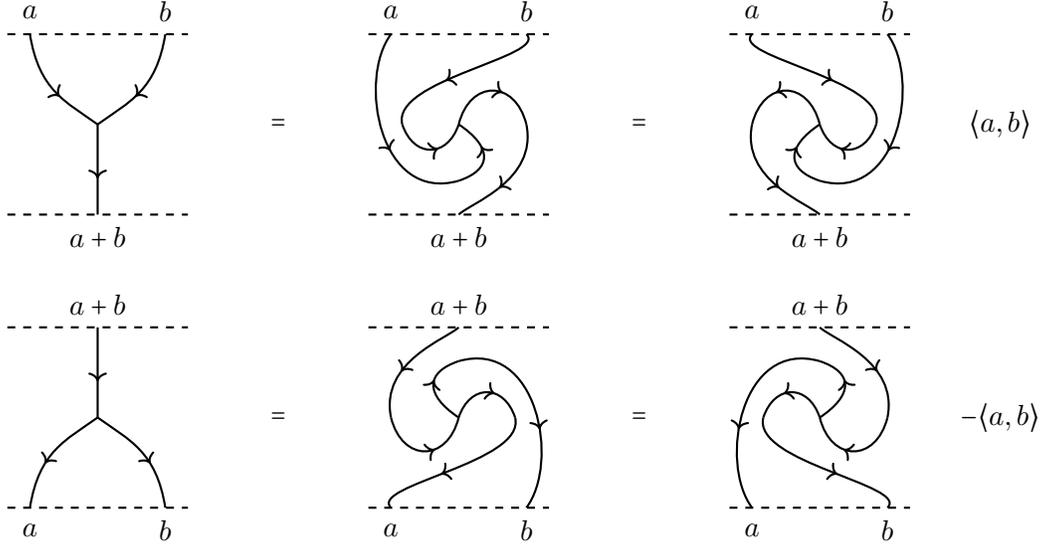

When additive lines are oriented downwards, we keep the vertices to stay the same but rotate the additive vertex clockwise or counterclockwise, resulting in a contribution of $\langle a,b\rangle$ or $-\langle a,b\rangle$. See Figure~\ref{fig_1007}.

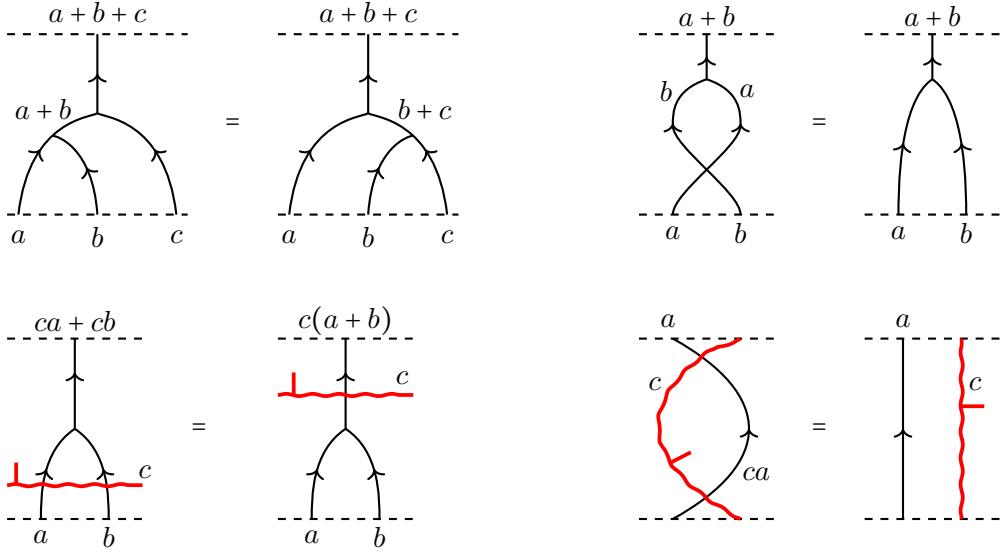
\begin{figure}
\begin{center}
\begin{tikzpicture}[scale=0.6,decoration={
    markings,
    mark=at position 0.5 with {\arrow{>}}}]
\begin{scope}[shift={(0,0)}]

\draw[thick,dashed] (0,4) -- (4,4);

\draw[thick,dashed] (0,0) -- (4,0);

\draw[thick,postaction={decorate}] (0.25,0) .. controls (0.5,2.15) and (1.9,2.15) .. (2,2.25);

\draw[thick,postaction={decorate}] (3.75,0) .. controls (3.5,2.15) and (2.1,2.15) .. (2,2.25);

\draw[thick,postaction={decorate}] (2,0) .. controls (1.9,1.5) and (1.1,1.75) .. (1.0,1.75);

\draw[thick,postaction={decorate}] (2,2.25) -- (2,4);

\node at (2,4.5) {$a+b+c$};

\node at (0.8,2.35) {$a+b$};

\node at (0.25,-0.5) {$a$};

\node at (2,-0.5) {$b$};

\node at (3.75,-0.5) {$c$};

\node at (5,2) {$=$};
\end{scope}

\begin{scope}[shift={(6,0)}]

\draw[thick,dashed] (0,4) -- (4,4);

\draw[thick,dashed] (0,0) -- (4,0);

\draw[thick,postaction={decorate}] (0.25,0) .. controls (0.5,2.15) and (1.9,2.15) .. (2,2.25);

\draw[thick,postaction={decorate}] (3.75,0) .. controls (3.5,2.15) and (2.1,2.15) .. (2,2.25);

\draw[thick,postaction={decorate}] (2,0) .. controls (2.1,1.5) and (2.9,1.75) .. (3,1.75);

\draw[thick,postaction={decorate}] (2,2.25) -- (2,4);

\node at (2,4.5) {$a+b+c$};

\node at (3.25,2.35) {$b+c$};

\node at (0.25,-0.5) {$a$};

\node at (2,-0.5) {$b$};

\node at (3.75,-0.5) {$c$};
\end{scope}


\begin{scope}[shift={(14,0)}]

\draw[thick,dashed] (0,4) -- (3,4);

\draw[thick,dashed] (0,0) -- (3,0);

\draw[thick,->] (0.75,0) .. controls (0.75,0.5) and (2.25,1.5) .. (2.25,2.02);

\draw[thick,->] (2.25,0) .. controls (2.25,0.5) and (0.75,1.5) .. (0.75,2.02);

\draw[thick] (0.75,2) .. controls (0.75,2.25) and (0.75,2.75) .. (1.5,3);

\draw[thick] (2.25,2) .. controls (2.25,2.25) and (2.25,2.75) .. (1.5,3);

\draw[thick,postaction={decorate}] (1.5,3) -- (1.5,4);

\node at (0.6,2.75) {$b$};

\node at (2.4,2.75) {$a$};

\node at (1.5,4.4) {$a+b$};

\node at (0.75,-0.4) {$a$};

\node at (2.25,-0.4) {$b$};

\node at (4,2) {$=$};

\end{scope}

\begin{scope}[shift={(19,0)}]

\draw[thick,dashed] (0,4) -- (3,4);

\draw[thick,dashed] (0,0) -- (3,0);

\draw[thick,postaction={decorate}] (0.75,0) .. controls (0.75,0.5) and (0.75,2.5) .. (1.5,3);

\draw[thick,postaction={decorate}] (2.25,0) .. controls (2.25,0.5) and (2.25,2.5) .. (1.5,3);

\draw[thick,postaction={decorate}] (1.5,3) -- (1.5,4);

\node at (1.5,4.4) {$a+b$};

\node at (0.75,-0.4) {$a$};

\node at (2.25,-0.4) {$b$};
\end{scope}


\begin{scope}[shift={(0,-6.75)}]

\draw[thick,dashed] (0,4) -- (3,4);
\draw[thick,dashed] (0,0) -- (3,0);

\draw[thick,postaction={decorate}] (0.75,0) .. controls (0.75,0.5) and (0.75,1.5) .. (1.5,2);

\draw[thick,postaction={decorate}] (2.25,0) .. controls (2.25,0.5) and (2.25,1.5) .. (1.5,2);

\draw[thick,->] (1.5,2) -- (1.5,3.25);

\draw[thick] (1.5,3.20) -- (1.5,4);

\node at (1.5,4.4) {$ca + cb$};

\node at (0.75,-0.4) {$a$};

\node at (2.25,-0.4) {$b$};

\draw[line width=0.50mm,red] (0,0.75) decorate [decoration={snake,amplitude=0.15mm}] {-- (3,0.75)};
\draw[line width=0.50mm,red] (0.20,0.75) -- (0.20,1.25);

\node at (3.05,1.05) {$c$};

\node at (4.25,2) {$=$};
\end{scope}

\begin{scope}[shift={(6, -6.75)}]


\draw[thick,dashed] (0,4) -- (3,4);

\draw[thick,dashed] (0,0) -- (3,0);

\draw[thick,postaction={decorate}] (0.75,0) .. controls (0.75,0.5) and (0.75,1.5) .. (1.5,2);

\draw[thick,postaction={decorate}] (2.25,0) .. controls (2.25,0.5) and (2.25,1.5) .. (1.5,2);

\draw[thick,->] (1.5,2) -- (1.5,3.25);

\draw[thick] (1.5,3.20) -- (1.5,4);

\node at (1.5,4.4) {$c(a+b)$};

\node at (0.75,-0.4) {$a$};

\node at (2.25,-0.4) {$b$};

\draw[line width=0.50mm,red] (0,2.75) decorate [decoration={snake,amplitude=0.15mm}] {-- (3,2.75)};

\draw[line width=0.50mm,red] (0.35,2.75) -- (0.35,3.25);

\node at (2.75,3.13) {$c$};

\end{scope}


\begin{scope}[shift={(14,-6.75)}]

\draw[thick,dashed] (0,4) -- (3,4);

\draw[thick,dashed] (0,0) -- (3,0);

\draw[thick,postaction={decorate}] (0.75,0) .. controls (3,1.25) and (3,2.75) .. (0.75,4);

\draw[line width=0.50mm,red] (2.25,0) decorate [decoration={snake,amplitude=0.15mm}] {.. controls (0,1.25) and (0,2.75) .. (2.25,4)};

\draw[line width=0.50mm,red] (0.7,1.25) -- (1.15,1.48);
 
\node at (0.65,4.40) {$a$};

\node at (0.35,3) {$c$};

\node at (2.6,1) {$ca$};

\node at (4,2) {$=$};
\end{scope}

\begin{scope}[shift={(19, -6.75)}]

\draw[thick,dashed] (0,4) -- (3,4);

\draw[thick,dashed] (0,0) -- (3,0);

\draw[thick,postaction={decorate}] (0.85,0) --  (0.85,4);

\draw[line width=0.50mm,red] (2.15,0) decorate [decoration={snake,amplitude=0.15mm}] {-- (2.15,4)};

\draw[line width=0.50mm,red]  (2.15,2.5) -- (2.65,2.5);

\node at (0.85,4.4) {$a$};

\node at (2.45,3) {$c$};

\end{scope}


\end{tikzpicture}
\end{center}
    \caption{Top left: This diagram represents the $2$-cocycle condition~\eqref{item:Cath_cocycle} in Cathelineau's $\kk$-vector space.
    Top right: Two lines crossing is virtual, but the additive vertex contributes $\langle b,a\rangle = \langle a,b \rangle$. Bottom left: the additive vertex on the left hand side gives $\langle ac, bc\rangle$ while the additive vertex on the right hand side gives $c\langle a,b\rangle$. This cobordism implies that they are equal.
    Bottom right: These two isotopies imply that one may pull the red multiplicative line away from black additive line.
    }
    \label{fig_1002}
\end{figure}

Additive and multiplicative lines satisfy natural isotopy relations, as well as the Cathelineau's $2$-cocycle relation in \eqref{item:Cath_cocycle}. The relation $\brak{a,b+c} + \brak{b,c}   = \brak{a+b,c} + \brak{a,b}$ is satisfied by Figure~\ref{fig_1002} top left by summing over the additive vertices. Figure~\ref{fig_1002} top right satisfies the symmetry relation $\langle a,b \rangle = \langle b,a \rangle$. Figure~\ref{fig_1002} bottom left satisfies the scaling relation $\langle ca, cb \rangle = c\langle a,b\rangle$ since the additive vertex on the left-hand side of the diagram contributes $\langle ca, cb \rangle$ since the additive lines have weights $ca$ and $cb$ while the additive vertex on the right-hand side of the diagram contributes $\langle a,b\rangle$, which we then rescale by the scalar $c$. For Figure~\ref{fig_1002} bottom right, the diagrams show that we can pull apart additive and multiplicative lines that have crossed virtually.

%

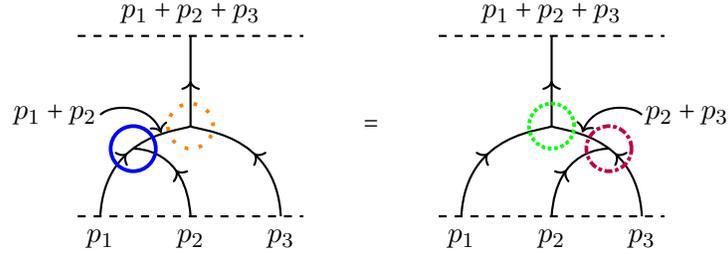
\begin{figure}
    \centering
\begin{tikzpicture}[scale=0.6,decoration={
    markings,
    mark=at position 0.5 with {\arrow{>}}}]
\begin{scope}[shift={(0,0)}]
\draw[thick,dashed] (0,4) -- (5,4);
\draw[thick,dashed] (0,0) -- (5,0);

\draw[thick,postaction={decorate}] (0.5,0) .. controls (0.6,1.9) and (2.4,1.9) .. (2.5,2);
\draw[thick,postaction={decorate}] (4.5,0) .. controls (4.4,1.9) and (2.6,1.9) .. (2.5,2);
\draw[thick,postaction={decorate}] (2.5,0) .. controls (2.4,1.5) and (1.33,1.5) .. (1.23,1.5);

\draw[thick,postaction={decorate}] (2.5,2) -- (2.5,4);

\node at (2.5,4.5) {$p_1+p_2+p_3$};
\node at (-0.50,2.25) {$p_1+p_2$};

\node at (0.5,-0.5) {$p_1$};
\node at (2.5,-0.5) {$p_2$};
\node at (4.5,-0.5) {$p_3$};

\draw[line width=0.05cm,blue] (1.73,1.5) arc (0:360:0.5);

\draw[line width=0.05cm,loosely dotted,orange] (3,2) arc (0:360:0.5);

\draw[thick,->] (0.5,2.25) .. controls (0.75,2.5) and (1.6,2.5) .. (1.85,1.85);

\node at (6.5,2) {$=$};
\end{scope}

\begin{scope}[shift={(8,0)}]
\draw[thick,dashed] (0,4) -- (5,4);
\draw[thick,dashed] (0,0) -- (5,0);

\draw[thick,postaction={decorate}] (0.5,0) .. controls (0.6,1.9) and (2.4,1.9) .. (2.5,2);
\draw[thick,postaction={decorate}] (4.5,0) .. controls (4.4,1.9) and (2.6,1.9) .. (2.5,2);
\draw[thick,postaction={decorate}] (2.5,0) .. controls (2.6,1.5) and (3.67,1.5) .. (3.77,1.5);

\draw[thick,postaction={decorate}] (2.5,2) -- (2.5,4);

\draw[line width=0.05cm,densely dashdotted,purple] (4.27,1.5) arc (0:360:0.5);

\draw[line width=0.05cm,densely dotted,green] (3,2) arc (0:360:0.5);

\draw[thick,->] (4.5,2.25) .. controls (4.25,2.5) and (3.43,2.5) .. (3.18,1.86);

\node at (2.5,4.5) {$p_1+p_2+p_3$};
\node at (5.5,2.25) {$p_2+p_3$};

\node at (0.5,-0.5) {$p_1$};
\node at (2.5,-0.5) {$p_2$};
\node at (4.5,-0.5) {$p_3$};
\end{scope}

\end{tikzpicture}
    \caption{When $\mathbf{k}=\mathbb{R}$ and $p_1+p_2 +p_3 =1$, each diagram on the left and the right evaluates to Shannon entropy.}
    \label{fig8_001}
\end{figure}

Furthermore, in Figure~\ref{fig_1002} top left, if the values $a,b,c$ sum to $1$, i.e., $a+ b+ c = 1$, then we get Shannon entropy if all the lines in a diagram merge and we sum over all the additive vertex contributions. That is, consider Figure~\ref{fig8_001}.
Although Proposition~\ref{prop:strands_merging_entropy} holds for $n$ strands merging into one, we first prove the case for three additive lines. 

\begin{proposition}
\label{prop:strands_merging_entropy}
Each side in the equality in Figure~\ref{fig8_001} evaluates to Shannon entropy in \eqref{eq_shannon}.     
\end{proposition}

\begin{proof}
Recall from~\cite{IK24_dilogarithms_entropy} that a relation between diagrammatics and scaled entropy is \eqref{eqn_diagrammatic_entropy}.
The additive vertex in the blue circle gives a contribution of $\langle p_1, p_2 \rangle$.
The additive vertex in the dotted orange circle gives a contribution of 
$\langle p_1 + p_2, p_3 \rangle$. 
So the sum of the two additive vertices gives $\langle p_1, p_2 \rangle + \langle p_1 + p_2, p_3 \rangle$. 

On the other hand, the additive vertex in the purple dashed-dotted circle provides $\langle p_2 , p_3 \rangle$. 
The additive vertex in the green dotted circle contributes $\langle p_1, p_2+ p_3 \rangle$. 
Combining the two additive vertices on the right hand side, we have $\langle p_2 , p_3 \rangle + \langle p_1, p_2+ p_3 \rangle$. 
We see that 
\[ 
\langle p_1, p_2 \rangle + \langle p_1 + p_2, p_3 \rangle = \langle p_2 , p_3 \rangle + \langle p_1, p_2+ p_3 \rangle,
\] 
which is the $2$-cocycle condition for Cathelineau's vector space. So by the proof of Lemma~\ref{lemma_Cath_properties}, each side sums to $H(p_X)$. 
\end{proof}

Now consider Figure~\ref{fig_1005}.

\begin{proposition}
\label{proposition:entropy_4_strands}
Figure~\ref{fig_1005} corresponds to entropy $H(p_X)=  -p_1 \log |p_1| - p_2 \log |p_2| - p_3 \log | p_3| -p_4 \log |p_4|$. 
\end{proposition}

\begin{proof}
Here, we will compute only the left hand side in Figure~\ref{fig_1005}, where $-p_1 \log p_1 - p_2 \log p_2 - p_3 \log (p_3) - p_4 \log p_4=H(p_X)$. We will leave the other diagrams as an exercise for the reader. 
From bottom to top, the additive vertices contribute 
\begin{align*}
\langle p_1, p_2 \rangle 
&= (p_1+ p_2) H\left( \frac{p_1}{p_1 + p_2} \right) \\ 
\langle p_1 + p_2, p_3 \rangle &= (p_1+ p_2+p_3) H\left( \frac{p_1+p_2}{p_1 + p_2 + p_3} \right) \\ 
\langle p_1 + p_2 +p_3, p_4 \rangle &= (p_1+ p_2+p_3+ p_4) H\left( \frac{p_1+p_2+ p_3}{p_1 + p_2 + p_3 + p_4} \right),
\end{align*}
respectively. 
Combining the three equations above gives
\begin{align*}
\langle p_1, & \: p_2 \rangle + 
\langle p_1 + p_2, p_3 \rangle  + 
\langle p_1 + p_2 +p_3, p_4 \rangle \\
&= (p_1+ p_2) H\left( \frac{p_1}{p_1 + p_2} \right) + 
(p_1+ p_2+p_3) H\left( \frac{p_1+p_2}{p_1 + p_2 + p_3} \right) \\  
&\hspace{4mm}+ 
(p_1+ p_2+p_3+ p_4) H\left( \frac{p_1+p_2+ p_3}{p_1 + p_2 + p_3 + p_4} \right) \\
&= (p_1 + p_2) \left(-\frac{p_1}{p_1 + p_2} \log \left| \frac{p_1}{p_1 + p_2} \right|
- \left(1 - \frac{p_1}{p_1 + p_2}\right) 
\log \left| 1 - \frac{p_1}{p_1 + p_2}\right| \right) \\
&\hspace{4mm}+ (p_1 + p_2 + 
    p_3) \left(-\frac{p_1 + p_2}{p_1 + p_2 + p_3} \log \left| \frac{p_1 + p_2}{
      p_1 + p_2 + p_3}\right| \right. \\
&\hspace{4mm}\left. - \left(1 - \frac{p_1 + p_2}{p_1 + p_2 + p_3}\right) \log \left|
      1 - \frac{p_1 + p_2}{p_1 + p_2 + p_3}\right|
      \right) \\
&\hspace{4mm}+ (p_1 + p_2 + p_3 + 
    p_4) \left(-\frac{p_1 + p_2 + p_3}{p_1 + p_2 + p_3 + p_4} \log \left| \frac{p_1 + p_2 + p_3}{p_1 + p_2 + p_3 + p_4} \right| \right. \\
&\hspace{4mm}\left. - \left(1 - \frac{p_1 + p_2 + p_3}{p_1 + p_2 + p_3 + p_4}\right) \log \left|
      1 - \frac{p_1 + p_2 + p_3}{p_1 + p_2 + p_3 + p_4}\right|
      \right) \\ 
&= \cancel{(p_1 + p_2)} \left(-\frac{p_1}{\cancel{p_1 + p_2}} \log \left| \frac{p_1}{p_1 + p_2} \right|
-  \frac{p_2}{\cancel{p_1 + p_2}} 
\log \left| \frac{p_2}{p_1 + p_2}\right| \right) \\
&\hspace{4mm}+ \cancel{(p_1 + p_2 + 
    p_3)} \left(-\frac{p_1 + p_2}{\cancel{p_1 + p_2 + p_3}} \log \left| \frac{p_1 + p_2}{
      p_1 + p_2 + p_3}\right| \right. \\
&\hspace{4mm}\left. 
- \frac{p_3}{\cancel{p_1 + p_2 + p_3}} \log \left|
       \frac{p_3}{p_1 + p_2 + p_3}\right|
      \right) \\
&\hspace{4mm}+ \cancel{(p_1 + p_2 + p_3 + 
    p_4)} \left(-\frac{p_1 + p_2 + p_3}{\cancel{p_1 + p_2 + p_3 + p_4}} \log \left| \frac{p_1 + p_2 + p_3}{p_1 + p_2 + p_3 + p_4} \right| \right. \\
&\hspace{4mm}\left. - \frac{p_4}{\cancel{p_1 + p_2 + p_3 + p_4}} \log \left|
      \frac{p_4}{p_1 + p_2 + p_3 + p_4}\right|
      \right) \\ 
&=  -p_1\log \left| \frac{p_1}{p_1 + p_2} \right| - p_2  
\log \left| \frac{p_2}{p_1 + p_2}\right| 
- (p_1 + p_2) 
\log \left| \frac{p_1 + p_2}{
 p_1 + p_2 + p_3}\right|  
- p_3 \log \left|
 \frac{p_3}{p_1 + p_2 + p_3}\right| \\
&\hspace{4mm} - (p_1 + p_2 + p_3) \log \left| \frac{p_1 + p_2 + p_3}{p_1 + p_2 + p_3 + p_4} \right|   - p_4 \log \left|
      \frac{p_4}{p_1 + p_2 + p_3 + p_4}\right| \\ 
&= -p_1 \log |p_1| - p_2 \log |p_2| 
+ \cancel{(p_1 + p_2) \log |p_1 + p_2|}
- \cancel{(p_1 + p_2) \log |p_1 + p_2|} 
- p_3 \log |p_3| \\
&\hspace{4mm}
+ \cancel{(p_1 + p_2 + p_3) \log |p_1 + p_2 + p_3|}
- \cancel{(p_1 + p_2 + p_3) \log |p_1 + p_2 + p_3|} \\ 
&\hspace{4mm}- 
 p_4 \log |p_4| + (\underbrace{p_1 + p_2 + p_3 + p_4}_{1}) \log |\underbrace{p_1 + p_2 + p_3 + p_4}_{1}| \\
&= -p_1 \log |p_1| - p_2 \log |p_2| - p_3 \log |p_3| - p_4 \log |p_4| = H(p_X). 
\end{align*}
This concludes the proof.
\end{proof}

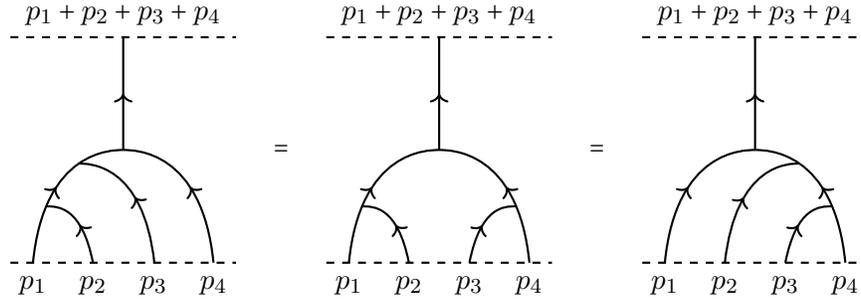
\begin{figure}
    \centering
\begin{tikzpicture}[scale=0.6,decoration={
    markings,
    mark=at position 0.5 with {\arrow{>}}}]


\begin{scope}[shift={(0,0)}]

\draw[thick,dashed] (0,5) -- (5,5); 

\draw[thick,dashed] (0,0) -- (5,0); 

\draw[thick,postaction={decorate}] (0.5,0) .. controls (0.75,2.5) and (2.25,2.5) .. (2.5,2.5);

\draw[thick,postaction={decorate}] (4.5,0) ..  controls (4.25,2.5) and (2.75,2.5) .. (2.5,2.5);

\draw[thick,postaction={decorate}] (1.83,0) .. controls (1.58,1.25) and (1.05,1.25) .. (0.8,1.25);

\draw[thick,postaction={decorate}] (3.19,0) .. controls (2.94,2.2) and (1.75,2.2) .. (1.5,2.2);

\draw[thick,postaction={decorate}] (2.5,2.5) -- (2.5,5);

\node at (0.5,-0.5) {$p_1$};

\node at (1.83,-0.5) {$p_2$};

\node at (3.17,-0.5) {$p_3$};

\node at (4.5,-0.5) {$p_4$};

\node at (2.5,5.5) {$p_1 + p_2 + p_3 + p_4$};

\node at (6,2.5) {$=$};

\end{scope}


\begin{scope}[shift={(7,0)}]

\draw[thick,dashed] (0,5) -- (5,5); 

\draw[thick,dashed] (0,0) -- (5,0); 

\draw[thick,postaction={decorate}] (0.5,0) .. controls (0.75,2.5) and (2.25,2.5) .. (2.5,2.5);

\draw[thick,postaction={decorate}] (4.5,0) ..  controls (4.25,2.5) and (2.75,2.5) .. (2.5,2.5);

\draw[thick,postaction={decorate}] (1.83,0) .. controls (1.58,1.25) and (1.05,1.25) .. (0.8,1.25);

\draw[thick,postaction={decorate}] (3.17,0) .. controls (3.42,1.25) and (3.94,1.25) .. (4.19,1.25);

\draw[thick,postaction={decorate}] (2.5,2.5) -- (2.5,5);

\node at (0.5,-0.5) {$p_1$};

\node at (1.83,-0.5) {$p_2$};

\node at (3.17,-0.5) {$p_3$};

\node at (4.5,-0.5) {$p_4$};

\node at (2.5,5.5) {$p_1 + p_2 + p_3 + p_4$};

\node at (6,2.5) {$=$};

\end{scope}


\begin{scope}[shift={(14,0)}]

\draw[thick,dashed] (0,5) -- (5,5); 

\draw[thick,dashed] (0,0) -- (5,0); 

\draw[thick,postaction={decorate}] (0.5,0) .. controls (0.75,2.5) and (2.25,2.5) .. (2.5,2.5);

\draw[thick,postaction={decorate}] (4.5,0) ..  controls (4.25,2.5) and (2.75,2.5) .. (2.5,2.5);

\draw[thick,postaction={decorate}] (1.83,0) .. controls (2.08,2.2) and (3.24,2.2) .. (3.49,2.2);

\draw[thick,postaction={decorate}] (3.17,0) .. controls (3.42,1.25) and (3.94,1.25) .. (4.19,1.25);

\draw[thick,postaction={decorate}] (2.5,2.5) -- (2.5,5);

\node at (0.5,-0.5) {$p_1$};

\node at (1.83,-0.5) {$p_2$};

\node at (3.17,-0.5) {$p_3$};

\node at (4.5,-0.5) {$p_4$};

\node at (2.5,5.5) {$p_1 + p_2 + p_3 + p_4$};

\end{scope}


\end{tikzpicture}
    \caption{All possible ways of network of additive lines merging into one additive line.}
    \label{fig_1005}
\end{figure}

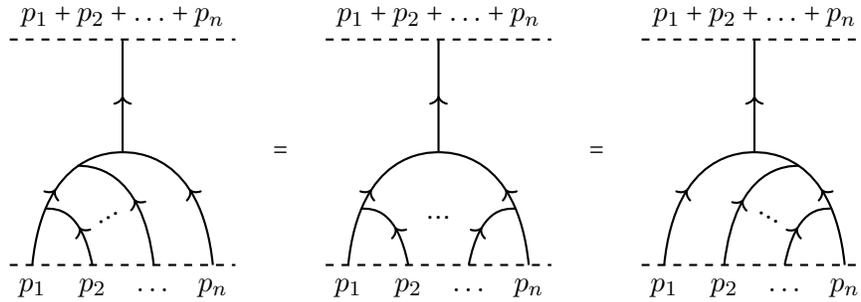
\begin{figure}
    \centering
\begin{tikzpicture}[scale=0.6,decoration={
    markings,
    mark=at position 0.5 with {\arrow{>}}}]


\begin{scope}[shift={(0,0)}]

\draw[thick,dashed] (0,5) -- (5,5); 

\draw[thick,dashed] (0,0) -- (5,0); 

\draw[thick,postaction={decorate}] (0.5,0) .. controls (0.75,2.5) and (2.25,2.5) .. (2.5,2.5);

\draw[thick,postaction={decorate}] (4.5,0) ..  controls (4.25,2.5) and (2.75,2.5) .. (2.5,2.5);

\draw[thick,postaction={decorate}] (1.83,0) .. controls (1.58,1.25) and (1.05,1.25) .. (0.8,1.25);

\draw[thick,postaction={decorate}] (3.19,0) .. controls (2.94,2.2) and (1.75,2.2) .. (1.5,2.2);

\draw[thick,postaction={decorate}] (2.5,2.5) -- (2.5,5);

\node at (2.22,1) {\rotatebox[origin=c]{30}{$\cdots$}};

\node at (0.5,-0.5) {$p_1$};

\node at (1.83,-0.5) {$p_2$};

\node at (3.17,-0.55) {$\ldots$};

\node at (4.5,-0.5) {$p_n$};

\node at (2.5,5.5) {$p_1 + p_2 + \ldots + p_n$};

\node at (6,2.5) {$=$};

\end{scope}


\begin{scope}[shift={(7,0)}]

\draw[thick,dashed] (0,5) -- (5,5); 

\draw[thick,dashed] (0,0) -- (5,0); 

\draw[thick,postaction={decorate}] (0.5,0) .. controls (0.75,2.5) and (2.25,2.5) .. (2.5,2.5);

\draw[thick,postaction={decorate}] (4.5,0) ..  controls (4.25,2.5) and (2.75,2.5) .. (2.5,2.5);

\draw[thick,postaction={decorate}] (1.83,0) .. controls (1.58,1.25) and (1.05,1.25) .. (0.8,1.25);

\draw[thick,postaction={decorate}] (3.17,0) .. controls (3.42,1.25) and (3.94,1.25) .. (4.19,1.25);

\draw[thick,postaction={decorate}] (2.5,2.5) -- (2.5,5);

\node at (2.5,1) {$\cdots$};

\node at (0.5,-0.5) {$p_1$};

\node at (1.83,-0.5) {$p_2$};

\node at (3.17,-0.55) {$\ldots$};

\node at (4.5,-0.5) {$p_n$};

\node at (2.5,5.5) {$p_1 + p_2 + \ldots + p_n$};

\node at (6,2.5) {$=$};

\end{scope}


\begin{scope}[shift={(14,0)}]

\draw[thick,dashed] (0,5) -- (5,5); 

\draw[thick,dashed] (0,0) -- (5,0); 

\draw[thick,postaction={decorate}] (0.5,0) .. controls (0.75,2.5) and (2.25,2.5) .. (2.5,2.5);

\draw[thick,postaction={decorate}] (4.5,0) ..  controls (4.25,2.5) and (2.75,2.5) .. (2.5,2.5);

\draw[thick,postaction={decorate}] (1.83,0) .. controls (2.08,2.2) and (3.24,2.2) .. (3.49,2.2);

\draw[thick,postaction={decorate}] (3.17,0) .. controls (3.42,1.25) and (3.94,1.25) .. (4.19,1.25);

\draw[thick,postaction={decorate}] (2.5,2.5) -- (2.5,5);

\node at (2.78,1) {\rotatebox[origin=c]{-30}{$\cdots$}};

\node at (0.5,-0.5) {$p_1$};

\node at (1.83,-0.5) {$p_2$};

\node at (3.17,-0.55) {$\ldots$};

\node at (4.5,-0.5) {$p_n$};

\node at (2.5,5.5) {$p_1 + p_2 + \ldots + p_n$};

\end{scope}


\end{tikzpicture}
    \caption{Various diagrams depicting entropy.}
    \label{fig_1006}
\end{figure}

\begin{proposition}
\label{prop_entropy_general_case_n}
For each diagram in Figure~\ref{fig_1006}, the additive vertices sum to Shannon entropy. 
\end{proposition}

The proof of Proposition~\ref{prop_entropy_general_case_n} is analogous to Propositions~\ref{prop:strands_merging_entropy} and \ref{proposition:entropy_4_strands}.

Now, recall equation \eqref{eqn_diagrammatic_entropy} and Cathelineau's symmetry relation~\eqref{item:symm}. They give us the following correspondence: 
\begin{equation}
\label{entropy_special_case_n2}
\langle a,b\rangle = \langle b,a\rangle  \Leftrightarrow (a+b)H\left(\frac{a}{a+b}\right) = (a+b)H\left(\frac{b}{a+b}\right).
\end{equation}

In particular, in \eqref{entropy_special_case_n2}, take $a=p$ and $b = 1-p$ to obtain 
\[
\langle p,1-p\rangle = \langle 1-p,p\rangle  \Leftrightarrow H(p) = H(1-p).
\]

\begin{figure}
    \centering
\begin{tikzpicture}[scale=0.6,decoration={
    markings,
    mark=at position 0.50 with {\arrow{>}}}]

\begin{scope}[shift={(0,0)}]

\draw[thick,dashed] (0,4) -- (4,4);

\draw[thick,dashed] (0,0) -- (4,0);

\node at (0.5,4.5) {$x_1$};

\node at (3.5,4.5) {$x_2$};

\node at (0.5,-0.5) {$x_1$};

\node at (3.5,-0.5) {$x_2$};

\draw[thick,postaction={decorate}] (0.5,0) -- (2,1.5);

\draw[thick,postaction={decorate}] (3.5,0) -- (2,1.5);

\draw[thick,postaction={decorate}] (2,1.5) -- (2,2.5);

\draw[thick,postaction={decorate}] (2,2.5) -- (0.5,4);

\draw[thick,postaction={decorate}] (2,2.5) -- (3.5,4);

\end{scope}


\begin{scope}[shift={(7,0)}]

\draw[thick,dashed] (0,4) -- (4,4);

\draw[thick,dashed] (0,0) -- (4,0);

\draw[thick,postaction={decorate}] (2,1.5) -- (2,2.5);

\node at (2.5,2) {$0$};

\node at (0.5,4.5) {$q$};

\node at (3.5,4.5) {$-q$};

\node at (0.5,-0.5) {$p$};

\node at (3.5,-0.5) {$-p$};

\draw[thick,postaction={decorate}] (0.5,0) -- (2,1.5);

\draw[thick,postaction={decorate}] (3.5,0) -- (2,1.5);

\draw[thick,postaction={decorate}] (2,2.5) -- (0.5,4);

\draw[thick,postaction={decorate}] (2,2.5) -- (3.5,4);

\node at (5,2) {$=$};

\end{scope}


\begin{scope}[shift={(13,0)}]

\draw[thick,dashed] (0,4) -- (4,4);

\draw[thick,dashed] (0,0) -- (4,0);

\node at (0.5,4.5) {$q$};

\node at (3.5,4.5) {$-q$};

\node at (0.5,-0.5) {$p$};

\node at (3.5,-0.5) {$-p$};

\draw[thick,postaction={decorate}] (0.5,0) -- (2,1.5);

\draw[thick,postaction={decorate}] (3.5,0) -- (2,1.5);

\draw[thick,postaction={decorate}] (2,2.5) -- (0.5,4);

\draw[thick,postaction={decorate}] (2,2.5) -- (3.5,4);

\end{scope}


\begin{scope}[shift={(20,0)}]

\draw[thick,dashed] (0,4) -- (4,4);

\draw[thick,dashed] (0,0) -- (4,0);


\draw[thick,->] (0.5,0) .. controls (0.6,1.25) and (0.9,1.5) .. (1,1.65);

\draw[thick,fill] (1.35,1.8) arc (0:360:1.5mm);

\draw[thick,postaction={decorate}] (3.5,0) .. controls (3.4,2.75) and (1.45,2.25) .. (1.35,2);

\node at (0.5,-0.5) {$p$};

\node at (3.5,-0.5) {$-p$};

\draw[thick,dotted,postaction={decorate}] (1.2,1.8) .. controls (1.1,2) and (0.6,2.5) .. (0.5,4);

\node at (0.5,4.5) {$0$};

\node at (0.25,3) {$0$};

\end{scope}

\end{tikzpicture}
    \caption{Left: the evaluation of this cobordism is $\langle x_1, x_2\rangle - \langle x_1, x_2\rangle=0$. So the entropy of this system is $0$. Middle: $0$-lines can be erased. Right: the $0$-line can be extended from a $\langle p,-p\rangle$-line.}
    \label{fig7_003}
\end{figure}
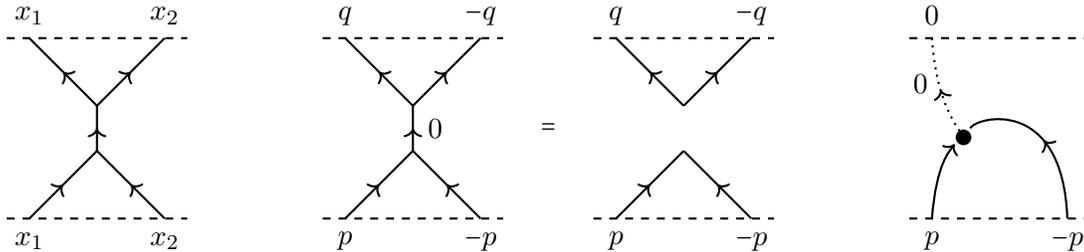

\begin{figure}
    \centering
\begin{tikzpicture}[scale=0.6,decoration={
    markings,
    mark=at position 0.50 with {\arrow{>}}}]


\begin{scope}[shift={(0,0)}]

\draw[thick,dashed] (0.5,4) -- (2.5,4);

\draw[thick,dashed] (0.5,0) -- (2.5,0);

\node at (1.5,4.5) {$-p$};

\node at (1.5,-0.5) {$p$};

\draw[thick,postaction={decorate}] (1.5,4) -- (1.5,2);

\draw[thick,fill] (1.65,2) arc (0:360:1.5mm);

\draw[thick,postaction={decorate}] (1.5,0) -- (1.5,2);

\node at (3.5,2) {$=$};

\end{scope}


\begin{scope}[shift={(4.5,0)}]

\draw[thick,dashed] (0,4) -- (3,4);

\draw[thick,dashed] (0,0) -- (3,0);

\node at (0.5,4.5) {$-p$};

\node at (0.5,-0.5) {$p$};

\node at (2.5,4.5) {$0$};

\node at (2.25,2.55) {$0$};

\draw[thick,postaction={decorate}] (0.5,4) -- (0.5,2);

\draw[thick,fill] (0.65,2) arc (0:360:1.5mm);

\draw[thick,postaction={decorate}] (0.5,2) .. controls (0.75,2) and (2.25,2.5) .. (2.5,4);

\draw[thick,postaction={decorate}] (0.5,0) -- (0.5,2);

\end{scope}


\begin{scope}[shift={(12,0)}]

\draw[thick,dashed] (0,4) -- (4,4);

\draw[thick,dashed] (0,0) -- (4,0);

\node at (0.5,-0.5) {$p$};

\node at (3.5,-0.5) {$-p$};

\draw[thick,postaction={decorate}] (0.5,0) -- (2,2.25);

\draw[thick,postaction={decorate}] (3.5,0) -- (2,2.25);

\draw[thick,postaction={decorate}] (2,2.25) -- (2,4);

\node at (2,4.5) {$0$};

\node at (2.5,3) {$0$};

\node at (5,2) {$=$}; 
\end{scope}


\begin{scope}[shift={(19,0)}]

\draw[thick,dashed] (0,4) -- (4,4);

\draw[thick,dashed] (0,0) -- (4,0);

\draw[thick,fill] (2.15,2.25) arc (0:360:1.5mm);

\node at (0.5,-0.5) {$p$};

\node at (3.5,-0.5) {$-p$};

\draw[thick,postaction={decorate}] (0.5,0) -- (2,2.25);

\draw[thick,postaction={decorate}] (3.5,0) -- (2,2.25);

\end{scope}


\end{tikzpicture}
    \caption{We introduce a dot to represent reversal of orientation. Left: we can erase the $0$-line. The 0-line can go in or out, to the top or bottom boundary.
    Right: we have $\langle p, -p\rangle =0$.}
    \label{fig7_004}
\end{figure}
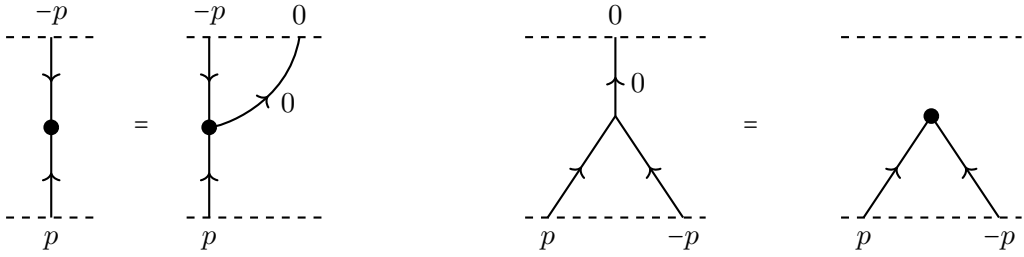

\begin{figure}
    \centering
\begin{tikzpicture}[scale=0.6,decoration={
    markings,
    mark=at position 0.50 with {\arrow{>}}}]


\begin{scope}[shift={(0,0)}]

\draw[thick,dashed] (0,3) -- (4,3);

\draw[thick,dashed] (0,0) -- (4,0);

\draw[thick,fill] (2.15,2.25) arc (0:360:1.5mm);

\node at (0.5,-0.5) {$p$};

\node at (3.5,-0.5) {$-p$};

\draw[thick,postaction={decorate}] (0.5,0) .. controls (0.75,2) and (1.75,2.25) .. (2,2.25);

\draw[thick,postaction={decorate}] (3.5,0) .. controls (3.25,2) and (2.25,2.25) .. (2,2.25);

\node at (5,2) {$=$}; 

\end{scope}


\begin{scope}[shift={(7,0)}]

\draw[thick,dashed] (0,6) -- (4,6);

\draw[thick,dashed] (0,3) -- (4,3);

\draw[thick,dashed] (0,0) -- (4,0);

\node at (0.5,-0.5) {$p$};

\node at (3.5,-0.5) {$-p$};

\draw[thick,postaction={decorate}] (0.5,3) -- (0.5,1.5);

\draw[thick,postaction={decorate}] (0.5,0) -- (0.5,1.5);

\draw[thick,fill] (0.65,1.5) arc (0:360:1.5mm);

\draw[thick,postaction={decorate}] (3.5,0) -- (3.5,3);

\draw[thick,postaction={decorate}] (3.5,3) .. controls (3.25,5) and (0.75,5) .. (0.5,3);

\end{scope}


\end{tikzpicture}
    \caption{We introduce a dot to represent reversal of orientation. Left: the $0$-line has been erased.  
    Right: we have $\langle p, -p\rangle =0$.}
    \label{fig8_004}
\end{figure}
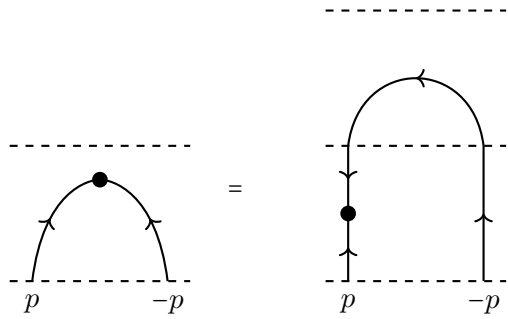

We may also introduce the $0$-line. See Figures~\ref{fig7_003}, \ref{fig7_004}, and \ref{fig8_004}.

We will now introduce black defects (dots) on additive networks, where whenever we insert a black dot on an additive network, we need to insert another defect on the additive network and reverse orientations on the portion of the additive network connecting these defects.
The contribution from Figure~\ref{fig8_008} left is $\langle p-q,1-p \rangle = (1-q)H\left( \frac{p-q}{1-q} \right)$. The contribution in Figure~\ref{fig8_008} center is 
$\langle q-1,p-q \rangle + \langle p-1,1-p\rangle -\langle q-1,1-q\rangle  = (p-1)H\left( \frac{q-1}{p-1} \right)$.
Finally, the contribution in Figure~\ref{fig8_008} right is $\langle p-q, q-1 \rangle + \langle 1-p, p-1\rangle -\langle 1-q,q-1\rangle  = (p-1)H\left( \frac{p-q}{p-1} \right)$.
They show that entropy also satisfies 
\[
(1-q)H\left( \frac{p-q}{1-q} \right) = (p-1)H\left( \frac{q-1}{p-1} \right) = (p-1)H\left( \frac{p-q}{p-1} \right). 
\]

\begin{figure}
    \centering
\begin{tikzpicture}[scale=0.6,decoration={
    markings,
    mark=at position 0.5 with {\arrow{>}}}]

\begin{scope}[shift={(0,0)}]


\draw[thick,dashed] (0,4) -- (4,4);

\draw[thick,dashed] (0,0) -- (4,0);

\draw[thick,postaction={decorate}] (0.75,0) .. controls (1,1.5) and (1.75,2) .. (2,2);

\draw[thick,postaction={decorate}] (3.25,0) .. controls (3,1.5) and (2.25,2) .. (2,2);

\draw[thick,postaction={decorate}] (2,2) -- (2,4);

\node at (2,4.5) {$1-q$};

\node at (0.75,-0.5) {$p-q$};

\node at (3.25,-0.5) {$1-p$};

\node at (6,2) {$=$};

\end{scope}

\begin{scope}[shift={(8,0)}]


\draw[thick,dashed] (0,4) -- (4,4);

\draw[thick,dashed] (0,0) -- (4,0);


\draw[thick,postaction={decorate}] (0.75,0) .. controls (1,1.5) and (1.75,2) .. (2,2);

\draw[thick,postaction={decorate}] (2,2) .. controls (2.52,2) and (2.9,1.25) .. (3,1);

\draw[thick,postaction={decorate}] (3.25,0) .. controls (3.2,0.25) and (3.1,0.75) .. (3.0,1);

\draw[thick,postaction={decorate}] (2,3) -- (2,4);

\draw[thick,postaction={decorate}] (2,3) -- (2,2);

\node at (1,2.35) {$q-1$};

\draw[thick,fill] (2.15,3) arc (0:360:1.5mm);

\draw[thick,fill] (3.15,1) arc (0:360:1.5mm);

\draw[thick,postaction={decorate}] (3,1) .. controls (3.5,1.25) and (3.5,3.25) .. (2,3);

\node at (3.4,2.85) {$0$};

\node at (2.1,1) {\rotatebox[origin=c]{35}{$p-1$}};

\node at (2,4.5) {$1-q$};

\node at (0.75,-0.5) {$p-q$};

\node at (3.25,-0.5) {$1-p$};

\node at (6,2) {$=$};

\end{scope}

\begin{scope}[shift={(16,0)}]

\draw[thick,dashed] (0,4) -- (4,4);

\draw[thick,dashed] (0,0) -- (4,0);

\draw[thick,postaction={decorate}] (0.5,0) .. controls (3.75,1) and (3.75,2) .. (2,2);

\draw[thin] (3.5,0) .. controls (0.25,1) and (0.25,2) .. (2,2);

\draw[thick,postaction={decorate}] (2,2) .. controls (0.45,1.93) and (0.7,1.3) .. (1.25,1);

\draw[thick,postaction={decorate}] (2,0.55) .. controls (1.75,0.65) and (1.5,0.75) .. (1.25,1);

\draw[thick,postaction={decorate}] (3.5,0) .. controls (3,0.15) and (2.5,0.35) .. (2,0.55);

\draw[thick,postaction={decorate}] (2,3) -- (2,4);

\draw[thick,postaction={decorate}] (2,3) -- (2,2);

\node at (2.9,2.55) {$q-1$};

\node at (4,1.5) {$p-q$};

\draw[thick,fill] (2.15,3) arc (0:360:1.5mm);

\draw[thick,fill] (1.4,1) arc (0:360:1.5mm);

\draw[thick,postaction={decorate}] (1.25,1) .. controls (0,1.25) and (0,3.5) .. (2,3);

\node at (0,2) {$0$};

\node at (1.90,1.45) {\rotatebox[origin=c]{-20}{$p-1$}};

\node at (2,4.5) {$1-q$};

\node at (0.5,-0.5) {$p-q$};

\node at (3.5,-0.5) {$1-p$};

\end{scope}

\end{tikzpicture}
    \caption{A local symmetry transformation around a trivalent vertex, i.e., change cyclic order at vertex and $2$ orientation reversals at a vertex and insert the $0$-line in the middle and right figures.}
    \label{fig8_008}
\end{figure}
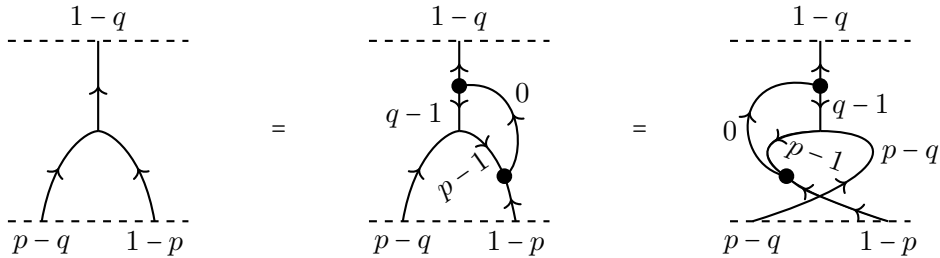

\section{Boundary wall}
\label{section_blue_line}

In this section, we introduce a boundary wall, depicted using a blue line at the base of a cobordism. Also see \cite[Section 2]{IK24_dilogarithms_entropy}. The blue horizontal line in a cobordism is called a boundary wall. It absorbs additive vertices and inserts a labeled floating dot to the left of the network. That is, if two lines weighted $a$ and $b$ merge at an additive vertex,  the vertex gets absorbed by the boundary wall and the boundary wall emits a floating defect with the label $(a+b)\left[\dfrac{a}{a+b}\right]$. See Figure~\ref{fig5_017} left. We can also absorb an additive vertex that is connected to the $0$-edge (this will give a contribution of $0$ so we can draw a vertex with weight $0$ or we do not need to draw the $0$-vertex). See Figure~\ref{fig5_017} right. 

\begin{figure}
    \centering
\begin{tikzpicture}[scale=0.6,decoration={
    markings,
    mark=at position 0.50 with {\arrow{>}}}]

\begin{scope}[shift={(0,0)}]

\draw[thick,dashed] (0,5) -- (4,5);

\draw[line width = 0.45mm, blue] (0,0) -- (4,0);

\draw[thick,postaction={decorate}] (0.25,0) .. controls (0.5,2) and (1.75,2.5) .. (2,2.5);

\draw[thick,postaction={decorate}] (3.75,0) .. controls (3.5,2) and (2.25,2.5) .. (2,2.5);

\node at (0.25,-0.4) {$a$};

\node at (3.75,-0.4) {$b$};

\node at (2,5.4) {$a+b$};

\draw[thick,postaction={decorate}] (2,2.5) -- (2,5); 

\node at (5,2.5) {$=$};

\end{scope}


\begin{scope}[shift={(6,0)}]

\draw[thick,dashed] (0,5) -- (5,5);

\draw[line width = 0.45mm, blue] (0,0) -- (5,0);

\draw[thick,postaction={decorate}] (4.5,0) -- (4.5,5);

\node at (4.35,-0.4) {$a+b$};

\node at (2,3) {$(a+b)\left[ \dfrac{a}{a+b}\right]$};

\draw[thick,fill] (2.15,1.9) arc (0:360:1.5mm);

\end{scope}


\begin{scope}[shift={(13,0)}]

\draw[thick,dashed] (0.5,5) -- (8.5,5);

\draw[line width = 0.45mm, blue] (0.5,0) -- (8.5,0);

\draw[thick,postaction={decorate}] (1,0) .. controls (1.2,1) and (2.05,2.1) .. (2.25,2.1);
\draw[thick,postaction={decorate}] (2.25,2.1) .. controls (2.45,2.5) and (4.3,3.5) .. (4.5,3.5);

\draw[thick,postaction={decorate}] (8,0) .. controls (7.8,2) and (4.7,3.5) .. (4.5,3.5);

\node at (1,-0.4) {$b$};

\node at (7.85,-0.4) {$1-b$};

\draw[thick,postaction={decorate}] (4.5,1.25) .. controls (4.3,1.75) and (2.45,2.1) .. (2.25,2.1);
\node at (4.10,2.10) {$0$};

\draw[thick,postaction={decorate}] (3.25,0) .. controls (3.45,0.5) and (4.3,1.25) .. (4.5,1.25);
\node at (3.25,-0.4) {$a-b$};

\draw[thick,postaction={decorate}] (5.75,0) .. controls (5.55,0.5) and (4.7,1.25) .. (4.5,1.25);
\node at (5.75,-0.4) {$b-a$};

\draw[thick,postaction={decorate}] (4.5,3.5) -- (4.5,5);
\node at (5,4.1) {$1$};

\node at (9.25,2.5) {$=$};
\end{scope}


\begin{scope}[shift={(22.5,0)}]

\draw[thick,dashed] (0.5,5) -- (2,5);

\draw[line width = 0.45mm, blue] (0.5,0) -- (2,0);

\node at (0.8,2.8) {$\left[ b \right]$};
 
\draw[thick,fill] (0.95,2.1) arc (0:360:1.5mm);

\draw[thick,postaction={decorate}] (1.5,0) -- (1.5,5);

\node at (2,2.5) {$1$};

\end{scope}

\end{tikzpicture}
    \caption{Left: the blue boundary wall absorbs the additive defect on the $a$-line and $b$-line, emitting a floating point labeled $(a+b)\left[ \dfrac{a}{a+b}\right]$. Right: The three additive vertices are absorbed, emitting one vertex with the label $[b]$.  }
    \label{fig5_017}
\end{figure}
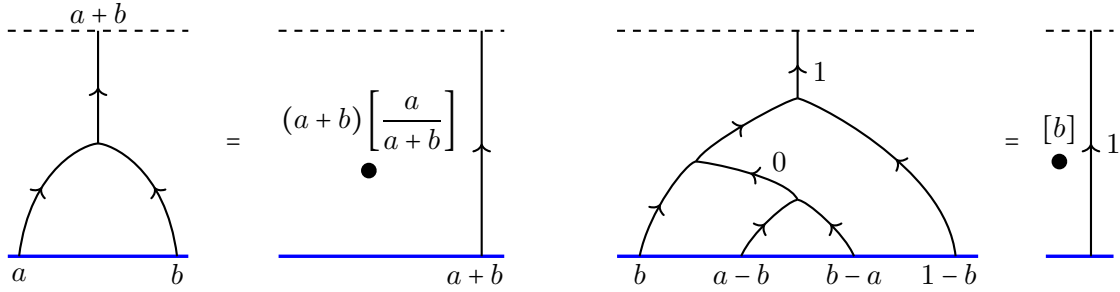

In Figure~\ref{fig5_016}, although the network initially appears complicated, after all additive vertices have been absorbed by the boundary wall, we are left with one line weighted $1$ and one vertex with appropriate weight.

\begin{figure}
    \centering
\begin{tikzpicture}[scale=0.6,decoration={
    markings,
    mark=at position 0.50 with {\arrow{>}}}]
\begin{scope}[shift={(0,0)}]
\draw[thick,dashed] (0,5) -- (9,5);
\draw[line width = 0.45mm, blue] (0,0) -- (9,0);

\draw[thick,postaction={decorate}] (0.75,0) .. controls (0.95,1) and (1.8,2) .. (2,2);
\draw[thick,postaction={decorate}] (3.25,0) .. controls (3.05,1) and (2.2,2) .. (2,2);
\node at (0.75,-0.35) {$b$};
\node at (3.25,-0.35) {$a-b$};

\draw[thick,postaction={decorate}] (5.75,0) .. controls (5.95,1) and (6.8,2) .. (7,2);
\draw[thick,postaction={decorate}] (8.25,0) .. controls (8.05,1) and (7.2,2) .. (7,2);
\node at (5.75,-0.35) {$b-a$};
\node at (8.25,-0.35) {$1-b$};

\draw[thick,postaction={decorate}] (2,2) .. controls (2.2,2.5) and (4.3,3.5) .. (4.5,3.5);
\draw[thick,postaction={decorate}] (7,2) .. controls (6.8,2.5) and (4.7,3.5) .. (4.5,3.5);
\node at (2.5,3) {$a$};
\node at (6.90,3) {$1-a$};

\draw[thick,postaction={decorate}] (4.5,3.5) -- (4.5,5);
\node at (5,4.15) {$1$};

\node at (10.5,2.5) {$=$};
\end{scope}

\begin{scope}[shift={(12.6,0)}]

\node at (0.5,2.5) {$a\left[\dfrac{b}{a} \right]$};
\node at (2,2.5) {$+$};
\node at (4.5,2.5) {$(1-a)\left[ \dfrac{1-b}{1-a} \right]$};
\node at (7,2.5) {$+$};
\node at (8,2.5) {$\left[ a \right]$};

\draw[thick,fill] (9.1,2.5) arc (0:360:1.5mm);
\end{scope}

\begin{scope}[shift={(21.5,0)}]

\draw[thick,dashed] (-9,5) -- (2.5,5);

\draw[line width = 0.45mm, blue] (-9,0) -- (2.5,0);

\draw[thick,postaction={decorate}] (1,0) -- (1,5);

\node at (1.5,2.5) {$1$};
\end{scope}

\end{tikzpicture}
    \caption{The three additive vertices are absorbed by the blue boundary wall, emitting three floating vertices with the weights $a\left[\dfrac{b}{a}\right]$, $(1-a)\left[\dfrac{1-b}{1-a}\right]$, and $[a]$. Since the floating vertices are additive, we can replace three floating vertices with one floating vertex with the label as the sum $a\left[\dfrac{b}{a}\right]+(1-a)\left[\dfrac{1-b}{1-a}\right] + [a] $.  }
    \label{fig5_016}
\end{figure}
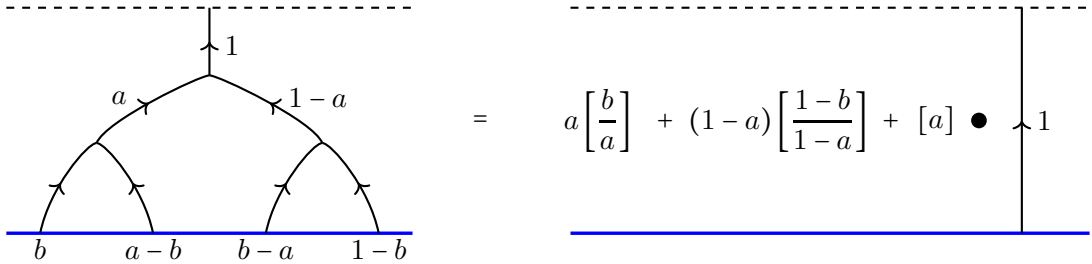

In Figure~\ref{fig5_102}, we can separate the one additive vertex by using appropriate isotopies, resulting in one floating defect with the weight given by the right-hand side of \eqref{eqn_relating_entropy_inner}, which is Shannon entropy $H(p_X)$. 

\begin{figure}
    \centering
\begin{tikzpicture}[scale=0.6,decoration={
    markings,
    mark=at position 0.50 with {\arrow{>}}}]


\begin{scope}[shift={(0,0)}]

\draw[thick,dashed] (0,3) -- (5,3); 

\draw[line width = 0.45mm,blue] (0,0) -- (5,0); 

\draw[thick,postaction={decorate}] (2,2) -- (2,3);

\draw[thick,postaction={decorate}] (0.5,0) -- (2,2);
\draw[thick,postaction={decorate}] (2.0,0) -- (2,2);
\draw[thick,postaction={decorate}] (4.5,0) -- (2,2);

\node at (3,0.5) {$\cdots$};

\node at (2,3.5) {$1$};

\node at (0.5,-0.5) {$p_1$};
\node at (2.0,-0.5) {$p_2$};
\node at (3.5,-0.5) {$\cdots$};
\node at (4.5,-0.5) {$p_n$};

\node at (6,1.5) {$=$};

\end{scope}


\begin{scope}[shift={(7,0)}]

\draw[thick,dashed] (0,3) -- (5,3); 

\draw[line width = 0.45mm,blue] (0,0) -- (5,0); 

\draw[thick,fill] (2,1) arc (0:360:1.5mm);

\draw[thick,postaction={decorate}] (3.5,0) -- (3.5,3);

\node at (1.85,1.85) {$H(p_X)$};

\node at (3.5,-0.5) {$1$};

\end{scope}


\begin{scope}[shift={(12,0)}]
\node at (5,2.1) {$H(p_X) = -\displaystyle{\sum_{i=1}^n} p_i \log p_i$};

\node at (5,0.65) {$p_1+p_2+\ldots + p_n=1$};
\end{scope}


\end{tikzpicture}
    \caption{If $p_1+ p_2 + \ldots + p_n =1$, the blue boundary wall absorbs all additive vertices, giving us entropy.}
    \label{fig5_102}
\end{figure}
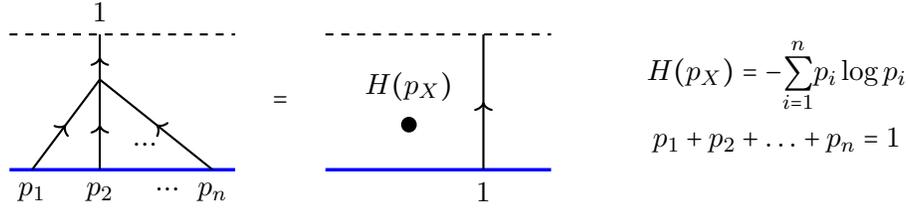

In Figure~\ref{fig1_001}, we have a more global picture of our cobordisms. That is, inside the purple dashed ellipse, we have classical physics, involving entropy and information theory. In the larger blue dotted ellipse, we have the theory of probability. When networks turn around, pointing downwards, and when we can rescale the additive lines, we have a larger depiction in the theory of information theory. We are yet to interpret the additive and multiplicative networks outside the blue ellipse, but we believe these observables may have an important interpretation in (quantum) physics.

\begin{figure}
    \centering
\begin{tikzpicture}[scale=0.6,decoration={
    markings,
    mark=at position 0.65 with {\arrow{>}}}]
\begin{scope}[shift={(0,0)}]

\draw[thick,dashed] (0,9) -- (25.5,9);
\draw[thick,dashed] (0,0) -- (25.5,0);

\draw[thick,postaction={decorate}] (2,0) .. controls (2.1,1) and (2.4,1.5) .. (3,2.5);

\draw[thick,postaction={decorate}] (4,0) .. controls (3.9,1) and (3.6,1.5) .. (3,2.5);

\draw[thick,postaction={decorate}] (3,2.5) .. controls (3.1,2.75) and (3.4,3.0) .. (4,3.5);

\draw[thick,postaction={decorate}] (8,0) .. controls (7,2.5) and (5.7,3) .. (4,3.5);

\draw[thick,postaction={decorate}] (4,3.5) -- (4,4.5);

\draw[thick] (4,4.5) -- (4,5);

\draw[thick,postaction={decorate}] (6,0) .. controls (6.5,2.25) and (8,2.5) .. (9,2.5);

\draw[thick,postaction={decorate}] (10,0) .. controls (9.9,1) and (9.6,1.5) .. (9,2.5);

\draw[thick,postaction={decorate}] (9,2.5) -- (9,4.75);

\draw[line width=0.55mm,dotted,red!100!green!36!blue!100] (6.25,2.60) ellipse (5 and 2);

\draw[line width=0.63mm, densely dotted,blue] (6.25,3.42) ellipse (6 and 3.18);

\draw[thick,postaction={decorate}] (4,5) .. controls (3.9,5.1) and (3.5,5.2) .. (3,6.5);

\draw[thick,postaction={decorate}] (4,5) .. controls (4.1,5.00) and (4.5,5.10) .. (5,5.40);

\draw[thick,postaction={decorate}] (3,6.5) .. controls (2.8,7.5) and (2.3,6.75) .. (2,9);

\draw[thick,postaction={decorate}] (9,4.75) .. controls (7,5) and (5,5.5) .. (4.5,6.5);

\draw[thick,postaction={decorate}] (5,5.40) .. controls (5.2,5.6) and (6.5,6) .. (8,9);

\draw[thick,postaction={decorate}] (4.5,6.5) .. controls (4.6,7.25) and (4.7,7.25) .. (6,9);

\draw[thick,postaction={decorate}] (9,4.75) .. controls (11,6.5) and (10,7.5) .. (7.45,8);

\draw[thick,postaction={decorate}] (9.8,7) .. controls (10.5,7.2) and (11,7.5) .. (12,9);

\draw[line width=0.50mm,red] (4,9) decorate [decoration={snake,amplitude=0.20mm}] {.. controls (4.25,6) and (9,8.5) .. (12,5.5)};

\draw[line width=0.50mm,red] (12,5.5) decorate [decoration={snake,amplitude=0.20mm}] {.. controls (13,4.75) and (13.5,4) .. (14,0)};

\draw[line width=0.50mm,red] (3.65,8.1) -- (4.15,8.25);
\node at (3.6,8.55) {$c_2$};

\draw[line width=0.50mm,red] (13.45,1.25) -- (13.95,1.25);
\node at (13.4,1.75) {$c_1$};

\draw[thick,postaction={decorate}] (10.75,7.5) .. controls (13.5,7) and (21,6) .. (22,0);

\draw[line width=0.50mm,red] (12,5.5) decorate [decoration={snake,amplitude=0.20mm}] {.. controls (12.5,6) and (13,6.5) .. (14,9)};

\draw[line width=0.50mm,red] (13.1,8) -- (13.6,8);
\node at (14.65,8.1) {$c_2^{-1}c_1$};

\draw[thick,postaction={decorate}] (16,0) .. controls (16.5,2) and (19.5,2) .. (20,0);

\draw[line width=0.50mm,red] (18,0) decorate [decoration={snake,amplitude=0.20mm}] {.. controls (18.5,2) and (19,7) .. (20,9)};

\draw[line width=0.50mm,red] (19,7.5) -- (19.5,7.5);
\node at (20,7.5) {$c_3$};

\draw[thick,<-] (17,3.5) arc (0:360:1);
\node at (15,4.5) {$p_j$};

\begin{scope}[shift={(-0.25,0)}]
\draw[thick,<-] (24,6) arc (0:180:1.5);
\draw[thick,<-] (21,6) arc (180:360:1.5);

\draw[thick,postaction={decorate}] (22.75,4.52) arc (-110:-260:1.5);
\node at (20.85,4.9) {$p_k$};
\node at (22.35,5.75) {$p_{\ell}$};
\end{scope}
\node at (24.40,4.60) {$p_k+p_{\ell}$};

\node at (2,-0.75) {$p_1$};
\node at (4,-0.75) {$p_2$};
\node at (6,-0.75) {$\ldots$};
\node at (8,-0.75) {$\ldots$};
\node at (10,-0.75) {$p_n$};

\node at (16,-0.75) {$p_{n+1}$};
\node at (20,-0.75) {$\ldots$};
\node at (22,-0.75) {$p_m$};
\node at (24,-0.75) {$\ldots$};

\draw[thick,->] (1,-2) .. controls (0.5,0) and (0.5,1.5) .. (2.25,2.5);
\node at (2,-2.5) {probability for};
\node at (2.25,-3.5) {classical physics};
\node at (6.25,-3.5) {$\displaystyle{\sum_{i=1}^n} \: p_i=1$};

\draw[thick,->] (1.5,10) .. controls (0,8) and (0,6) .. (2,4.5);
\node at (1.60,10.5) {probability};

\end{scope}
\end{tikzpicture}
    \caption{The diagram shows a morphism (cobordism). Small purple ellipse encloses a part of the morphism that admits an interpretation in earlier categorical approaches to entropy. Larger blue ellipse encloses a part of the diagram that can be interpreted in classical probability. Parts of the diagram that include wavy (red) multiplicative lines, U-turns of additive lines and additive lines pointing down or closing upon themselves make sense in a larger category, but not in the familiar classical probability, since rescaling is involved and one could remove the renormalization to $1$. Multiplicative lines scale labels of additive lines.}
    \label{fig1_001}
\end{figure}
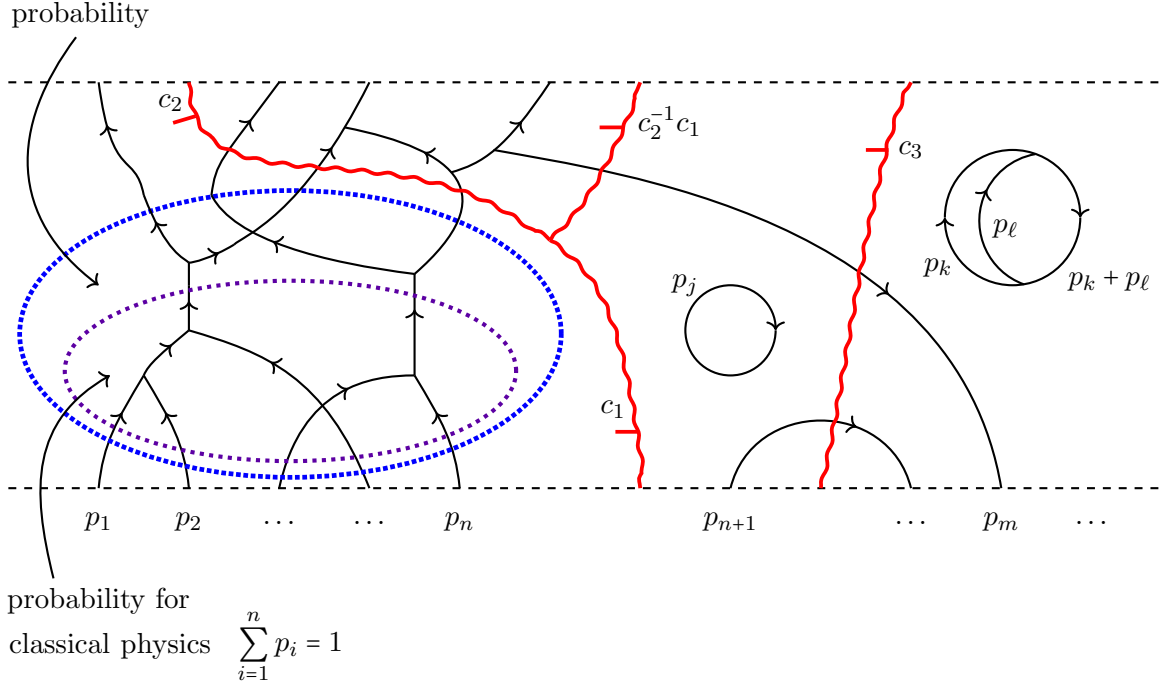

\section{Connections from diagrammatics to \texorpdfstring{$J(\kk)$}{J(k)}}
\label{section:diagrammatics_vector_space_Jk}

Now, we may want to ask ourselves how do we go from the diagrammatics in Section~\ref{section:Shannon_entropy} to elements in the vector space $J(\kk)$. 
Given a cobordism (diagram) $\gamma$, 
define $\jmath:\textsf{Diagrams} \rightarrow J(\kk)$, where 
\begin{equation}
\jmath(\gamma) = \sum_{p\in \mathrm{add}(\gamma)} s(p) \omega(p,\gamma)\langle a_p, b_p\rangle \ \in \ J(\mathbf{k}),
\end{equation}
where we sum over all the additive vertices in $\gamma$. The coefficient $s(p)$ equals $1$ if the vertex $p$ is a merge, and $-1$ if the additive vertex $p$ is a split. We denote $a_p$ and $b_p$ the two lines weighted $a$ and $b$, respectively, at point $p$. The coefficient $\omega(p,\gamma)$ is the product of all the multiplicative red lines as we virtually move the additive vertex $p$ far to the left as possible, crossing over any multiplicative lines and taking into account the conormal direction of these lines. 

The value $\jmath(\gamma)$ encodes contribution from additive vertices. It computes $2$-cocycles and it is an invariant under isotopies.
So to a diagrammatic morphism $\gamma$, assign element $\jmath(\gamma)$ of $J(\kk)$. 
Generating objects of this monoidal category are labelled by $a\in \kk$ (endpoints of additive networks) and $c\in \kk^{\ast}$ (endpoints of multiplicative networks). 

From \cite[Section 5]{IK24_dilogarithms_entropy}, we have the following: 
\begin{proposition}[Im--Khovanov]
The invariant $\jmath(\gamma)$ depends only on the source and target objects of the morphism $\gamma$.
\end{proposition}
Because $\jmath(\gamma)$ only depends on the source and the target objects of morphism $\gamma$, and not on any of the interior network, $\jmath$ is an invariant. 
We leave it as future work to explore other invariants of these cobordisms.

\section{Diagrammatics of conditional entropy}

Let $H(Y|X)$ be the conditional entropy of a discrete random variable $Y$ given the random variable $X$.
In terms of probability, conditional entropy is defined to be 
\begin{equation}
\label{eqn:cond_entropy}
H(X|Y) = - \sum_{i=1}^{n} \sum_{j=1}^{m}
P(X_i,Y_j) \log P(X_i|Y_j), 
\end{equation}
where 
\begin{equation}
\label{eqn:joint_entropy}
H(X,Y) = H(Y|X) +  H(X).
\end{equation}
This is reminiscent of an additive version of the ordinary chain rule. 

\begin{figure}
    \centering
\begin{tikzpicture}[scale=0.6,decoration={
    markings,
    mark=at position 0.60 with {\arrow{>}}}]


\begin{scope}[shift={(0,0)}]

\begin{scope}[shift={(0.15,0)}]
\draw[thick, dashed] (-0.5,4) -- (3.5,4);

\draw[thick, dashed] (0,2) -- (4.2,2);

\draw[thick, dashed] (0.5,0) -- (4.5,0);

\draw[thick, dashed] (0.5,0) -- (-0.5,4);

\draw[thick, dashed] (2.5,0) -- (1.5,4);

\draw[thick, dashed] (4.5,0) -- (3.5,4);

\node at (1.5,-0.5) {$p_1$};

\node at (3.5,-0.5) {$p_2$};

\node at (-0.75,3) {$q_2$};

\node at (-0.25,1) {$q_1$};
\end{scope}

\begin{scope}[shift={(-0.5,-0.25)}]
\draw[thick,postaction={decorate}] (1,3) -- (2.25,4.75);

\draw[thick,postaction={decorate}] (3,3) -- (2.25,4.75);
\end{scope}

\begin{scope}[shift={(0.75,0.25)}]
\draw[thick,postaction={decorate}] (0.5,1) -- (2,3.5);

\draw[thick,postaction={decorate}] (3.25,1) -- (2,3.5);
\end{scope}

\node at (0.6,5.5) {$p_{12}+p_{22}$};

\node at (4.15,5) {$p_{11} + p_{21}$};

\node at (0.9,2.4) {$p_{12}$};

\node at (2.7,2.4) {$p_{22}$};

\node at (1.35,0.8) {$p_{11}$};

\node at (3.75,0.9) {$p_{21}$};

\draw[thick,postaction={decorate}] (1.75,4.5) -- (2.5,7.5);

\draw[thick,postaction={decorate}] (2.75,3.75) -- (2.5,7.5);

\draw[thick,postaction={decorate}] (2.5,7.5) -- (2.5,9);

\node at (7.5,9) {$p_{11}+p_{21}+p_{12}+p_{22} = 1$};

\begin{scope}[shift={(0,6.5)}]
\draw[thick, dashed] (-0.5,4) -- (3.5,4);

\draw[thick, dashed] (0,2) -- (4,2);

\draw[thick, dashed] (0.5,0) -- (4.5,0);

\draw[thick, dashed] (0.5,0) -- (-0.5,4);

\draw[thick, dashed] (2.5,0) -- (1.5,4);

\draw[thick, dashed] (4.5,0) -- (3.5,4);

\end{scope}
 
\end{scope}


\end{tikzpicture}
    \caption{Joint entropy for two random variables $X$ and $Y$ with probabilities $p_1$ and $p_2$ and $q_1$ and $q_2$, respectively. }
    \label{fig8_003}
\end{figure}
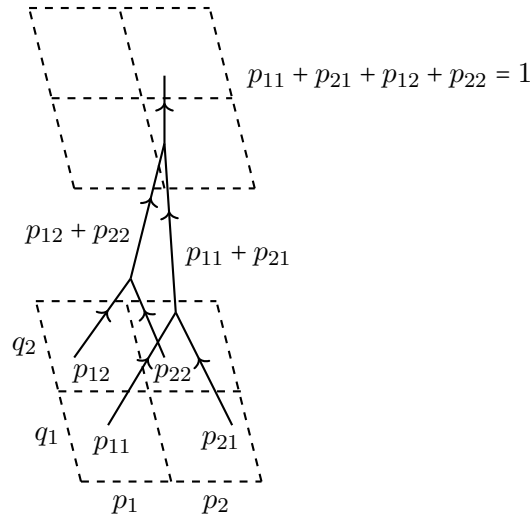

\begin{figure}
    \centering
\begin{tikzpicture}[scale=0.6,decoration={
    markings,
    mark=at position 0.60 with {\arrow{>}}}]


\begin{scope}[shift={(0,0)}]

\begin{scope}[shift={(0.15,0)}]
\draw[thick, dashed] (-0.5,4) -- (5.5,4);

\draw[thick, dashed] (0,2) -- (6,2);

\draw[thick, dashed] (0.5,0) -- (6.5,0);

\draw[thick, dashed] (0.5,0) -- (-0.5,4);

\draw[thick, dashed] (2.5,0) -- (1.5,4);

\draw[thick, dashed] (4.5,0) -- (3.5,4);

\draw[thick, dashed] (6.5,0) -- (5.5,4);

\node at (1.5,-0.5) {$p_1$};

\node at (3.5,-0.5) {$p_2$};

\node at (5.5,-0.5) {$p_3$};

\node at (-0.75,3) {$q_2$};

\node at (-0.25,1) {$q_1$};
\end{scope}

\begin{scope}[shift={(-0.5,-0.25)}]
\draw[thick,postaction={decorate}] (1,3) -- (4.25,7);

\draw[thick,postaction={decorate}] (5,3) -- (4.25,7);

\draw[thick,postaction={decorate}] (3,3) -- (2.5,4.85);
\end{scope}

\begin{scope}[shift={(0.75,0.25)}]
\draw[thick,postaction={decorate}] (1,1) -- (4.25,4.75);

\draw[thick,postaction={decorate}] (5,1) -- (4.25,4.75);

\draw[thick,postaction={decorate}] (3,1) -- (2.5,2.75);
\end{scope}

\draw[thick,postaction={decorate}] (3.75,6.75) -- (3.5,9);

\draw[thick,postaction={decorate}] (5,5) -- (5.5,7.25);

\node at (-3,9) {$\mathsf{y}_2:= p_{12}+p_{22}+p_{32}$};

\draw[thick,->] (0,9) -- (3,9);

\node at (1.5,5.75) {$p_{12}+p_{22}$};

\node at (8,4.25) {$p_{11} + p_{21}$};

\draw[thick,<-] (4.5,4.25) -- (6.5,4.25);

\draw[thick,<-] (5.75,7.25) -- (7.5,7.25);

\node at (10.5,7.25) {$\mathsf{y}_1:= p_{11}+p_{21}+p_{31}$};

\node at (1.25,2.75) {$p_{12}$};

\node at (2.90,3.5) {$p_{22}$};

\node at (4.9,2.5) {$p_{32}$};

\node at (1.35,1) {$p_{11}$};

\node at (3.35,1) {$p_{21}$};

\node at (5.35,1) {$p_{31}$};

\begin{scope}[shift={(0.25,6.5)}]
\draw[thick, dashed] (-0.5,4) -- (5.5,4);

\draw[thick, dashed] (0,2) -- (6,2);

\draw[thick, dashed] (0.5,0) -- (6.5,0);

\draw[thick, dashed] (0.5,0) -- (-0.5,4);

\draw[thick, dashed] (2.5,0) -- (1.5,4);

\draw[thick, dashed] (4.5,0) -- (3.5,4);

\draw[thick, dashed] (6.5,0) -- (5.5,4);
\end{scope}
 
\end{scope}


\end{tikzpicture}
    \caption{A diagrammatic for the conditional entropy $H(X|Y) = H(X,Y)-H(Y)$. That is, we group along $Y$'s, but we do not group the $Y$'s themselves. }
    \label{fig7_005}
\end{figure}
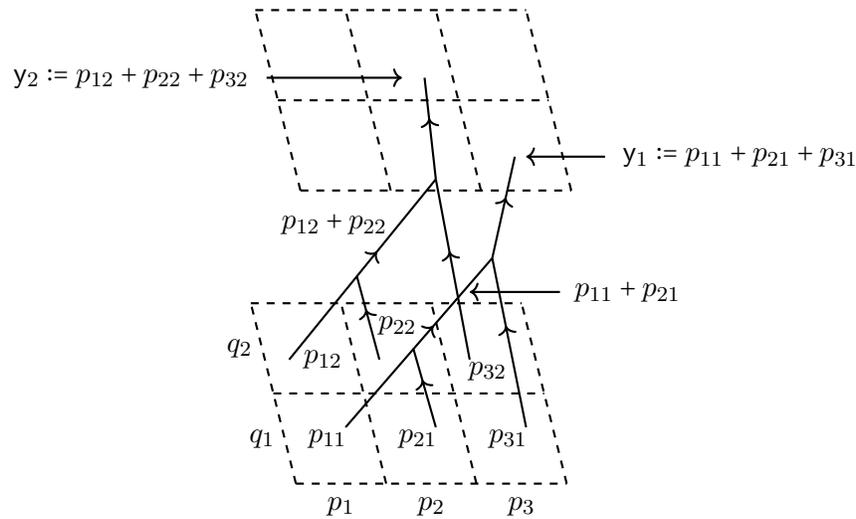

Diagram associated to Proposition~\ref{prop_joint_entropy_two} is Figure~\ref{fig8_003}. 
In order to understand conditional entropy in terms of joint entropy~\eqref{eqn:joint_entropy} using diagrammatics, we provide the diagrammatics for 
\begin{equation}
H(Y|X) =  H(X,Y) - H(X).
\end{equation}
In Figure~\ref{fig7_005}, we give an example of how conditional entropy of two random variables can be interpreted. 
In Figure~\ref{fig7_001}, we provide two different ways to obtain entropy of two random variables: $H(X,Y) = H(X|Y)+H(Y)$ and $H(X,Y) = H(Y|X)+H(X)$. 
In Figure~\ref{fig7_002}, the morphism simplifies to the identity cobordism.

\begin{figure}
    \centering
\begin{tikzpicture}[scale=0.6,decoration={
    markings,
    mark=at position 0.60 with {\arrow{>}}}]


\begin{scope}[shift={(0,0)}]

\draw[thick, dashed] (0,5) -- (10,5);

\draw[thick, dashed] (0,0) -- (10,0);

\draw[thick,postaction={decorate}] (1,0) -- (7,3);

\draw[thick,postaction={decorate}] (9,0) -- (7,3);

\draw[thick,postaction={decorate}] (7,3) -- (7,5);

\draw[thick,postaction={decorate}] (3.5,0) -- (2.85,0.9);

\draw[thick,postaction={decorate}] (5,0) -- (4,1.5);

\draw[thick,postaction={decorate}] (7.5,0) -- (5.85,2.44);

\node at (7,5.5) {$1$};

\node at (1,-0.5) {$p_1$};

\node at (3.5,-0.5) {$p_2$};

\node at (5,-0.5) {$p_3$};

\node at (6.5,-0.5) {$\ldots$};

\node at (9,-0.5) {$p_k$};

\node at (8.5,2) {$X$};
\end{scope}


\begin{scope}[shift={(12,0)}]

\draw[thick, dashed] (0,5) -- (10,5);

\draw[thick, dashed] (0,0) -- (10,0);

\draw[thick,postaction={decorate}] (1,0) -- (7,3);

\draw[thick,postaction={decorate}] (9,0) -- (7,3);

\draw[thick,postaction={decorate}] (7,3) -- (7,5);

\draw[thick,postaction={decorate}] (3.5,0) -- (2.85,0.9);

\draw[thick,postaction={decorate}] (5,0) -- (4,1.5);

\draw[thick,postaction={decorate}] (7.5,0) -- (5.85,2.44);

\node at (7,5.5) {$1$};

\node at (1,-0.5) {$q_1$};

\node at (3.5,-0.5) {$q_2$};

\node at (5,-0.5) {$q_3$};

\node at (6.5,-0.5) {$\ldots$};

\node at (9,-0.5) {$q_m$};

\node at (8.5,2) {$Y$};
\end{scope}


\begin{scope}[shift={(3,-12)}]

\begin{scope}[shift={(0.15,0)}]
\draw[thick, dashed] (-0.5,4) -- (5.5,4);

\draw[thick, dashed] (0,2) -- (6,2);

\draw[thick, dashed] (0.5,0) -- (6.5,0);

\draw[thick, dashed] (0.5,0) -- (-0.5,4);

\draw[thick, dashed] (2.5,0) -- (1.5,4);

\draw[thick, dashed] (4.5,0) -- (3.5,4);

\draw[thick, dashed] (6.5,0) -- (5.5,4);

\node at (1.75,-0.5) {$p_1$};

\node at (3.75,-0.5) {$p_2$};

\node at (5.75,-0.5) {$p_3$};

\node at (-0.75,2.75) {$q_2$};

\node at (-0.25,1) {$q_1$};
\end{scope}

\begin{scope}[shift={(-0.5,-0.25)}]
\draw[thick,postaction={decorate}] (1,3) -- (4.25,7);

\draw[thick,postaction={decorate}] (5,3) -- (4.25,7);

\draw[thick,postaction={decorate}] (3,3) -- (2.5,4.85);
\end{scope}

\begin{scope}[shift={(0.75,0.25)}]
\draw[thick,postaction={decorate}] (1,1) -- (4.25,5);

\draw[thick,postaction={decorate}] (5,1) -- (4.25,5);

\draw[thick,postaction={decorate}] (3,1) -- (2.5,2.85);
\end{scope}

\draw[thick,postaction={decorate}] (3.75,6.75) -- (4.5,8);

\draw[thick,postaction={decorate}] (5,5.25) -- (4.5,8);

\draw[thick,postaction={decorate}] (4.5,8) -- (4.5,10);

\node at (5,9) {$1$};

\node at (1.25,7.50) {$\mathsf{y}_2:= p_{12}+p_{22}+p_{32}$};

\node at (1.5,6) {$p_{12}+p_{22}$};

\node at (8,4.5) {$p_{11} + p_{21}$};

\draw[thick,<-] (4.5,4.5) -- (6.5,4.5);

\node at (7.75,6.5) {$\mathsf{y}_1:= p_{11}+p_{21}+p_{31}$};

\node at (1.25,2.75) {$p_{12}$};

\node at (2.90,3.5) {$p_{22}$};

\node at (4.9,2.5) {$p_{32}$};

\node at (1.35,1) {$p_{11}$};

\node at (3.35,1) {$p_{21}$};

\node at (5.35,1) {$p_{31}$};

\draw[thick] (11,7.5) -- (11,10);

\draw[thick] (11,10) -- (10.5,10);

\draw[thick] (11,7.5) -- (10.5,7.5);

\node at (12.25,8.75) {$H(Y)$};

\draw[thick] (11,0) -- (11,7);

\draw[thick] (10.5,7) -- (11,7);

\draw[thick] (10.5,0) -- (11,0);

\node at (13.75,3.5) {$\underbrace{H(X,Y)-H(Y)}_{H(X|Y)}$};

\end{scope}


\begin{scope}[shift={(3,-24)}]

\begin{scope}[shift={(0.15,0)}]
\draw[thick, dashed] (-0.5,4) -- (5.5,4);

\draw[thick, dashed] (0,2) -- (6,2);

\draw[thick, dashed] (0.5,0) -- (6.5,0);

\draw[thick, dashed] (0.5,0) -- (-0.5,4);

\draw[thick, dashed] (2.5,0) -- (1.5,4);

\draw[thick, dashed] (4.5,0) -- (3.5,4);

\draw[thick, dashed] (6.5,0) -- (5.5,4);

\node at (1.75,-0.5) {$p_1$};

\node at (3.75,-0.5) {$p_2$};

\node at (5.75,-0.5) {$p_3$};

\node at (-0.75,2.75) {$q_2$};

\node at (-0.25,1) {$q_1$};
\end{scope}

\draw[thick,postaction={decorate}] (2,1) -- (1.5,3.5);

\draw[thick,postaction={decorate}] (0.75,2.5) -- (1.5,3.5);

\begin{scope}[shift={(2,0)}] 
\draw[thick,postaction={decorate}] (2,1) -- (1.5,3.5);

\draw[thick,postaction={decorate}] (0.75,2.5) -- (1.5,3.5);
\end{scope}

\begin{scope}[shift={(4,0)}] 
\draw[thick,postaction={decorate}] (2,1) -- (1.5,3.5);

\draw[thick,postaction={decorate}] (0.75,2.5) -- (1.5,3.5);
\end{scope}

\draw[thick,postaction={decorate}] (1.5,3.5) -- (3,6);

\draw[thick,postaction={decorate}] (3.5,3.5) -- (3,6);

\draw[thick,postaction={decorate}] (3,6) -- (4.5,8);

\draw[thick,postaction={decorate}] (5.5,3.5) -- (4.5,8);

\draw[thick,postaction={decorate}] (4.5,8) -- (4.5,10);

\node at (5,9) {$1$};

\node at (2.4,7.0) {$\mathsf{x}_1 + \mathsf{x}_2$};

\node at (0.25,5) {$\mathsf{x}_1:= p_{11}+p_{12}$};

\node at (8,4.5) {$\mathsf{x}_2 := p_{21} + p_{22}$};

\draw[thick,<-] (3.5,4.5) -- (6.0,4.5);

\node at (6.9,6.5) {$\mathsf{x}_3:= p_{31}+p_{32}$};

\node at (0.65,3.35) {$p_{12}$};

\node at (2.65,3.35) {$p_{22}$};

\node at (4.65,3.35) {$p_{32}$};

\node at (1.5,0.75) {$p_{11}$};

\node at (3.5,0.75) {$p_{21}$};

\node at (5.5,0.75) {$p_{31}$};

\draw[thick] (11,5.25) -- (11,10);

\draw[thick] (11,10) -- (10.5,10);

\draw[thick] (11,5.25) -- (10.5,5.25);

\node at (12.25,7.625) {$H(X)$};

\draw[thick] (11,0) -- (11,5);

\draw[thick] (10.5,5) -- (11,5);

\draw[thick] (10.5,0) -- (11,0);

\node at (13.75,2.5) {$\underbrace{H(X,Y)-H(X)}_{H(Y|X)}$};

\end{scope}


\end{tikzpicture}
    \caption{Let $\{p_1,p_2, p_3 \}$ be probabilities for the random variable $X$ and $\{q_1,q_2,q_3\}$ be probabilities for the random variable $Y$. Let $p_{ij}:=(p_i, q_j)$. Top left: random variable $X$ with probabilities $\{p_1, p_2, \ldots, p_k\}$. Top right: random variable $Y$ with probabilities $\{q_1, q_2, \ldots, q_m\}$.  
    Middle: absorbing all the additive vertices gives us $H(X,Y)$.
    Absorbing all the additive vertices except the top one gives $H(X,Y)- H(Y)$.  Bottom: absorbing all the additive vertices except the top one gives $H(X,Y) - H(X)$.}
    \label{fig7_001}
\end{figure}
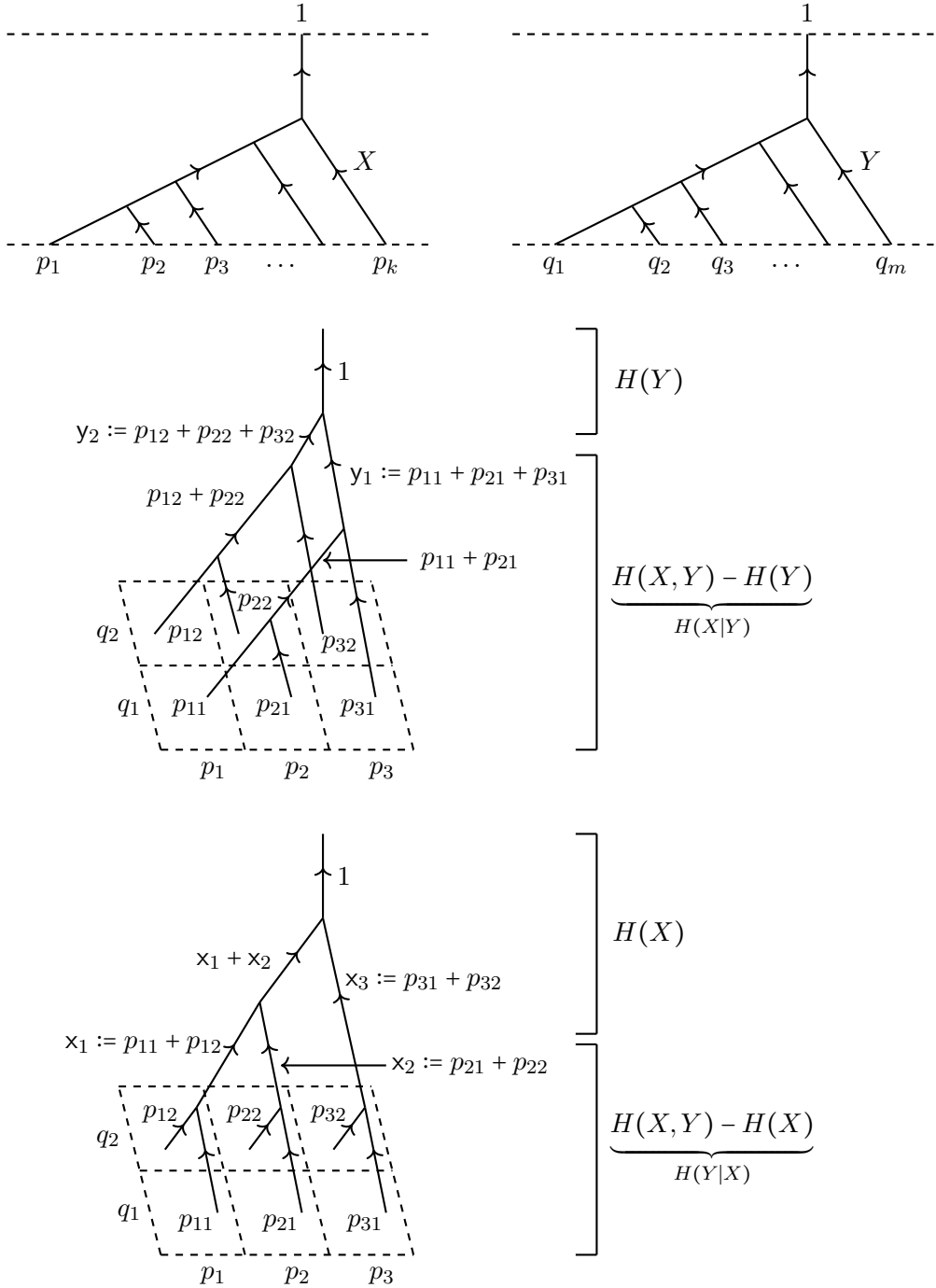

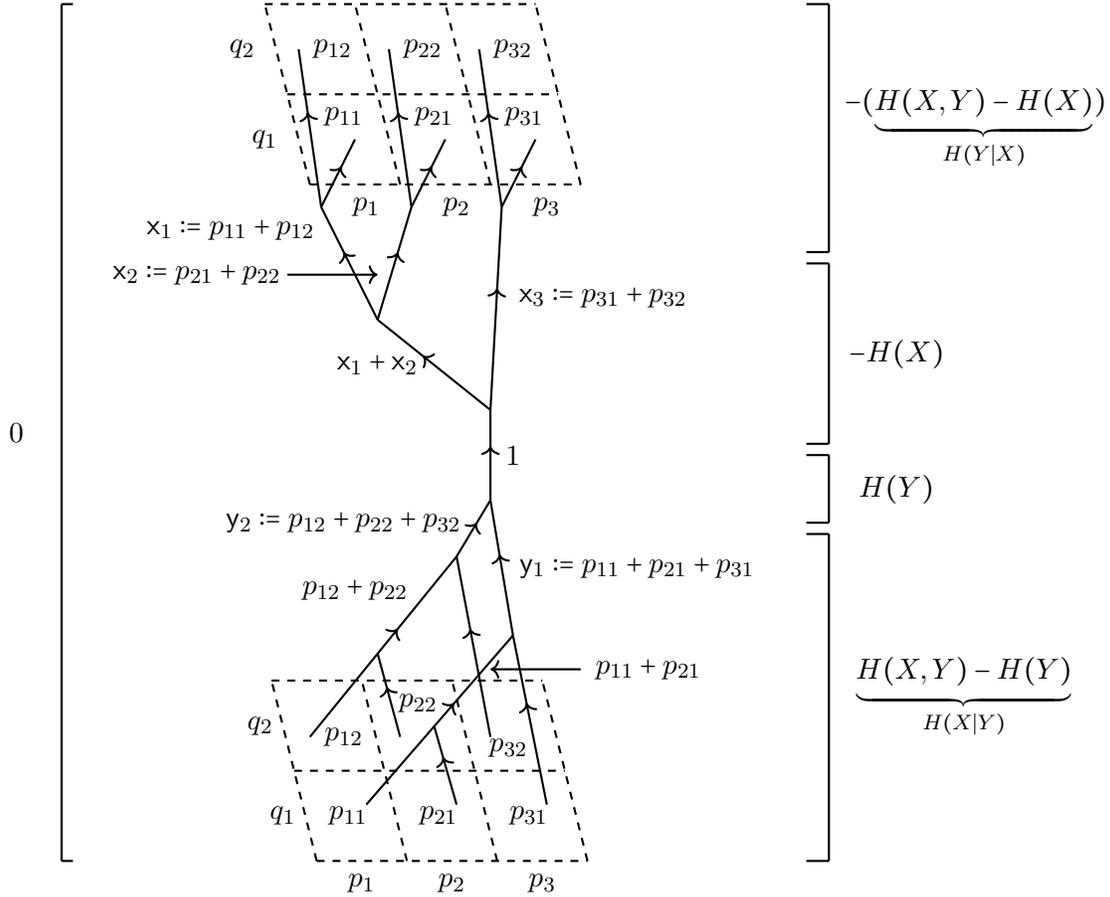
\begin{figure}
    \centering
\begin{tikzpicture}[scale=0.6,decoration={
    markings,
    mark=at position 0.60 with {\arrow{>}}}]


\begin{scope}[shift={(0,0)}]

\begin{scope}[shift={(0.15,0)}]
\draw[thick, dashed] (-0.5,4) -- (5.5,4);

\draw[thick, dashed] (0,2) -- (6,2);

\draw[thick, dashed] (0.5,0) -- (6.5,0);

\draw[thick, dashed] (0.5,0) -- (-0.5,4);

\draw[thick, dashed] (2.5,0) -- (1.5,4);

\draw[thick, dashed] (4.5,0) -- (3.5,4);

\draw[thick, dashed] (6.5,0) -- (5.5,4);

\node at (1.5,-0.5) {$p_1$};

\node at (3.5,-0.5) {$p_2$};

\node at (5.5,-0.5) {$p_3$};

\node at (-0.75,3) {$q_2$};

\node at (-0.25,1) {$q_1$};
\end{scope}

\begin{scope}[shift={(-0.5,-0.25)}]
\draw[thick,postaction={decorate}] (1,3) -- (4.25,7);

\draw[thick,postaction={decorate}] (5,3) -- (4.25,7);

\draw[thick,postaction={decorate}] (3,3) -- (2.5,4.85);
\end{scope}

\begin{scope}[shift={(0.75,0.25)}]
\draw[thick,postaction={decorate}] (1,1) -- (4.25,4.75);

\draw[thick,postaction={decorate}] (5,1) -- (4.25,4.75);

\draw[thick,postaction={decorate}] (3,1) -- (2.5,2.75);
\end{scope}

\draw[thick,postaction={decorate}] (3.75,6.75) -- (4.5,8);

\draw[thick,postaction={decorate}] (5,5) -- (4.5,8);

\draw[thick,postaction={decorate}] (4.5,8) -- (4.5,10);

\node at (5,9) {$1$};

\node at (1.25,7.50) {$\mathsf{y}_2:= p_{12}+p_{22}+p_{32}$};

\node at (1.5,6) {$p_{12}+p_{22}$};

\node at (8,4.25) {$p_{11} + p_{21}$};

\draw[thick,<-] (4.5,4.25) -- (6.5,4.25);

\node at (7.75,6.5) {$\mathsf{y}_1:= p_{11}+p_{21}+p_{31}$};

\node at (1.25,2.75) {$p_{12}$};

\node at (2.90,3.5) {$p_{22}$};

\node at (4.9,2.5) {$p_{32}$};

\node at (1.35,1) {$p_{11}$};

\node at (3.35,1) {$p_{21}$};

\node at (5.35,1) {$p_{31}$};

\begin{scope}[shift={(0,15)}]
\draw[thick, dashed] (-0.5,4) -- (5.5,4);

\draw[thick, dashed] (0,2) -- (6,2);

\draw[thick, dashed] (0.5,0) -- (6.5,0);

\draw[thick, dashed] (0.5,0) -- (-0.5,4);

\draw[thick, dashed] (2.5,0) -- (1.5,4);

\draw[thick, dashed] (4.5,0) -- (3.5,4);

\draw[thick, dashed] (6.5,0) -- (5.5,4);
\end{scope}

\node at (1.75,14.5) {$p_1$};

\node at (3.75,14.5) {$p_2$};

\node at (5.75,14.5) {$p_3$};

\node at (-1.0,18) {$q_2$};

\node at (-0.5,16) {$q_1$};

\node at (5,18) {$p_{32}$};

\node at (3,18) {$p_{22}$};

\node at (1,18) {$p_{12}$};

\node at (5.25,16.5) {$p_{31}$};

\node at (3.25,16.5) {$p_{21}$};

\node at (1.25,16.5) {$p_{11}$};

\draw[thick,postaction={decorate}] (0.75,14.5) -- (0.25,18);

\draw[thick,postaction={decorate}] (0.75,14.5) -- (1.5,16);

\begin{scope}[shift={(2,0)}]
\draw[thick,postaction={decorate}] (0.75,14.5) -- (0.25,18);

\draw[thick,postaction={decorate}] (0.75,14.5) -- (1.5,16);
\end{scope}

\begin{scope}[shift={(4,0)}]
\draw[thick,postaction={decorate}] (0.75,14.5) -- (0.25,18);

\draw[thick,postaction={decorate}] (0.75,14.5) -- (1.5,16);
\end{scope}

\draw[thick,postaction={decorate}] (2,12) -- (0.75,14.5);

\draw[thick,postaction={decorate}] (2,12) -- (2.75,14.5);

\draw[thick,postaction={decorate}] (4.5,10) -- (2,12);

\draw[thick,postaction={decorate}] (4.5,10) -- (4.75,14.5);

\node at (-1.25,14) {$\mathsf{x}_1:= p_{11}+ p_{12}$};

\node at (-2,13) {$\mathsf{x}_2:= p_{21}+p_{22}$};

\draw[thick,->] (0,13) -- (2,13);

\node at (2,11) {$\mathsf{x}_1+ \mathsf{x}_2$};

\node at (7,12.5) {$\mathsf{x}_3:= p_{31}+p_{32}$};

\draw[thick] (12,13.5) -- (12,19);

\draw[thick] (11.5,19) -- (12,19);

\draw[thick] (11.5,13.5) -- (12,13.5);

\node at (15.25,16.25) {$-(\underbrace{H(X,Y)-H(X)}_{H(Y|X)})$};

\draw[thick] (12,9.25) -- (12,13.25);

\draw[thick] (11.5,13.25) -- (12,13.25);

\draw[thick] (11.5,9.25) -- (12,9.25);

\node at (13.5,11.25) {$-H(X)$};

\draw[thick] (12,7.5) -- (12,9);

\draw[thick] (11.5,9) -- (12,9);

\draw[thick] (11.5,7.5) -- (12,7.5);

\node at (13.5,8.25) {$H(Y)$};

\draw[thick] (12,0) -- (12,7.25);

\draw[thick] (11.5,7.25) -- (12,7.25);

\draw[thick] (11.5,0) -- (12,0);

\node at (15,3.625) {$\underbrace{H(X,Y)-H(Y)}_{H(X|Y)}$};

\draw[thick] (-5,0) -- (-5,19);

\draw[thick] (-5,19) -- (-4.75,19);

\draw[thick] (-5,0) -- (-4.75,0);

\node at (-6,9.5) {$0$};

\end{scope}


\end{tikzpicture}
    \caption{Let $\{p_1,p_2, p_3 \}$ be probabilities for the random variable $X$ and $\{q_1,q_2,q_3\}$ be probabilities for the random variable $Y$. Let $p_{ij}:=(p_i, q_j)$. 
    Above is a diagrammatic of $0$ since the top and the bottom collections of boundary points in it are identical.  }
    \label{fig7_002}
\end{figure}

We see that one can draw diagrammatics for the information content, i.e., self-information, which is  
$I(X) = \log \frac{1}{p_X}$. In terms of entropy, it can be reinterpreted as $H(p_X)= E[I(X)]$. We leave it as an exercise to draw the diagrams for mutual information: 
\begin{equation}
\begin{split}
    I(X,Y) &= H(X) - H(X|Y)  = H(Y)- H(Y|X) \\ 
    &= H(X) + H(Y)- H(X,Y) \\ 
    &= H(X,Y)-H(X|Y) - H(Y|X). 
\end{split}
\end{equation}

\section{Additional properties of \texorpdfstring{$J(\mathbf{k})$}{J(k)}}
\label{section_basic_properties}

From Im--Khovanov in \cite{IK24_dilogarithms_entropy}, we saw that the two symbols $\langle \cdot , \cdot \rangle$ in Cathelineau's $\kk$-vector space is related to $H(p_X)$ via the equation
\[
\langle a,b\rangle = (a+b) H\left( \frac{a}{a+b} \right)
\] 
if $a+b\not=0$; otherwise, $\langle a, b \rangle =0$.
This shows that rescaling the two symbols by $a+b$ gives us a rescaled entropy: 
\begin{equation}
\frac{1}{a+b} \langle a,b \rangle = H\left( \frac{a}{a+b} \right) = H\left(\frac{b}{a+b} \right). 
\end{equation}
In particular, when $a=b$, then 
\begin{equation}
\label{equation:inner_product_a_with_a}
\langle a,a \rangle = 2a H\left(\frac{1}{2}\right),
\end{equation}
or equivalently, 
\[
\frac{1}{2a} \langle a,a \rangle = H\left( \frac{1}{2} \right) 
\quad \mbox{ or } 
\quad  \frac{1}{2} \langle 1,1 \rangle = H\left( \frac{1}{2} \right).
\]





Although this is an interesting observation, length or norm does not make sense since the formal pair $\langle \cdot ,\cdot \rangle$ of symbols is not linear. Perhaps there may be other interesting properties about Cathelineau's vector space that could easily be observed from the diagrammatics.


\section{Deforming \texorpdfstring{$5$}{5}-term dilogarithm into \texorpdfstring{$4$}{4}-term infinitesimal dilogarithm}
\label{section_deformation}

We refer to \cite[Section 1.3]{Cath88} where Cathelineau derives the dilogarithm to infinitesimal dilogarithm. In this section, we give a completely written out derivation of this deformation. 

Let $A$ be a unital commutative ring.
Let $P_{A}$ be an abelian group generated by $\{ z\}$, where $z$ and $1-z$ are invertible in $A$, with relations 
\begin{enumerate}
\item $\{ z\} + \left\{ \frac{1}{z} \right\} = 0$, 
\item $\left\{ \frac{1}{z} \right\} - \{1-z \}=0$,
\item
\label{item_5_term_dilog} $\{z_1 \}-\{ z_2\} 
+ \left\{ \frac{z_2}{z_1} \right\} 
- \left\{ \frac{1-z_2}{1-z_1} \right\} 
+ \left\{ \frac{(1-z_2)z_1}{(1-z_1)z_2} \right\} = 0$, 
\end{enumerate} 
where the last relation~\eqref{item_5_term_dilog} is called the $5$-term dilogarithm equation. 
A homomorphism of unital rings $A\rightarrow A'$ induces a homomorphism of abelian groups $P_A\rightarrow P_{A'}$. 

We now specialize to the case when $A=\mathbf{k}[\eps]$, where $\eps^2=0$, and $A' = \mathbf{k}$. 
Let $\widetilde{T}P(\mathbf{k}) := \ker( P_{\mathbf{k}[\varepsilon]}\rightarrow P_{\mathbf{k}})$, where the homomorphism $\mathbf{k}[\varepsilon]\lra \mathbf{k}$ is identity on $\mathbf{k}$ and $\varepsilon \mapsto 0$. 
This induces the short exact sequence 
\[
\xymatrix{
0 \ar[rr] & & 
\widetilde{T}P(\mathbf{k}) 
\ar@^{(_->}[rr] & & \ar@/_15pt/[ll]_{\sigma}
P_{\mathbf{k}[\varepsilon]} \ar@{->>}[rr] & & P_{\mathbf{k}} \ar[rr] & & 0
}
\]
with the section $\sigma$ given by 
\begin{equation}
\label{eqn_tangent_generators}
\sigma(\{ a+ b\varepsilon \}) = \{a + b\varepsilon \} - \{ a \},  \qquad  a \not=0,1.
\end{equation}
The kernel $\widetilde{T}P(\mathbf{k})$ is also generated by the images of $\{a + b\varepsilon\}$ under $\sigma$ as in 
\eqref{eqn_tangent_generators} over all $a\in \kk^{*}\setminus\{1\}$ and $b\in \kk$.  
We also have $\mathbf{k}^{*}$ acting on $\widetilde{T}P(\mathbf{k})$ via the translation on the coefficient of $\varepsilon$:
\[
c \: \sigma(\{ a+ b\varepsilon\})
= c(\{ a + b \varepsilon\} - \{ a\} ) 
= \{ a + c b \varepsilon\} - \{ a\}
= \sigma(\{a + c b \varepsilon \}). 
\]
Let $\mu:\widetilde{T}P(\mathbf{k}) \rightarrow \mathbf{k}^+ \wedge_{\mathbb{Z}} \mathbf{k}^+$ be a morphism of abelian groups, where 
$\mu(\{ a+b\varepsilon\} - \{ a\} ) = \frac{b}{a}\wedge \frac{b}{1-a}$, or  alternatively, 
$\mu(\{ a+ 2b\varepsilon \} + \{ a\} - 2\{a+b\varepsilon \} )= 2\left( \frac{b}{a}\wedge \frac{b}{1-a} \right)\not=0$. 
Let $N_{\mathbf{k}}$ be the subgroup of $\widetilde{T}P(\mathbf{k})$ generated by expressions of the form 
\begin{equation}
\label{eqn_reln_for_Nk}
\{ a + (b+b')\varepsilon \} + \{ a \} - \{ a+ b\varepsilon\} - \{ a + b'\varepsilon\}.
\end{equation}
In $\widetilde{T}P(\mathbf{k})$, we can view  $\{ a+ b\varepsilon\} - \{ a \}$ as equal to its deformation. That is, let $TP(\mathbf{k}) := \widetilde{T}P(\mathbf{k})/N_{\mathbf{k}}$. Then in $TP(\mathbf{k})$, we have: 
\[\{ a + (b+b')\varepsilon \} - \{ a + b'\varepsilon\} = \{ a+ b\varepsilon\} - \{ a \}.
\]

\begin{proposition}
\label{prop_5_term_dilog_to_infinitesimal_dilog}
Let $\mathbf{k}$ be a field of characteristic $0$. Then there exists an isomorphism of vector spaces 
\begin{equation}
\label{eqn_isom_tangent_space_derivations}
\varphi:\beta_{\mathbf{k}}\stackrel{\cong}{\lra} TP(\mathbf{k}), 
\quad 
\mbox{ where }
\varphi([a]) = \{ a + a(1-a)\varepsilon\} - \{ a \}.
\end{equation}
\end{proposition}

This leads us to a commutative diagram with two short exact sequences:
\[
\xymatrix@-1pc{
& & 0 \ar[dd]& & & & & & \\ 
& &  & & & & & & \\
& & N_{\mathbf{k}} \ar@^{(_->}[dd] & & & & & & \\
& &  & & & & & & \\
0 \ar[rr] & & 
\widetilde{T}P(\mathbf{k}) 
\ar@^{(_->}[rr] \ar@{->>}[dd]^{\mbox{quotient by } N_{\mathbf{k}}} & & \ar@/_15pt/[ll]_{\sigma}
P_{\mathbf{k}[\varepsilon]} \ar@{->>}[rr] & & P_{\mathbf{k}} \ar[rr] & & 0 \\
& &  & & & & & & \\
\beta_{\mathbf{k}}\ar[rr]^{\cong\ \varphi} & & TP(\mathbf{k}) \ar[dd] & & & & & & \\
& &  & & & & & & \\
& & 0. & & & & & & \\
}
\]

We will now prove Proposition~\ref{prop_5_term_dilog_to_infinitesimal_dilog}.
\begin{proof}
We will prove that $\varphi$ is well-defined by checking compatibility with the relation \eqref{item_beta_k_4_term} in Section~\ref{subsection:vs_Shannon_entropy}. 
Let 
\[
u = a + a (1-a)\varepsilon, 
\qquad 
v = b + b (1-b)\varepsilon. 
\]
Then 
\begin{align*}
\frac{u(1-v)}{v(1-u)}
&= \frac{(a + a (1-a)\varepsilon)(1-(b + b (1-b)\varepsilon))}{(b + b (1-b)\varepsilon)(1-(a + a (1-a)\varepsilon))} \\ 
&= \frac{a - a b + a(1 - b)(1 - a - b) \varepsilon}{b - a b + b(1 - a)(1 - a - b) \varepsilon} \\ 
&= \frac{a \cancel{b} \cancel{(1-a)}(1-b)}{b^{\cancel{2}}(1-a)^{\cancel{2}}} \\ 
&= \frac{a(1-b)}{b(1-a)}.
\end{align*}

Now, 
\begin{align*}
\varphi&\left( [a] - [b] + (1-a) \left[\frac{1-b}{1-a} \right] + a \left[ 
\frac{b}{a} \right]
\right) \\ 
&= \varphi([a]) - \varphi([b]) + (1-a)\varphi\left( \left[ \frac{1-b}{1-a} \right]\right) + a \varphi\left( \left[ \frac{b}{a} \right] \right) \\ 
&= \{ a + a(1-a)\varepsilon \} - \{ a \} -\left( \{ b + b(1-b)\varepsilon \} - \{ b \}
\right) \\
&\hspace{4mm}+ (1-a) \left( 
\left\{ 
\frac{1-b}{1-a} + \frac{1-b}{1-a} \left( 1- \frac{1-b}{1-a} \right) \varepsilon
\right\} - \left\{  \frac{1-b}{1-a} \right\}
\right) 
+ a \left( \left\{ \frac{b}{a} + \frac{b}{a}\left( 1-\frac{b}{a} \right)\varepsilon \right\} - \left\{ \frac{b}{a} \right\}
\right) \\ 
&= \sigma(\{ a+a(1-a)\varepsilon \}) 
- \sigma(\{b+b(1-b)\varepsilon \}) 
+ (1-a) \sigma \left( 
\left\{ 
\frac{1-b}{1-a} + \frac{1-b}{1-a} 
\left( 1- \frac{1-b}{1-a} \right) \varepsilon 
\right\}
\right) \\ 
&\hspace{4mm}+ a \sigma\left( \left\{
\frac{b}{a} + \frac{b}{a} \left( 1-\frac{b}{a} \right) \varepsilon 
\right\} 
\right) \\ 
&= \sigma(\{ u\}) - \sigma(\{ v \}) +  \sigma \left( 
\left\{ 
\frac{1-b}{1-a} + (1-b) 
\left( 1- \frac{1-b}{1-a} \right) \varepsilon 
\right\}
\right) +  \sigma\left( \left\{
\frac{b}{a} + b \left( 1-\frac{b}{a} \right) \varepsilon 
\right\} 
\right) \\ 
&= \sigma(\{ u\}) - \sigma(\{ v \}) +  \sigma \left( 
\left\{ 
\frac{1-b}{1-a} - \frac{1-b}{1-a}(a-b)  \varepsilon 
\right\}
\right) +  \sigma\left( \left\{
\frac{b}{a} + \frac{b}{a} \left( a-b \right) \varepsilon 
\right\} 
\right) \\ 
&\stackrel{\dddag}{=} \sigma(\{ u\}) - \sigma(\{ v\}) 
+  \sigma\left( \left\{ 
\frac{b}{a} + \frac{b}{a}(a - b) \eps
\right\}\right) 
- \sigma\left( \left\{ 
\frac{1 - b}{1-a} + \frac{1 - b}{1-a}(a - b) \eps 
\right\} \right)  \\ 
&\hspace{4mm}
+  \left\{  \frac{a(1-b)}{b(1-a)} \right\}
-  \left\{  \frac{a(1-b)}{b(1-a)} \right\}
\\
&=  \sigma(\{ u\}) - \sigma(\{ v\}) + \sigma\left( \left\{ 
\frac{b}{a} + \frac{b}{a}(a - b) \eps
\right\}\right) - \sigma\left( \left\{ 
\frac{1 - b}{1-a} + \frac{1 - b}{1-a}(a - b) \eps 
\right\} \right) \\
&\hspace{4mm} + \sigma \left( \left\{  \frac{a(1-b)}{b(1-a)} \right\}\right) \\ 
&=  \sigma\left( \{ u\} - \{ v\}
+  \left\{
\frac{b}{a} (1 + (a - b) \eps)
\right\} -  \left\{ 
\frac{(1 - b) + (1 - b) (a - b) \eps}{1-a}
\right\}  +  \left\{ 
\frac{a(1-b)}{b(1-a)} \right\}\right)
\\ 
&= \sigma \left( 
\{ u \} - \{ v \} + \left\{ \frac{v}{u} \right\} - \left\{  \frac{1-v}{1-u} \right\} 
+ \left\{ \frac{u(1-v)}{v(1-u)} \right\} 
\right) 
\end{align*}
since 
\begin{align*}
\frac{v}{u} 
&= \frac{b + b (1 - b) \eps}{a + a (1 - a) \eps} 
= \frac{b + b (1 - b) \eps}{a + a (1 - a) \eps} 
\frac{a - a (1 - a) \eps}{a - a (1 - a) \eps} \\ 
&= \frac{a b + a b (a - b) \eps}{a^2} 
= \frac{b}{a} + \frac{b}{a}(a-b) \eps, 
\end{align*}
\begin{align*}
\frac{1-v}{1-u} 
&= \frac{1 - (b + b (1 - b) \eps)}{1 - (a + a (1 - a) \eps)} 
= \frac{1 - b - b (1 - b) \eps}{1 - a - a (1 - a) \eps} \\
&= \frac{1 - b - b (1 - b) \eps}{1 - a - a (1 - a) \eps} \frac{1 - a + a (1 - a) \eps}{1 - a + a (1 - a) \eps} \\
&= \frac{(1 - a) (1 - b) + (1 - a) (1 - b) (a - b) \eps}{(1-a)^2} \\ 
&= \frac{(1 - b) + (1 - b) (a - b) \eps}{1-a} \\
&= \frac{1 - b}{1-a} + \frac{1 - b}{1-a}(a - b) \eps, 
\end{align*}
\begin{align*}
\frac{u(1-v)}{v(1-u)} 
&= \frac{(a + a (1 - a) \eps) (1 - (b + b (1 - b) \eps))}{(b + b (1 - b) \eps) (1 - (a + a (1 - a) \eps)} \\ 
&= \frac{a (1 - b) \cancel{(1 + (1 - (a + b)) \eps)}}{b (1 - a) \cancel{(1 + (1 - (a + b)) \eps)}}  \\
&= \frac{a(1-b)}{b(1-a)},
\end{align*}
and for $b'=-b$ in \eqref{eqn_reln_for_Nk}, we have 
\begin{align*}
\{a + (b-b)\eps \} &+ \{ a\}-\{ a+b\eps\} -\{a-b\eps \} = \{a \} + \{ a\}-\{ a+b\eps\} -\{a-b\eps \}. 
\end{align*}
So in $TP(\mathbf{k})$, 
\begin{align*}
0 
&= \sigma(\{a \} + \{ a\}-\{ a+b\eps\} -\{a-b\eps \}) \\ 
&= 2\sigma(\{ a\}) - \sigma(\{ a+ b\eps\})
- \sigma(\{ a-b\eps\}) \\ 
&= 2(\{ a\} -\{ a\}) - \sigma(\{ a+ b\eps\})
- \sigma(\{ a-b\eps\}) \\ 
&= 0 - \sigma(\{ a+ b\eps\})
- \sigma(\{ a-b\eps\}).\\ 
\end{align*}
The equality $\dddag$ holds in $TP(\mathbf{k})$ since 
\[
- \sigma(\{ a+ b\eps\}) =  \sigma(\{ a-b\eps\}).
\]
Now to prove that $\varphi$ is injective, 
let 
\begin{equation}
\rho: TP(\mathbf{k}) \lra \mathbf{k}^{*} \otimes \mathbf{k}^+,
\qquad 
\rho(\{ a + b\eps \} - \{ a \}) 
= a \otimes \frac{b}{1-a} +  (1-a) \otimes \frac{b}{a}, 
\end{equation}
be a map of vector spaces over $\mathbf{k}$. Then we have a commutative diagram
\[
\xymatrix@-1pc{
\beta_{\mathbf{k}} \ar@^{(_->}[rrrr]^D \ar[rrdd]_{\varphi} & & & & \mathbf{k}^{*}\otimes \mathbf{k}^+, \\ 
& & & & \\ 
& & TP(\mathbf{k}) \ar@^{(_->}[rruu]_{\rho}& & \\ 
}
\]
where $D$ and $\rho$ are injective. So $\varphi$ is also injective. 
\end{proof}

\begin{remark}
Note that in general, $\rho$ may not be surjective, but $D$ and $\rho$ are injective imply that $TP(\mathbf{k})$ is not too small. 
\end{remark}

{\bf Second proof.} We give a second argument that the $5$-term dilogarithm deforms to the infinitesimal $4$-term dilogarithm. 
Let $\kk$ be a field, and let $\kk^\flat= \kk^{*}\setminus \{ 1\} = \kk \setminus \{0,1 \}$. 

Let $\beta_2(\kk)$ be the quotient of $\Q[\kk^\flat]$ by 
\[
[a] - [b] +\left[\frac{b}{a} \right] + \left[ \frac{1-b^{-1}}{ 1 - a^{-1}} \right] - \left[\frac{1-b}{1-a}\right] = 0,
\]
with $(1-a)(1-b)\left(1-\frac{b}{a}\right) \in \kk^{*}$. 

Let $\kk_2:= \kk[t]/(t^2)$ be the ring of dual numbers. 
If $a\in \kk^\flat$, define 
\[
\langle  a \rangle := a + a(1-a)t = a(1+(1-a)t) \in \kk_2.
\]
Note that $(1+ct)^{-1}=1-ct$ in $\kk_2$ since $t^2=0$. We see that this element $\langle a \rangle$ is invertible in $\kk_2$, with $\langle a\rangle^{-1} = a^{-1}(1-(1-a)t)$.

Define the action of $\kk^*$ on the ring $\kk_2$ of dual numbers as $\lambda \ast (b_0 + b_1 t) = b_0 + \lambda b_1 t$, where $\lambda\in\kk^\ast$. 
That is, $\lambda \ast 1 = 1$ and $\lambda \ast t = \lambda t$.
So $\lambda$ only translates the coefficient of $t$.

\begin{lemma}
\label{lemma:quot_dual_numbers}
We have 
$\dfrac{\langle b \rangle}{\langle a\rangle} = a \ast \left\langle \dfrac{b}{a} \right\rangle$. 
\end{lemma}

\begin{proof}
The left hand side shows 
\begin{align*}
\dfrac{\langle b \rangle}{\langle a\rangle} 
= \dfrac{b(1+(1-b)t)}{a(1+(1-a)t)} 
= \dfrac{b}{a} \left(1 + ((1-b)-(1-a))t \right) 
= \dfrac{b}{a} (1+(a-b)t).
\end{align*}
On the other hand, the right hand side shows 
\begin{align*}
a \ast \left\langle \dfrac{b}{a} \right\rangle 
= a\ast \left(\frac{b}{a} \left(1+ \left(1-\frac{b}{a} \right)t \right)  \right)
= \frac{b}{a} \left(1+ \left(1-\frac{b}{a}\right) a t \right) 
= \frac{b}{a} \left( 1+ \frac{a-b}{a} at \right) 
= \frac{b}{a} (1+(a-b)t). 
\end{align*}
\end{proof}

\begin{lemma}
\label{lemma:diff_quot_dual_numbers}
The relation 
\[
\frac{1-\langle b\rangle}{1-\langle a \rangle} 
= (a-1) \ast \left\langle \dfrac{1-b}{1-a} \right\rangle
\]
holds. 
\end{lemma}

\begin{proof}
On the left hand side, we have 
\begin{align*}
\frac{1-\langle b\rangle}{1-\langle a\rangle}   
&= \dfrac{1 - b(1+(1-b)t)}{1-a(1+(1-a)t)}  \\
&=  \dfrac{1 - b - b(1-b)t}{1 - a - a(1-a)t}\cdot \dfrac{1-a + a(1-a)t}{1-a + a(1-a)t} \\
&=\dfrac{(1-a)(1-b)+( a(1-a)(1-b) -(1-a)b(1-b) )t}{(1-a)^2} \\
&=\dfrac{\cancel{(1-a)}(1-b)+ (a-b)\cancel{(1-a)}(1-b) t }{(1-a)^{\cancel{2}}} \\
&=\dfrac{1-b}{1-a}\left(1+(a-b)t \right)
\end{align*}
while on the right hand side, we have 
\begin{align*}
(a - 1) \ast \left\langle \dfrac{1-b}{1-a} \right\rangle
= (a-1) \ast \frac{1-b}{1-a}\left( 1+ \left( 1-\frac{1-b}{1-a}\right)t \right) 
= \frac{1-b}{1-a}\left( 1+ \left( 1-\frac{1-b}{1-a} \right) (a-1) t \right),
\end{align*}
which are indeed equal.
\end{proof}

\begin{lemma}
\label{lemma:add_mult_inv}
    We have $1-\langle a\rangle^{-1} = (1-a^{-1})(1-t)$. 
\end{lemma}

\begin{proof}
Since 
\begin{align*}
1 - \langle a\rangle^{-1} 
&= 1 - a^{-1}(1+(1-a)t)^{-1}
= 1 - a^{-1}(1-(1-a)t) \\ 
&= 1-a^{-1} + a^{-1}(1-a)t 
 = 1-a^{-1} + (a^{-1} - 1)t \\ 
 &= (1-a^{-1})(1-t), 
\end{align*}
the lemma holds. 
\end{proof}

\begin{lemma}
\label{lemma:add_mult_inv_quot}
The equation 
\[
\frac{1-\langle b \rangle^{-1}}{1-\langle a \rangle^{-1}} 
= \frac{1 - b^{-1}}{1-a^{-1}}
\] 
holds.
\end{lemma}

\begin{proof}
The lemma is immediate since 
\begin{align*}
\frac{1-\langle b \rangle^{-1} }{ 1-\langle a \rangle^{-1} }
= \frac{ (1 - b^{-1})\cancel{(1-t)} }{(1 - a^{-1})\cancel{(1-t)}}
= \frac{1 - b^{-1}}{1 - a^{-1}}. 
\end{align*}
\end{proof}

\begin{lemma}
\label{lemma_scalar_bracket}
We have $[b\ast \langle c\rangle]=b[\langle c\rangle]$ where $b\in \kk$ and $c\in \kk_2$ and  
$(-1) [\langle 1-a\rangle]=-[\langle a \rangle]$ and 
$a [\langle a^{-1}\rangle] = -[\langle a \rangle ]$. 
\end{lemma}

\begin{proof}
This is clear using the construction of $\beta(\kk)$. 
\end{proof}

\begin{lemma}
\label{lemma:5_term_relation}
The $5$-term dilogarithm 
\begin{equation}
\label{eqn:5_term_relation}
[\langle a \rangle] - [\langle b \rangle] 
+ \left[\frac{\langle b \rangle}{\langle a \rangle}\right]
+ \left[\frac{1-\langle b \rangle^{-1}}{1-\langle a \rangle^{-1}}\right] 
- \left[\frac{1-\langle b \rangle }{1-\langle a \rangle } \right] = 0
\end{equation}
holds in the commutative ring $\beta_2(\kk_2)$. 
\end{lemma}

\begin{proof}
This is clear by the definition of the ring $\beta_2(\kk_2)$.
\end{proof}

\begin{proposition}
Taking the limit as $t\rightarrow 0$, the $5$-term dilogarithm in \eqref{eqn:5_term_relation} deforms to the $4$-term infinitesimal dilogarithm
\begin{equation}
[a]-[b]+ a\biggl[ \dfrac{b}{a}\biggr] + (1-a)\biggl[ \dfrac{1-b}{1-a} \biggr]=0, \ \ a\in \kk\setminus\{0,1\}, \ b\in \kk^{*}. 
\end{equation}
\end{proposition}

\begin{proof}
We have 
$\left[\dfrac{\langle b \rangle}{\langle a \rangle}\right] = \left[ a * \left\langle \dfrac{b}{a} \right\rangle \right]= a \left[ \left\langle \dfrac{b}{a} \right\rangle \right]$ by Lemmas~\ref{lemma:quot_dual_numbers} and \ref{lemma_scalar_bracket}. 
We also have 
\[
\left[\frac{1-\langle b \rangle^{-1}}{1-\langle a \rangle^{-1}}\right] = \left[  \frac{(1-b^{-1})\cancel{(1-t)}}{(1-a^{-1})\cancel{(1-t)}} \right] =  \left[  \frac{1-b^{-1}}{1-a^{-1}} \right]
\] by Lemma~\ref{lemma:add_mult_inv_quot}. 
Finally, we have 
\begin{align*}
\left[\frac{1-\langle b \rangle }{1-\langle a \rangle } \right]
&= \left[(a-1)*\left\langle \frac{1-b}{1-a} \right\rangle \right] = (a-1) \left[\left\langle \frac{1-b}{1-a} \right\rangle \right] 
\end{align*}
by Lemmas~\ref{lemma:diff_quot_dual_numbers} and \ref{lemma_scalar_bracket}. Putting these together, we have 
\[
[\langle a \rangle] 
- [\langle b \rangle] 
+ a \left[\left\langle\frac{ b}{ a }\right\rangle\right]
+ \left[\frac{1- b^{-1}}{1-  a^{-1}}\right] 
+ (1-a)\left[\left\langle \frac{1-b}{1-a} \right\rangle \right] =0.
\]
The element $\left[\frac{1- b^{-1}}{1-  a^{-1}}\right]$ is the zero element in 
\[
\ker(\beta(\kk_2)\longrightarrow \beta(\kk))\big/\left([ a + (b+b') t ] + [a] - [ a+ b t] - [ a + b' t]\right).
\]
We thus obtain 
\[
[\langle a \rangle] 
- [\langle b \rangle] 
+ a \left[\left\langle\frac{ b}{ a }\right\rangle\right]
+ (1-a)\left[\left\langle \frac{1-b}{1-a} \right\rangle \right] =0,
\]
which $t$-deforms to 
\[
[a ] 
- [ b ] 
+ a \left[\frac{ b}{ a }\right] 
+ (1-a)\left[ \frac{1-b}{1-a} \right] =0 \mbox{ for } a\not=b \mbox{ and } a\in\kk^{\flat}.
\]
\end{proof}


\section{Future Work}
\label{sec_future}
 It would be interesting to extend the work of Im--Khovanov of diagrammatic interpretation in \cite{IK24_dilogarithms_entropy} to R\'enyi entropy, which is a generalization of Shannon entropy, von Neumann entropy, which is a quantized analogue of Gibbs entropy, and relative entropy, which is also known as Kullback--Leibler divergence. There may be a deeper interpretation if one were to extend the diagrammatics of entropy to their higher-dimensional analogues.

\bibliographystyle{amsalpha} 
\bibliography{quantum_entropy}

\end{document}